\documentclass[pra,twocolumn,amssymb,superscriptaddress]{revtex4}
\usepackage{amsmath, verbatim, graphicx,color,amsthm, ulem, enumerate, array}
\usepackage{epsfig}
\usepackage{bbm}
\usepackage[latin1]{inputenc} 

\newcommand{\bra}[1]{\langle #1|}
\newcommand{\ket}[1]{|#1\rangle}
\newcommand{\ketbra}[1]{| #1\rangle \langle #1|}
\newcommand{\ketbradot}[1]{| #1\rangle \langle \cdot |}

\newcommand{\be}{\begin{equation}}
\newcommand{\ee}{\end{equation}}
\newcommand{\eea}{\end{eqnarray}}
\newcommand{\bea}{\begin{eqnarray}}

\newcommand{\g}[2]{\ensuremath{\gamma_{#1}^{#2}}}

\newcommand{\eins}{\openone}
\newcommand{\comp}[1]{\ensuremath{\overline{#1}}}
\newcommand{\W}{\ensuremath{\mathcal{W}}}
\newcommand{\NN}{\ensuremath{\mathcal{N}}}

\newcommand{\AAA}{\ensuremath{\mathcal{A}}}
\newcommand{\BB}{\ensuremath{\mathcal{B}}}

\newcommand{\SSS}{\ensuremath{\mathcal{S}}}

\newcommand{\TT}{\ensuremath{\mathcal{T}}}

\newcommand{\DD}{\ensuremath{\mathcal{D}}}
\newcommand{\kommentar}[1]{}
\newcommand{\trace}{{\rm Tr}}
\newcommand{\NNN}{\ensuremath{\widetilde{\mathcal{N}}}}

\newcommand\T{\rule{0pt}{3.2ex}}
\newcommand\B{\rule[-1.2ex]{0pt}{0pt}}

\newcommand\Tsmall{\rule{0pt}{2.6ex}}
\newcommand\Bsmall{\rule[-1.2ex]{0pt}{0pt}}

\newcommand{\ie}{i.e.}

\renewcommand{\vr}{\ensuremath{\varrho}}

\newcommand{\forget}[1]{}

\newtheorem{lemma}{Lemma}

\newtheorem{corollary}[lemma]{Corollary}

\begin{document}
\title{Entanglement Witnesses for Graph States: General Theory and Examples}

\author{Bastian Jungnitsch}
\affiliation{Institut f\"{u}r Quantenoptik und Quanteninformation,
\"{O}sterreichische Akademie der Wissenschaften, Technikerstra{\ss}e
21A, A-6020 Innsbruck, Austria\\}
\author{Tobias Moroder}
\affiliation{Institut f\"{u}r Quantenoptik und Quanteninformation,
\"{O}sterreichische Akademie der Wissenschaften, Technikerstra{\ss}e
21A, A-6020 Innsbruck, Austria\\}
\author{Otfried G\"uhne}
\affiliation{Naturwissenschaftlich-Technische Fakult\"at, Universit\"at Siegen, Walter-Flex-Stra{\ss}e 3, D-57068 Siegen, Germany\\}
\affiliation{Institut f\"{u}r Quantenoptik und Quanteninformation,
\"{O}sterreichische Akademie der Wissenschaften, Technikerstra{\ss}e
21A, A-6020 Innsbruck, Austria\\}

\date{\today}
\begin{abstract}
We present a general theory for the construction of witnesses that detect genuine 
multipartite entanglement in graph states. First, we present explicit witnesses for 
all graph states of up to six qubits which are better than all criteria so far.
Therefore, lower fidelities are required in experiments that aim at the preparation 
of graph states. Building on these results, we develop analytical methods to 
construct two different types of entanglement witnesses for general graph states. 
For many classes of states, these operators exhibit white noise tolerances that converge 
to one when increasing the number of particles. We illustrate our approach for states 
such as the linear and the 2D cluster state. Finally, we study an entanglement 
monotone motivated by our approach for graph states.
\end{abstract}

\maketitle

\section{Introduction}
The key role of entanglement is illustrated not only by its usefulness in many 
quantum-informational tasks, such as measurement-based quantum computation \cite{mqc} 
and high-precision metrology \cite{metrology}, but also by its fundamental importance 
for excluding certain models of nature in Bell tests \cite{belltests}. Nowadays, 
experiments have succeeded in the preparation of 14-qubit systems in ion traps 
\cite{blatt} and ten-qubit systems in photonic systems \cite{pan}, so the 
characterization of multipartite entanglement is of high interest. Especially 
in the case of the most interesting kind of entanglement, genuine multipartite 
entanglement, general treatments turn out to be difficult \cite{multient,ghzlimit,huberdicke, multientrev}.

In Ref.~\cite{ourpaper} we proposed an alternative approach to this 
characterization by considering a relaxed version of the problem, leading 
to a criterion for genuine multipartite entanglement. Besides being easily 
implementable as a semidefinite program, it also provides surprisingly strong
analytical entanglement criteria which can then be investigated further and 
generalized. As a first step, 
this has been done for the linear cluster state in Ref.~\cite{ourpaper}.

In this paper, we use this approach to develop a general theory of witnesses for 
graph states. Graph states are a family of multi-qubit states which are of 
eminent importance for tasks like measurement-based quantum computation \cite{mqc}
or quantum error correction \cite{errcorr}. These states have several interesting 
properties, for instance they are relatively robust against decoherence and 
violate certain Bell inequalities maximally \cite{heindiss}. Recently, several experiments
succeeded in preparing graph states of several qubits with photons \cite{pan, graphexp}, 
and also the theory of entanglement detection for such experiments has been
investigated in a number of papers \cite{entdecstab,entdecstabpra}.

The main results of our paper can be grouped into two parts: First, we provide 
entanglement criteria, so-called entanglement witnesses, for all graph states 
up to six qubits. These witnesses are optimal in the framework of Ref.~\cite{ourpaper},
they detect more states than the graph state witnesses known so far and thus require a 
lower fidelity when measured in an experiment. 


\begin{figure}
\includegraphics[width=0.9\columnwidth]{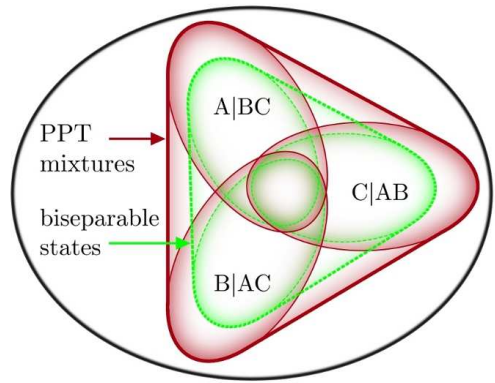}
\caption{\label{fig:PPTsets}In the case of three qubits, there are three convex sets of states that are separable with respect to a fixed bipartition, namely the bipartitions $A|BC$, $B|AC$ and $C|AB$ (green, dashed lines). The set of biseparable states (thick green, dashed line) is given by their convex hull. Each set of states that are separable with respect to a fixed bipartition is contained within the larger set of states that are PPT with respect to this bipartition (red, solid lines). The set of PPT mixtures (thick red, solid line) is then given by the convex hull of these larger sets.}
\end{figure}

Second, we extend our results to arbitrary qubit numbers by providing a general theory 
of how to construct witnesses for arbitrary graph states. In many cases, these witnesses 
improve the best known witnesses so far and have white noise tolerances that approach 
one for an increasing particle number. This implies that for this type of noise the 
state fidelity can decrease exponentially with the number of qubits, but still entanglement 
is present and can be detected. Moreover, this improvement comes with very low 
experimental costs, since it is realized by measuring one additional setting in the 
experiment. Furthermore, a similar improvement can be achieved for witnesses that 
require only two settings to be measured \cite{entdecstab}, which results in improved 
witnesses that consist of only two experimental settings in total. 

The paper is structured as follows. In Sec.~\ref{sec:setting}, we start by presenting 
the structure of entanglement in the multipartite case and introducing the notions 
that will be used later, such as entanglement witnesses and graph states. Then, 
we will briefly recall the criterion of Ref.~\cite{ourpaper} in Sec.~\ref{sec:criterion}.
We will show that it can be reduced to a linear program in the case of graph-diagonal 
states. 

Having laid the foundations, we first consider a certain class of witnesses, 
namely the class of fully decomposable witnesses \cite{ourpaper, lewenstein}. 
This is done in Sec.~\ref{sec:fdecwit}. We provide entanglement witnesses for 
all graph states of up to six qubits in Sec.~\ref{sec:6qubits}. 
Then, in Sec.~\ref{sec:anamethfdecwit}, we present analytical construction 
methods. We provide examples and give an extended construction for particular states in 
Sec.~\ref{sec:extconstr} including further examples. 

In Sec.~\ref{sec:fpptwit}, we move on to another class of witnesses, the fully 
PPT witnesses which are easier to characterize 
\cite{ourpaper}. Here, we do not only provide a construction method for witnesses of 
this class (in Sec.~\ref{sec:anamethfpptwit}), but we can extend it to an even 
larger number of graph states compared with the case of fully decomposable witnesses. 
We present this extension in Sec.~\ref{sec:exconstrfpptwit}. In order to illustrate 
that the presented methods can be exploited further, we supply a witness for the 2D 
cluster state (Sec.~\ref{sec:2dclfpptwit}). 

Finally, we discuss an entanglement monotone for genuine multiparticle entanglement
coming from the approach of Ref.~\cite{ourpaper} and show that graph states are the
maximally entangled states for this entanglement measure. In the conclusion, we disscuss 
our results and possible extensions for the future. In order to make this paper as 
readable as possible, we provide nearly all proofs in the Appendix. 

\section{Setting the stage}
\label{sec:setting}

\subsection{Multipartite entanglement}

First, we discuss the structure of the set of entangled states of multipartite systems, \ie, for systems of more than two particles \cite{multientrev}. For the sake of an easy illustration (cf. Fig. \ref{fig:PPTsets}), we consider the case of three particles here. Nevertheless, the generalization to a higher number of particles is straightforward.

A three-qubit state $\rho$ is {\it separable} with respect to some bipartition, say, $A|BC$, if it can be written in the form
\be
\label{eq:sep}
\vr = \sum_k q_k \ketbra{\phi^{k}_{A}} \otimes \ketbra{\psi^{k}_{BC}} \:.
\ee
Here, the coefficients $q_k$ form a probability distribution, \ie, they are positive and sum up to one. Let us denote states of this type by $\vr_{A|BC}^{\rm sep}$. Analogously, we define the sets of states which are separable with respect to other bipartitions and denote them by $\vr_{B|AC}^{\rm sep}$ and $\vr_{C|AB}^{\rm sep}$. In Fig. \ref{fig:PPTsets}, these three sets are drawn with a dashed, green border.

A state is called {\it biseparable}, if it can be written as a convex sum of states each of which is separable with respect to some bipartition. That is, any biseparable state $\vr^{\rm bs}$ can be written in the form
\be
\label{eq:bisep}
\vr^{\rm bs} = p_1 \vr_{A|BC}^{\rm sep} + p_2 \vr_{B|AC}^{\rm sep} + p_3 \vr_{C|AB}^{\rm sep}\:,
\ee
where, the $p_k$ form a probability distribution. Thus, the set of biseparable states is given by the convex hull of states that are separable with respect to some bipartition. In Fig. \ref{fig:PPTsets}, we show this convex hull with a dashed, thick green border. Any state that is not biseparable is called {\it genuinely multipartite entangled}.

Genuine multipartite entanglement is the strongest kind of entanglement, since biseparable states can be created by entangling, say, only two of three particles and then, to create a statistical mixture, forgetting to which pairs this operation was applied. In fact, in order to detect genuine multipartite entanglement, it is not enough to apply a bipartite criterion to every possible bipartition. Instead, in order to prove that a multipartite state is entangled, one has to show that it cannot be written in the form of Eq.~(\ref{eq:bisep}). 

There is, however, no efficient way to search through all possible decomposition of this form. Thus, it was the idea of Ref.~\cite{ourpaper} to relax the condition of being biseparable.

More precisely, for each fixed bipartition, we consider a superset of the set of separable states which can be characterized more easily than the set of separable states itself. There are different possible choices for supersets. However, in this paper, as a superset of the states that are separable with respect to, say, partition $A|BC$, we select the set of states that have a positive partial transposition (PPT) with respect to partition $A|BC$. A state $\vr=\sum_{ijkl} \vr_{ij,kl} \ket{i}\bra{j}\otimes \ket{k}\bra{l}$ is said to be a {\it PPT state} (with respect to $A|BC$), if its partial transposition 
\be
\label{eq:PT}
\vr^{T_A}=\sum_{ijkl} \vr_{ji,kl} \ket{i}\bra{j}\otimes \ket{k}\bra{l}
\ee
has no negative eigenvalues. We denote a state of this type by $\vr_{A|BC}^{\rm ppt}$ (and analogously for the other bipartitions). We refer to a state in the convex hull of these sets of PPT states as a {\it PPT mixture}. The set of PPT mixtures is therefore the set of states that can be written in the form
\be
\label{eq:PPTmix}
\vr^{\rm pmix} = p_1 \vr_{A|BC}^{\rm ppt} + p_2 \vr_{B|AC}^{\rm ppt} + p_3 \vr_{C|AB}^{\rm ppt}\:,
\ee
where, again, the $p_k$ form a probability distribution. In Fig. \ref{fig:PPTsets}, this set is shown with a solid, thick red border.

Since every separable state is necessarily PPT \cite{PPTcriterion, horodeckis}, every biseparable state is a PPT mixture. Therefore, showing that a state is no PPT mixture implies that it is not biseparable and therefore genuinely multipartite entangled. For some prominent states affected by white noise, such as the three- and the four-qubit GHZ state, the three-qubit W state and the four-qubit linear cluster state, being no PPT mixture happens to be necessary and sufficient for entanglement \cite{ourpaper}. This is also true for the case that the PPT states $\vr_{A|BC}^{\rm ppt}$, $\vr_{B|AC}^{\rm ppt}$ and $\vr_{C|AB}^{\rm ppt}$ live on a subspace of dimension $2 \otimes 2$ or $2 \otimes 3$. However, since there exist states that are PPT with respect to every bipartition and therefore of the form given in Eq.~(\ref{eq:PPTmix}), but are nevertheless genuinely multipartite entangled \cite{symmstates}, not every entangled state can be detected in this way.

By considering the set of PPT mixtures, we exploit that it can be characterized more easily than the set of biseparable states. Numerically, this characterization allows for the use of linear semidefinite programming (SDP) \cite{sdp} --- a standard problem of constrained convex optimization theory. As we will see later, for an important class of states, namely so-called graph-diagonal states, this characterization can be cast into the form of a linear program (LP) which is an even simpler program.

\subsection{Entanglement witnesses}
A useful tool that is often employed in experiments to show that a state is entangled are {\it entanglement witnesses}. A witness for genuine multipartite entanglement is an observable $W$ that has a non-negative expectation value on all biseparable states, but a negative expectation value on at least one entangled state. Therefore, measuring a negative expectation value for $W$ in an experiment proves the presence of entanglement.

Every entangled state is detected by at least one witness \cite{horodeckis}. Therefore, the question whether, for a given state $\vr$, there exists a witness that detects it is equivalent to the question whether $\vr$ is entangled. 

Let us now consider a certain subclass of witnesses that is central to our approach. In the case of two particles, A and B, a {\it decomposable witness} is defined as a witness $W$ that can be written as
\be
\label{eq:decwit}
W = P + Q^{T_A} \:,
\ee
where $P$ and $Q$ have no negative eigenvalues, \ie, are positive semidefinite, denoted by $P \geq 0$, $Q \geq 0$ \cite{lewenstein}. Furthermore, $T_A$ is the partial transposition with respect to $A$ as defined by Eq.~(\ref{eq:PT}).

It can be easily seen that observables of the form given by Eq.~(\ref{eq:decwit}) are positive (or zero) on all separable states $\vr^{\rm sep}$. Any separable state $\vr^{\rm sep}$ has a positive partial transpose and therefore 
\begin{align}
\label{eq:seppos}
\trace(\vr^{\rm sep}W)= \: &\trace(\vr^{\rm sep}P)+ \trace(\vr^{\rm sep}Q^{T_A})\\
= \: & \trace(\vr^{\rm sep}P)+ \trace[(\vr^{\rm sep})^{T_A}Q] \geq 0.
\end{align}

We can now generalize this definition to multipartite systems. A witness $W$ is called {\it fully decomposable}, if, for every strict subset $M$ of the set of all particles $\lbrace A, B, C, D, \dots \rbrace$, $W$ is decomposable with respect to the bipartition given by $M$ and its complement $\comp{M}$. In other words, there exist positive semidefinite operators $P_M$, $Q_M$, such that
\be
\label{eq:fullydecwit}
{\rm for \: all} \: M \subset \lbrace A, B, C, D, \dots \rbrace: \: W = P_M + Q_M^{T_M} \: .
\ee
For example, in the case of three qubits, a fully decomposable witness can be written in three ways,  
\be
\label{eq:3decwit}
W =\: P_A +  Q_A  ^{T_A} = \: P_B +  Q_B  ^{T_B} = \: P_C +  Q_C  ^{T_C} \:,
\ee
where all operators $P_M$ and $Q_M$ are positive semidefinite. Note that, e.g., the existence of the two positive operators $P_A$ and $Q_A$ implies the existence of two positive operators $P_{BC}$ and $Q_{BC}$. One simply has to set $P_{BC} = P_A$ and $Q_{BC} = Q_A^T$. Since $Q_M \geq 0 \Leftrightarrow Q_M^T \geq 0$, the two operators $P_{BC}$ and $Q_{BC}$ defined in this way are positive. Due to $Q_M^{T_M} = (Q_M^T)^{T_{\comp{M}}}$, they also obey Eq.~(\ref{eq:fullydecwit}).

Now, let us make the connection between the notions of PPT mixtures and fully decomposable witnesses by citing the following lemma of Ref. \cite{ourpaper}, which is based on \cite[Theorem 3]{lewenstein}. For the sake of completeness, we present it again here.
\begin{lemma}
\label{lem:1} 
$\varrho$ is a PPT mixture if and only if every fully decomposable witness $W$ is non-negative on $\varrho$.
\end{lemma}
\begin{proof}
``If'': Let us show that if a state is no PPT mixture, there is a fully decomposable witness that detects it. We note that the set of PPT mixtures is convex and compact. Therefore, for any state outside of it, there exists a witness that detects it and is positive on the set of PPT mixtures. Moreover, an operator which is positive on all states that are PPT with respect to a fixed (but arbitrary) bipartition is decomposable with respect to this fixed (but arbitrary) bipartition \cite{lewenstein}. Thus, $W = P_M + Q_M^{T_M}$ for any $M$.\\
``Only if'': A reasoning as in Eq.~(\ref{eq:seppos}) shows that fully decomposable witnesses are non-negative on any state that is PPT with respect to some bipartition. Therefore, these witnesses are also non-negative on all PPT mixtures.
\end{proof}

In the language of constrained optimization theory, the search for a fully decomposable witness with negative expectation value on $\vr$ is the dual problem to the search for a decomposition into PPT states as in Eq.~(\ref{eq:PPTmix}).

In the following, we will often use the subclass of the set of fully decomposable witnesses that we obtain when we require that ${\rm for \: all} \: M \subset \lbrace A, B, C, D, \dots \rbrace : \: P_M = 0$. We call them {\it fully PPT witnesses} and they are defined by
\be
\label{eq:fullyPPTwit}
{\rm for \: all} \: M \subset \lbrace A, B, C, D, \dots \rbrace : \: W^{T_M} \geq 0 \: .
\ee
Fully PPT witnesses are easier to characterize analytically than fully decomposable witnesses. This is due to the fact that, in order to show that an operator $W$ is a fully decomposable witness, one has to find a positive operator $P_M$ for every $M$ and prove the positivity of the corresponding $Q_M$ obtained from Eq.~\ref{eq:fullydecwit}. For fully PPT witnesses, it suffices to show that the partial transpose of $W$ with respect to any possible bipartition is positive.

Before we state the criterion for genuine multipartite entanglement of Ref. \cite{ourpaper}, we finish our introduction to the terms used in this paper by recalling some facts about the widely-used class of graph states.

\subsection{Graph states}
\label{sec:graph_states}
Graph states are defined by mathematical graphs in the following way \cite{heindiss}. 
Given a graph $G = (V,E)$ that is defined by a set $V$ of vertices which correspond to qubits and a set $E$ of edges that connect some of these vertices (cf. the examples in Table~\ref{tab:graphstates}). We denote the number of vertices by $n$.

Then, one can define a set of $n$ operators
\be
\label{eq:generators}
g_i = X_i \prod_{k \in \NN(i)} Z_k, \: i = 1, \dots ,n \:,
\ee
where $\NN(i)$ is the {\it neighborhood} of qubit $i$, \ie, the set of all qubits that are connected to qubit $i$ by an edge. Furthermore, $X_i$ and $Z_i$ are the Pauli operators $\sigma_x$ and $\sigma_y$, respectively, that act on qubit $i$.

\begin{table}[ht!!]
\centering
\begin{tabular}{c c c}
No. 1 --- Bell state & No. 2 --- $\rm GHZ_3$ & No. 3 --- $\rm GHZ_4$ \\
\parbox[c]{0.25\columnwidth}{\includegraphics[width=0.20\columnwidth]{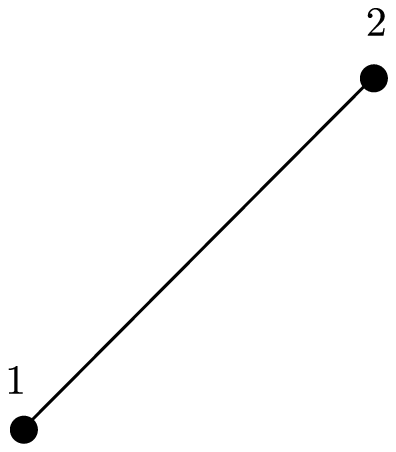}} & \parbox[c]{0.25\columnwidth}{\includegraphics[width=0.20\columnwidth]{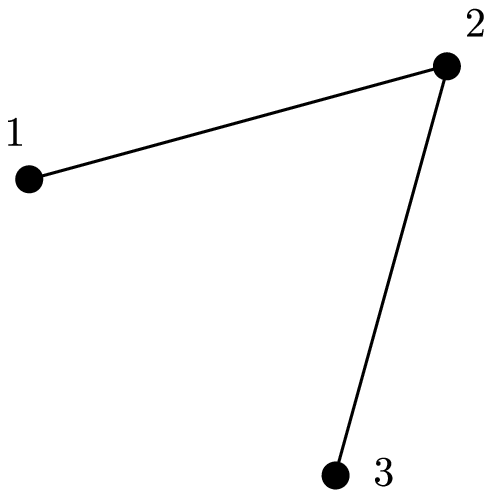}} &\parbox[c]{0.25\columnwidth}{\includegraphics[width=0.20\columnwidth]{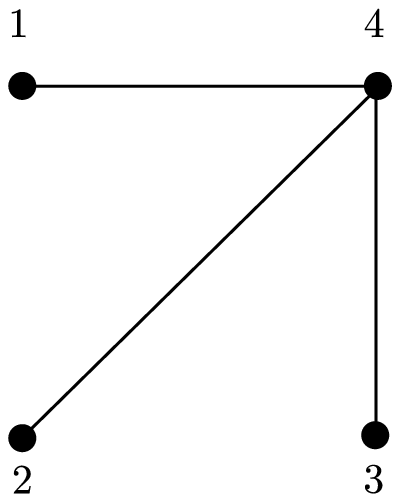}}\\
\T \B  No. 4 --- $\rm Cl_4$ & No. 5 --- $\rm GHZ_5$ & No. 6 --- $\rm Y_5$\\
\parbox[c]{0.25\columnwidth}{\includegraphics[width=0.20\columnwidth]{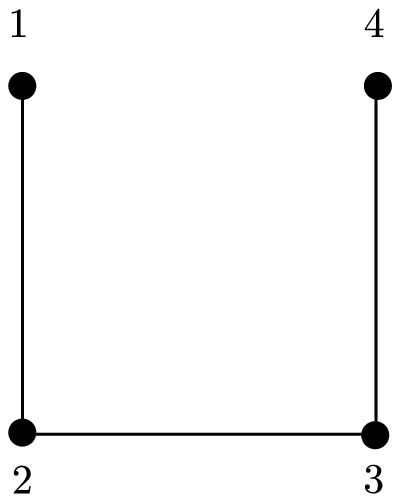}} & \parbox[c]{0.25\columnwidth}{\includegraphics[width=0.25\columnwidth]{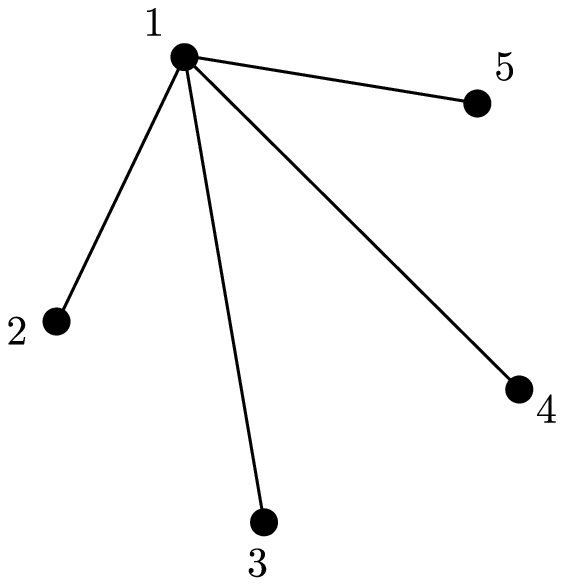}} &\parbox[c]{0.25\columnwidth}{\includegraphics[width=0.25\columnwidth]{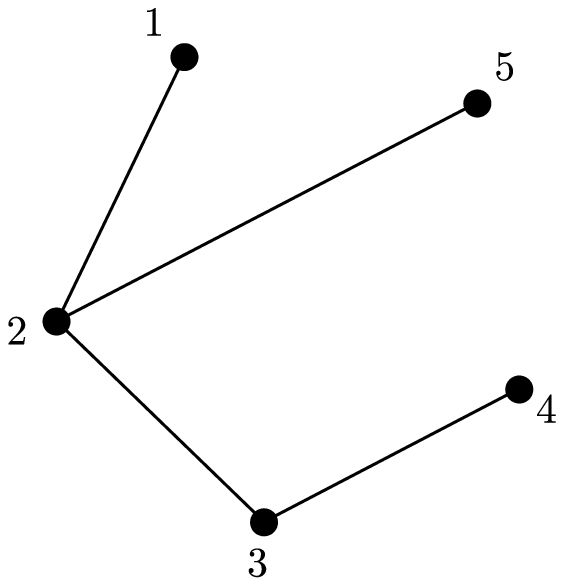}}\\
\T \B  No. 7 --- $\rm Cl_5$ & No. 8 --- $\rm R_5 $ & No. 9 --- $\rm GHZ_6$ \\
\parbox[c]{0.25\columnwidth}{\includegraphics[width=0.25\columnwidth]{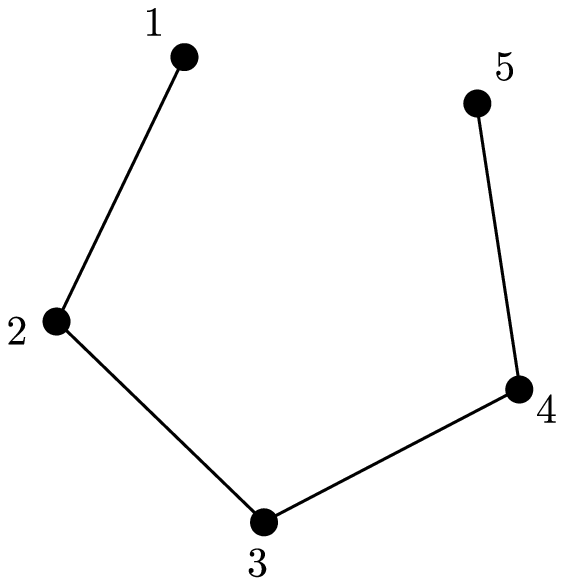}} & \parbox[c]{0.25\columnwidth}{\includegraphics[width=0.25\columnwidth]{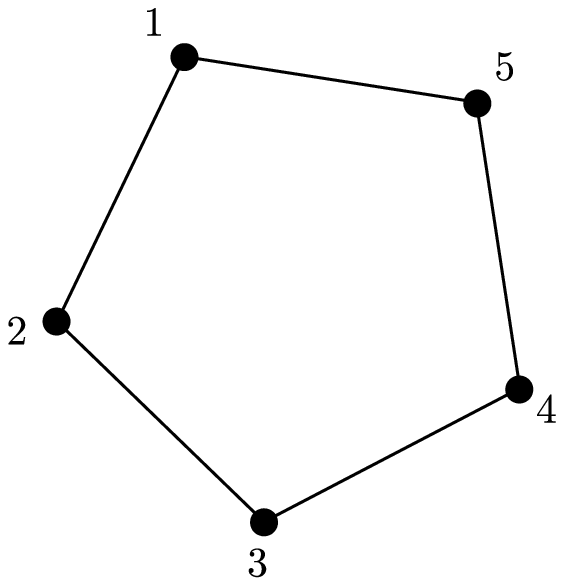}} &\parbox[c]{0.25\columnwidth}{\includegraphics[width=0.25\columnwidth]{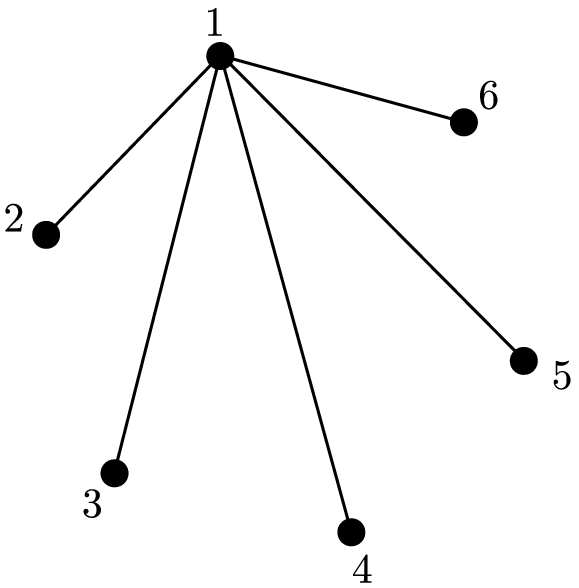}}\\
\T \B  No. 10 & No. 11 --- $\rm H_6$& No. 12 --- $\rm Y_6$\\
\parbox[c]{0.25\columnwidth}{\includegraphics[width=0.25\columnwidth]{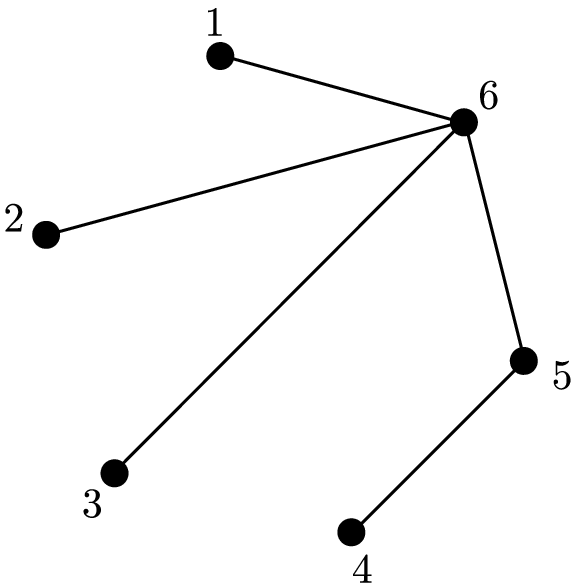}} & \parbox[c]{0.25\columnwidth}{\includegraphics[width=0.25\columnwidth]{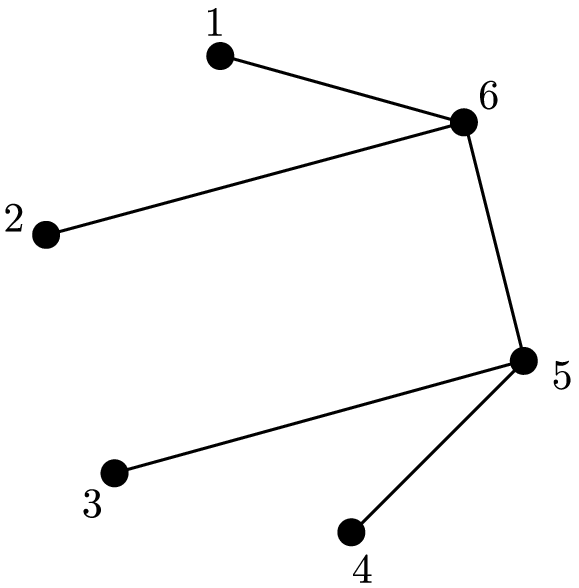}} &\parbox[c]{0.25\columnwidth}{\includegraphics[width=0.25\columnwidth]{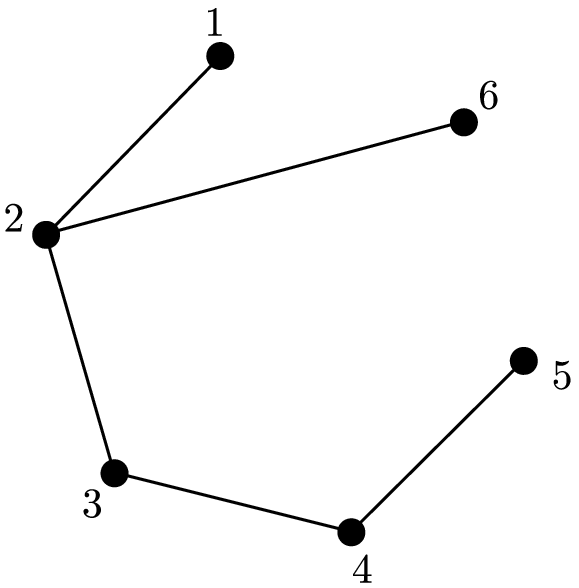}}\\
\T \B  No. 13 --- $\rm E_6$ & No. 14 --- $\rm Cl_6$& No. 15 \\
\parbox[c]{0.25\columnwidth}{\includegraphics[width=0.25\columnwidth]{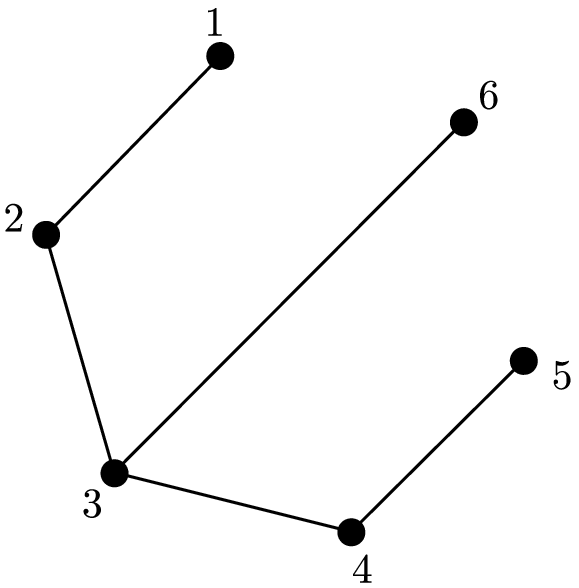}} & \parbox[c]{0.25\columnwidth}{\includegraphics[width=0.25\columnwidth]{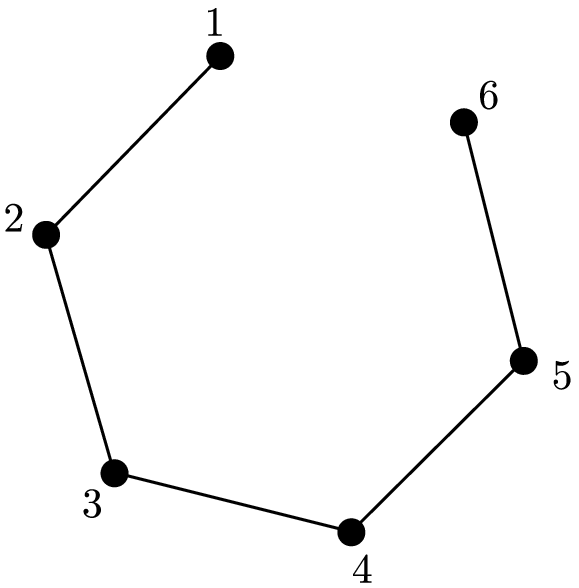}} &\parbox[c]{0.25\columnwidth}{\includegraphics[width=0.25\columnwidth]{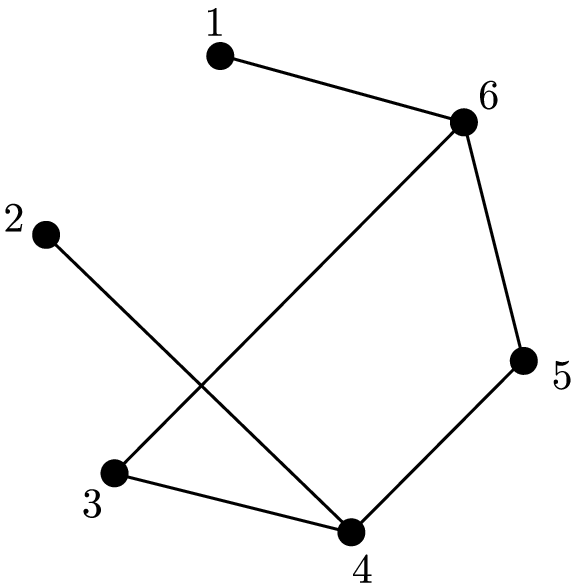}}\\
\T \B  No. 16 & No. 17 & No. 18 --- $\rm R_6$ \\
\parbox[c]{0.25\columnwidth}{\includegraphics[width=0.25\columnwidth]{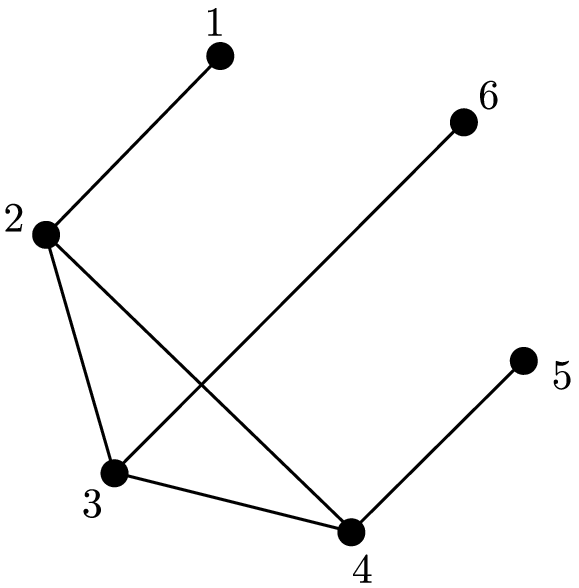}} & \parbox[c]{0.25\columnwidth}{\includegraphics[width=0.25\columnwidth]{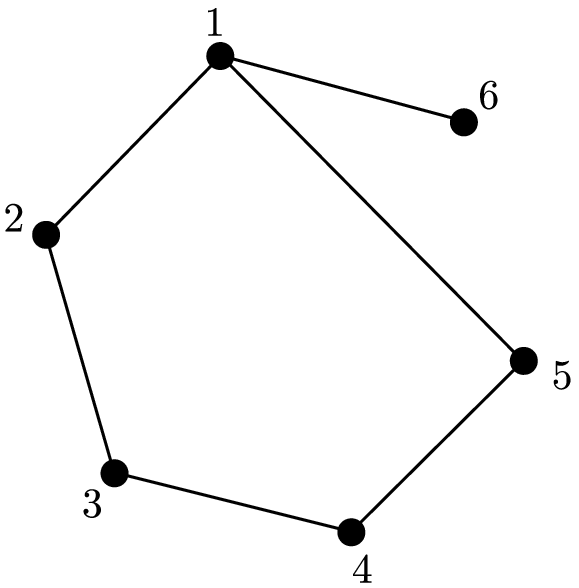}} &\parbox[c]{0.25\columnwidth}{\includegraphics[width=0.25\columnwidth]{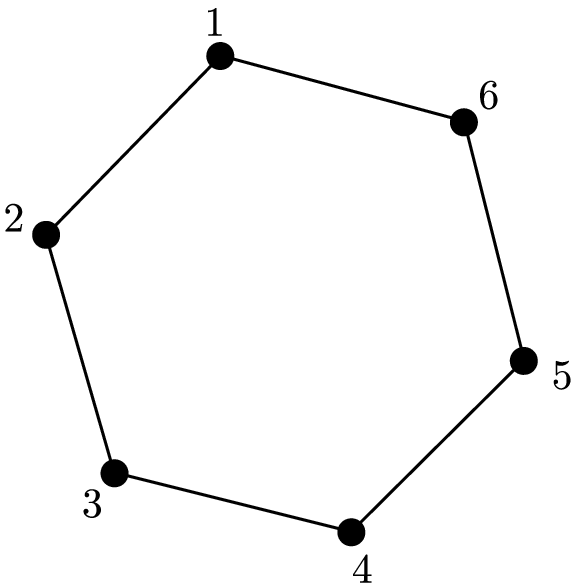}}\\
\T \B  No. 19 & &  \\
\parbox[c]{0.25\columnwidth}{\includegraphics[width=0.25\columnwidth]{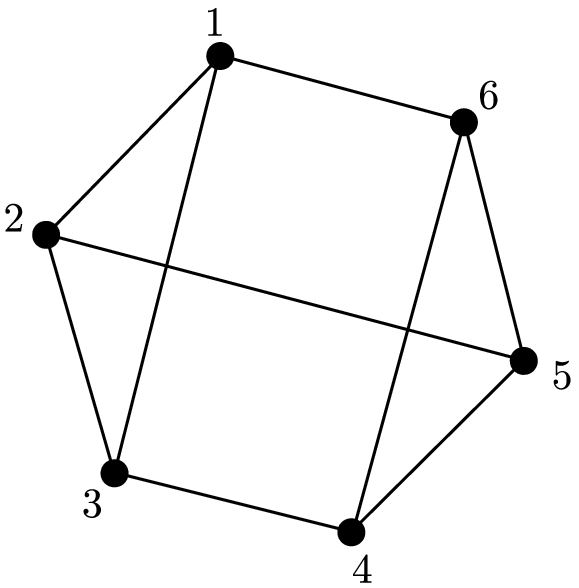}} &  &\\
 \end{tabular}
\caption{\label{tab:graphstates}The graph states of up to six qubits can be grouped into 19 LU equivalence classes. For each class, we show the representative state here.}
\end{table}

The operators $g_i$ commute and generate a set $\SSS$ of so-called {\it stabilizer operators} which consists of $2^n$ elements, \ie,
\be
\label{eq:stabgroupdef}
\SSS= \lbrace S_1 , \dots , S_{2^n} \rbrace = \left\lbrace \prod_{i = 1}^{n} g_i^{x_i} \vert \vec{x} \in \lbrace 0,1 \rbrace^n \right\rbrace \:.
\ee
This means that every operator $S_i \in \SSS$ can be written as a product of some generators $g_i$, in which every generators appears once or not at all. Note that due to $g_i g_i = \eins, \: i = 1, \dots, n$, also, e.g., the product of $g_i$ with itself is included in the definition of Eq.~(\ref{eq:stabgroupdef}). In particular, the identity operator is contained in $\SSS$.

To every graph $G$ we can then associate a {\it graph state $\ket{G}$} that is uniquely defined by
\be
\label{eq:graphstatedef}
g_i \ket{G} = \ket{G}, \; \forall \: i = 1, \dots ,n \:.
\ee
Thus, $\ket{G}$ is the unique state that is an eigenstate to eigenvalue $+1$ of all generators $g_i$. Moreover, every graph also defines a so-called {\it graph state basis}, whose elements are denoted by $\ket{a_1 \dots a_n}_G, \: a_i \in \lbrace 0,1 \rbrace$ and which are defined by
\be
g_i \ket{a_1 \dots a_n}_G = (-1)^{a_i} \ket{a_1 \dots a_n}_G, \; \forall \: i = 1, \dots ,n \:.
\ee
Consequently, $\ket{G} = \ket{0 \dots 0}_G$. Moreover, projectors on these vectors can be written as
\be
\label{eq:projectorform}
_G \ketbra{a_1 \dots a_n}_G = \prod \limits_{i = 1}^{n} \frac{(-1)^{a_i} g_i+\eins}{2} \:.
\ee
In the following, we will refer to states that are diagonal in a graph state basis as {\it graph-diagonal states}.

Note that two graph states that belong to two different mathematical graphs can still be physically equivalent, i.e., equivalent under local unitary transformations (LU-equivalent) and permutations of qubits \cite{hein, vandennest}. For example, this is the case for the star graph No.~9 of Table~\ref{tab:graphstates} and the fully connected graph, in which each of the six vertices is connected with every other vertex. Both graphs describe a state which is LU-equivalent to the GHZ state of six qubits.

It has been shown that, when taking into account states of up to six qubits, there are 19 LU-equivalence classes of connected graph states \cite{hein}. Note that the equivalence classes of up to eight qubits have been characterized in Ref.~\cite{cabellographstates}. Table~\ref{tab:graphstates} shows one representative state of each LU-equivalence class. Any graph state of six or less qubits can therefore be mapped by local unitaries and permutations onto the state associated to some graph in Table~\ref{tab:graphstates}. The local unitaries that one has to apply for this mapping are given in Ref.~\cite{hein}. In this way, one can also transform a witness for any state in a particular LU-equivalence class into a witness of any other state in the same class.

In order to improve readability, we will drop the subscript $G$ in the following and write $\ket{a_1 \dots a_n}_G = \ket{\vec{a}}$. Nevertheless, it is important to keep in mind that all partial transpositions $T_M$ are to be understood w.r.t the computational basis. This ensures that the only Pauli matrix that is changed under transposition is $Y$, whose transpose is $-Y$.

As before, we will refer to $M\vert \comp{M}$, where $M$ is a subset of the set of all qubits and $\comp{M}$ its complement, as a {\it bipartition} of our system.

Note that, for any graph state $\ket{G}$, the operator $W_{\rm proj} = \frac{1}{2} \eins - \ketbra{G}$ is a witness \cite{entdecstab}. We will refer to $W_{\rm proj}$ as the {\it projector witness} of $\ket{G}$ \cite{bourennane}.

\section{The Entanglement Criterion}
\label{sec:criterion}
In this section, we recall the criterion for genuine multipartite entanglement originally introduced in Ref.~\cite{ourpaper}. We then specialize it to graph-diagonal states in Lemma~\ref{lem:lp} which is our main result in this section.

Lemma~\ref{lem:1} naturally leads to an entanglement criterion which asks whether a given state is detected by a fully decomposable witness or not. Given a multipartite state $\vr$, we consider the optimization problem
\begin{align}
\label{eq:sdp1}
&\min \:\trace  (W \varrho)\\
&\begin{aligned}
{\mbox{s.t.}} \: &\trace(W) = 1 \: \mbox{and for all} \: M : \nonumber \\
&W = P_M + Q_M^{T_M}, Q_M \geq 0 ,\:P_M \geq 0 \nonumber \:.
\end{aligned}
\end{align}
In this minimization, the free parameters are given by $W$ and an operator $P_M$ for every strict subset $M$ of the set of all qubits. In practice, it is only necessary to ensure the existence of positive operators $P_M$ and $Q_M$ for $2^{n-1}-1$ partitions, since the two partitions $M\vert \comp{M}$ and $\comp{M} \vert M$ are equivalent, as argued before in the three-qubit case [cf. Eq.~(\ref{eq:3decwit})].

A negative minimum in Eq.~(\ref{eq:sdp1}) indicates that $\varrho$ is detected by the fully decomposable witness $W$ for which the minimum is obtained. Thus, it is entangled and, in particular, no PPT mixture. 

As mentioned before, this minimization can be performed numerically by an SDP. Variations of the program in Eq.~(\ref{eq:sdp1}) have been discussed in Ref. \cite{ourpaper}. There, it was applied to some important states as the W and the GHZ state for three and four qubits, and, for four qubits, the linear cluster state, the singlet and the Dicke state of two excitations. Its white noise tolerance turned out to be higher than in previous criteria. Moreover, the case in which no fully tomography, but only a restricted set of observables has been measured, was considered. In the next subsection, we will show that in the case of graph-diagonal states, the program reduces to an LP.

\subsection{Graph-diagonal states}
\label{sec:graphdiag}
If we are only interested in graph-diagonal states, the corresponding search for an optimal fully decomposable entanglement witness can w.l.o.g. be restricted to graph-diagonal witnesses, for which also the operators $P_M$ and $Q_M$ are graph-diagonal. This is summarized in the following lemma.
\begin{lemma}
\label{lem:lp}
For any graph diagonal state $\varrho_G=\sum_{\vec{k}} s_{\vec{k}} \ketbra{\vec{k}}$, the search for an optimal fully decomposable entanglement witness given by Eq.~(\ref{eq:sdp1}), can w.l.o.g. be restricted to a graph-diagonal form, \ie, to a linear program given by
\begin{align}
\label{eq:lp1}
&\min\: \trace(W_G \varrho_G)\\
&\begin{aligned} {\rm s.t.}\:& W_G =\sum w_{\vec{k}} \ketbra{\vec{k}}, \trace(W_G) = 1\: {\rm and \: for \: all} \: M :\\
& W_G = P_M + Q_M^{T_M}, P_M \geq 0, Q_M \geq 0, \\ 
& P_M =\sum p^M_{\vec{k}} \ketbra{\vec{k}}, Q_M = \sum q^M_{\vec{k}} \ketbra{\vec{k}} \: .
\end{aligned}
\end{align}
\end{lemma}
The proof is given in Sec.~\ref{sec:prooflp} of the Appendix.

This lemma has the following important implications: First, the optimization problems simplifies to a linear program, which are in general easier to solve than general semidefinite programs. Second, it provides a great simplification in order to derive analytic witnesses, because we know that there is an optimal witness which is diagonal in the graph state basis. Also, checking positivity of any operator simplifies to verifying non-negativity within the graph state basis. Instead of testing positivity of a whole matrix, it is enough to consider products of generators $g_i$ and sums thereof [cf. Eq.~(\ref{eq:projectorform})]. Third, let us point out that this lemma also implies that, if a state is a PPT mixture, each PPT state in its decomposition can be assumed to be graph-diagonal as well. Finally, note that a similar statement as Lemma \ref{lem:lp} holds for PPT witnesses as well.

\section{Fully decomposable witnesses}
\label{sec:fdecwit}

In this section, we present a general theory for fully decomposable witnesses of graph states. First, in Sec.~\ref{sec:6qubits}, we provide fully decomposable witnesses for all LU-equivalence classes of graph states up to six qubits. These witnesses are obtained by the criterion of Eq.~(\ref{eq:sdp1}). The graph states are given in Table~\ref{tab:graphstates}, while the witnesses' white noise tolerances are given in Table~\ref{tab:tolerances}. The witnesses can be found in Appendix~\ref{sec:graphstatewit}. 

Moreover, we introduce an analytical construction method for fully decomposable witnesses of general graph states in Sec.~\ref{sec:anamethfdecwit}. This construction method is a generalization of the linear cluster state witnesses in Ref.~\cite{ourpaper} and is, as one of this section's main results, formulated in Lemma~\ref{lem:witarbgraphstates}.

We provide specific examples in Sec.~\ref{sec:decex1}. Finally, we show how to construct witnesses that detect even more states using the witnesses of Lemma~\ref{lem:witarbgraphstates}. This result is given as Lemma~\ref{lem:witcomb} in Sec.~\ref{sec:extconstr}. Again, we give examples in Sec.~\ref{sec:decex2}.

\subsection{Graph states up to 6 qubits}
\label{sec:6qubits}
We now apply the criterion of Eq.~(\ref{eq:sdp1}) to certain graph states. To this end, we implemented it as a semidefinite program using the parser YALMIP \cite{yalmip} in combination with the solver modules SeDuMi \cite{sedumi} or SDPT3 \cite{sdpt3} in MATLAB. The program we wrote is called PPTMixer and can be found online \cite{matlabcentral}.

As mentioned before, there are 19 LU-equivalence classes of connected graph states up to six qubits. We apply our criterion to one state of each class (cf. Table~\ref{tab:graphstates}), obtaining the witnesses given in Appendix~\ref{sec:graphstatewit}. By applying the rules in Ref.~\cite{hein}, it is possible to transform these into witnesses for any graph state of up to six qubits.

Let us have a closer look at the witnesses of Appendix~\ref{sec:graphstatewit}. A widely-used indicator for how robust a witness is to noise in an experiment is the so-called {\it white noise tolerance}. It is defined in the following way: For a given state $\varrho$ and a given witness $W$, the white noise tolerance is the maximal amount $p_{\rm tol}$ of white noise, such that the state $\vr(p_{\rm tol}) = (1-p_{\rm tol})\vr +{p_{\rm tol}} \eins/ {2^n}$ is still detected by the witness $W$. Note that the criterion of Eq.~(\ref{eq:sdp1}) provides witnesses with the highest possible white noise tolerance among all fully decomposable witnesses. This can be seen by noting that both $\trace(W \ketbra{G})$ and $\trace[W \vr(p_{\rm tol})]$ reach their minimum for the same normalized witness $W$, since $\trace(W p_{\rm tol} \eins /{2^n}) = p_{\rm tol} /{2^n}$ is independent of $W$. Thus, the witness that one obtains for the state $\ket{G}$ is also a witness for $\vr(p_{\rm tol})$. In Table~\ref{tab:tolerances}, we give the witness tolerances of these witnesses.

Now, let us present some of these witnesses as examples. Note that the SDP yields witnesses whose trace is normalized to one. In order to make the structure of the witnesses more evident, we renormalized them for each state $\ket{G}$, such that $\bra{G}W \ket{G} = -1/2$.

For the GHZ states of three to six qubits (cf. states No.~2, No.~3, No.~5 and No.~9 in Table~\ref{tab:graphstates}), we obtain the well-known projector witnesses $W_{\rm proj} = \frac{1}{2} \eins - \ketbra{G}$. Since it is known that \mbox{$(1-p)\ketbra{GHZ_n} + p \eins/ 2^n$} is biseparable for \mbox{$p \geq \left[ 2 \left(1-2^{-n} \right) \right]^{-1}$} \cite{ghzlimit}, these witnesses have the maximal possible white noise tolerance.

The linear cluster state of four qubits $\ket{Cl_4}$, labelled as state No.~4, is detected by the witness
\be
\label{eq:cl4witness}
W_{\rm Cl4} = \frac{\eins}{2} - \ketbra{G} - \frac{1}{2} \g{1}{-} \g{4}{-} \:,
\ee
where we defined $\g{i}{\pm} = \left(\eins \pm g_i\right)/2$, $g_i$ being the generators of the stabilizer group of $\ket{\rm Cl_4}$, for the sake of a compact notation. Note that, alternatively, one can write $\g{1}{-} \g{4}{-} = \left(\eins - g_1\right)/2 \left(\eins - g_4\right)/2 = \sum_{i,j \in \lbrace 0,1\rbrace} \ketbra{1 i j 1}$ in the graph state basis.
We will gain a deeper understanding of the structure of this witness in the next section.

Strikingly, the similar state No.~6, which we call $Y_5$ state, is detected by a similar witness which has, however, some additional terms. The witness is given by
\be
\label{eq:WG6}
W_{G6} = \frac{\eins}{2} - \ketbra{G} - \frac{1}{2}\g{1}{-} \g{4}{-} - \frac{1}{2} \g{1}{+} \g{4}{-} \g{5}{-}\: .
\ee
For the symmetrized version of this state, state No.~11 (or $H_6$ state), we obtain a witness with even more terms, namely
\begin{align}
\label{eq:WG11}
W_{G11} = \frac{\eins}{2} - \ketbra{G} & - \frac{1}{2}\g{1}{-} \g{4}{-} - \frac{1}{2}  \g{1}{+} \g{2}{-} \g{4}{-}\nonumber \\
& - \frac{1}{2} \g{1}{-} \g{3}{-} \g{4}{+} - \frac{1}{2}  \g{1}{+} \g{2}{-} \g{3}{-} \g{4}{+} \: .
\end{align}

The special structure of these witnesses motivates an analytical inverstigation. In fact, we will gain more insight on the witness $W_{G6}$ and $W_{G11}$ in Section~\ref{sec:extconstr}.

\begin{table}[ht!!]
\centering
\begin{tabular}{|l|l|}
\hline
state & white noise tolerance \\
\hline
\Bsmall \Tsmall No. 1, Bell state & $p_{\rm tol} = \frac{2}{3} \approx 0.667$\\[2pt]
\hline
\Bsmall \Tsmall No. 2, ${\rm GHZ_3}$ & $p_{\rm tol} = \frac{4}{7} \approx 0.667$\\[2pt]
\hline
\Bsmall \Tsmall No. 3, ${\rm GHZ_4}$ & $p_{\rm tol} = \frac{8}{15} \approx 0.533$\\[2pt]
\hline
\Bsmall \Tsmall No. 4, ${\rm Cl_3}$& $p_{\rm tol} = \frac{8}{13} \approx 0.615$\\[2pt]
\hline
\Bsmall \Tsmall No. 5, ${\rm GHZ_5}$& $p_{\rm tol} = \frac{16}{31} \approx 0.516$\\[2pt]
\hline
\Bsmall \Tsmall No. 6, ${\rm Y_5}$ & $p_{\rm tol} = \frac{16}{25} =0.64 $\\[2pt]
\hline
\Bsmall \Tsmall No. 7, ${\rm Cl_5}$ & $p_{\rm tol} = \frac{16}{25} =0.64$\\[2pt]
\hline
\Bsmall \Tsmall No. 8, ${\rm R_5}$ & $p_{\rm tol} = \frac{12}{19} \approx 0.632$\\[2pt]
\hline
\Bsmall \Tsmall No. 9, ${\rm GHZ_6}$ & $p_{\rm tol} = \frac{32}{63} \approx 0.508$\\[2pt]
\hline
\Bsmall \Tsmall No. 10 & $p_{\rm tol} = \frac{32}{49} \approx 0.653$\\[2pt]
\hline
\Bsmall \Tsmall No. 11, ${\rm H_6}$ & $p_{\rm tol} = \frac{32}{45} \approx 0.711$\\[2pt]
\hline
\Bsmall \Tsmall No. 12, ${\rm Y_6}$ & $p_{\rm tol} = \frac{32}{45} \approx 0.711$\\[2pt]
\hline
\Bsmall \Tsmall No. 13, ${\rm E_6}$ & $p_{\rm tol} = \frac{32}{45} \approx 0.711$\\[2pt]
\hline
\Bsmall \Tsmall No. 14, ${\rm Cl_6}$& $p_{\rm tol} = \frac{128}{179} \approx 0.715$\\[2pt]
\hline
\Bsmall \Tsmall No. 15 & $p_{\rm tol} = \frac{32}{47} \approx 0.681$\\[2pt]
\hline
\Bsmall \Tsmall No. 16 & $p_{\rm tol} = \frac{8}{11} \approx 0.727$\\[2pt]
\hline
\Bsmall \Tsmall No. 17 & $p_{\rm tol} \approx 0.696$\\[2pt]
\hline
\Bsmall \Tsmall No. 18, ${\rm R_6}$& $p_{\rm tol}  \approx 0.667$\\[2pt]
\hline
\Bsmall \Tsmall No. 19 & $p_{\rm tol} = \frac{2}{3} \approx 0.667$\\[2pt]
\hline
\end{tabular}
\caption{\label{tab:tolerances}For graph states of up to six qubits, there are 19 classes of states which are equivalent under LU operations. Here, we show one state of each class. Using the presented criterion, one obtains a witness for each of these states (cf. Appendix~\ref{sec:graphstatewit}) which have the white noise tolerances given here.}
\end{table}

\subsection{Analytical construction methods}
\label{sec:anamethfdecwit}
In this section, we present an analytical method to construct fully decomposable witnesses for arbitrary graph states. This construction method results in witnesses which are a generalization of the linear cluster state witnesses in Eq.~(\ref{eq:cl4witness}). First, we recapitulate these witnesses in Sec.~\ref{sec:cln}, before we then generalize it to arbitrary graph states in Lemma~\ref{lem:witarbgraphstates} of Sec.~\ref{sec:arbgraphstates}.

\subsubsection{Linear cluster state}
\label{sec:cln}
We have pointed out that the witness $W_{\rm Cl4}$ of Eq.~(\ref{eq:cl4witness}) is a witness for the four-qubit linear cluster state. For the seven-qubit linear cluster state $\ket{\rm Cl_7}$ shown in Fig.~\ref{fig:cln} a), there exists a similar witness
\begin{align}
\label{eq:cl7witness}
W_{\rm Cl7} = &\frac{1}{2}\eins - \ketbra{\rm Cl_7} - \frac{1}{2} \left( \g{1}{-} \g{4}{-} \g{7}{-} \right. \nonumber \\
& +  \left. \g{1}{+} \g{4}{-} \g{7}{-}+\g{1}{-} \g{4}{+} \g{7}{-} +\g{1}{-} \g{4}{-} \g{7}{+} \right) \: .
\end{align}

\begin{figure}[ht!]
\includegraphics[height=1.2cm]{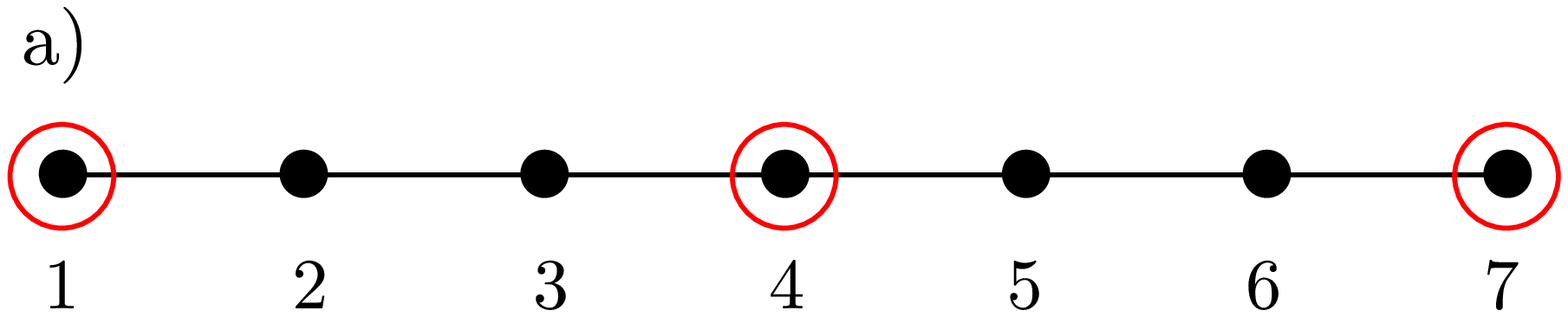}
\includegraphics[height=1.2cm]{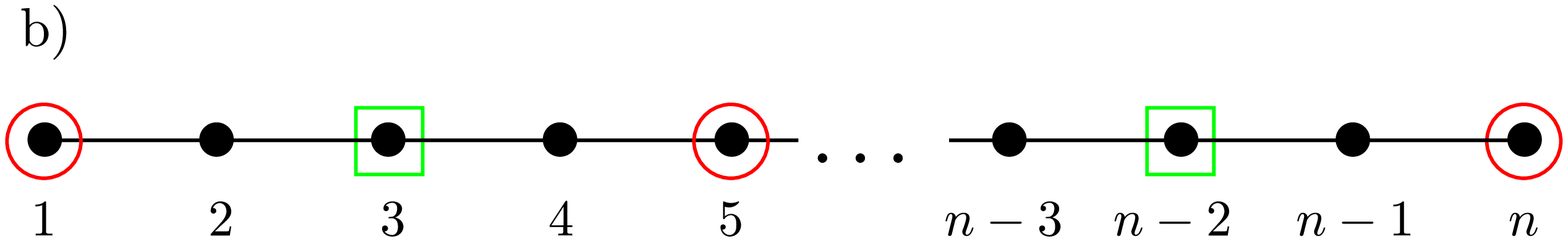}
\caption{\label{fig:cln}For the linear cluster state, we construct several witnesses. In a), the qubits in $\BB$ (marked by red circles) can be used to construct the fully decomposable witness of Eq.~(\ref{eq:cl7witness}) using Lemma~\ref{lem:witarbgraphstates}. b) illustrates the construction method of Lemma~\ref{lem:PPTwitcomb} which yields a fully PPT witness. Qubits in $\BB_1$ are marked by red circles, qubits in $\BB_2$ by green squares.}
\end{figure}

$W_{\rm Cl7}$ is a fully decomposable witness. However, since $W_{\rm Cl7}$ was not obtained from our SDP, but via Lemma~\ref{lem:witarbgraphstates}, there are --- most likely --- fully decomposable witness for $\ket{\rm Cl_7}$ with a higher white noise tolerance. This is in contrast to $W_{\rm Cl4}$ which was obtained by the semidefinite program and therefore has the maximal white noise tolerance among the fully decomposable witnesses.

$W_{\rm Cl7}$ has a very particular structure. The qubits $i$ whose generators $g_i$ appear in the witness are indicated with red circles in Fig.~\ref{fig:cln} a). Let us denote the set of these qubits by $\BB$. One can see that each two qubits in $\BB$ have at least two other qubits between them. Moreover, the terms $\g{1}{\pm} \g{4}{\pm} \g{7}{\pm}$ in Eq.~(\ref{eq:cl7witness}) all contain two or more minus signs. It turns out that witnesses of this kind can be constructed for general graph states.

\subsubsection{Arbitrary graph states}
\label{sec:arbgraphstates}
The construction in Eq.~(\ref{eq:cl7witness}) can be generalized in the following way:

\begin{lemma}
\label{lem:witarbgraphstates}
Given a connected graph state $\ket{G}$. Let $\BB = \lbrace \beta_i \rbrace$ be a subset of the set of all qubits such that any two qubits in $\BB$ are neither neighbors of each other nor have a neighbor in common. We define $b = \vert \BB\vert$. Let $\sum_{\vec{s}}$ be the sum over all vectors $\vec{s}$ of length $b$ with elements $s_i = \pm 1$ that contain at least two elements which equal $-1$, \ie, $\sum_{i=1}^{b} s_i \leq b -4$. In this case,
\be
\label{eq:decwitarbgraphstates}
W_{\rm G} = \frac{1}{2} \eins - \ketbra{G} - \frac{1}{2} \sum_{\vec{s}}\: \prod_{i \in \BB} \g{i}{s_i}
\ee
is a fully decomposable witness for $\ket{G}$.
\end{lemma}

For the detailed proof, we refer to Sec.~\ref{sec:proofarbgraph} of the Appendix. Its main idea is to construct a suitable positive operator $P_M$ for every subset $M$, such that $\left(W-P_M\right)^{T_M} = Q_M$ is positive semidefinite.

Furthermore, the proof takes advantage of the fact that, besides $\ketbra{G}$, all terms in Eq.~(\ref{eq:decwitarbgraphstates}) are invariant under any partial transposition $T_M$, since the identity is diagonal in any basis and there are no two generators $g_i$ in the product that are neighbors of each other. However, products of non-neighboring generators are only tensor products of the Pauli matrices $X$, $Z$ and the identity all of which are invariant under transposition. Moreover, the proof is simplified by $W_{\rm G}$ being diagonal in the graph state basis.

Note that in many cases, the choice of subset $\BB$ is not unique. For the seven-qubit linear cluster state, instead of the choice $\BB = \lbrace 1, 4, 7\rbrace$ which results in the witness of Eq.~(\ref{eq:cl7witness}), the choices $\BB = \lbrace 1, 6 \rbrace$ or $\BB = \lbrace 2, 5 \rbrace$ would also be valid. However, these sets would lead to witnesses that have a lower white noise tolerance.

As for the linear cluster witnesses of Ref.~\cite{ourpaper}, it turns out that for many graph states, the white noise tolerances of witnesses constructed according to Lemma~\ref{lem:witarbgraphstates} converge to one for an increasing particle number. More precisely, this is the case for graph states that can be defined for an arbitrary number of qubits such that, when increasing the number of qubits, also the number of qubits in $\BB$ grows. This includes the 2D cluster state for $n$ qubits and the ring cluster state. It does not include GHZ states, since for any number of qubits, no set $\BB$ (of more than one qubit) that contains only qubits with non-overlapping neighborhoods can be defined on the GHZ state. Let us formulate this observation as a corollary:

\begin{corollary}
\label{cor:whitenoise}
Let $\ket{G_n}$ be a graph state of $n$ qubits and $\BB(n)$ a subset of these $n$ qubits with the properties as in Lemma~\ref{lem:witarbgraphstates}. Let $W_{\rm Gn}$ be a witness for $\ket{G_n}$ as in Eq.~(\ref{eq:decwitarbgraphstates}). Then, the white noise tolerance of $W_{\rm Gn}$ with respect to $\ket{G_n}$ is given by
\be
\label{eq:generalnoise}
p(n) = \left(1-2^{-n+1}+2^{-\vert \BB(n) \vert}(\vert \BB(n) \vert+1)\right)^{-1} \: .
\ee
For a family of graph states on any number of qubits $n$ with $\vert \BB(n) \vert \xrightarrow{n \rightarrow \infty} \infty$, this expression implies
\be
\label{eq:noiseconvergence}
p(n) \xrightarrow{n \rightarrow \infty} 1 \: .
\ee
For high particle numbers, the fidelity $F_{\rm req}$ required to detect the state $\vr = (1-p) \ketbra{G_n} + p \eins/2^n$ is given by $F_{\rm req} \approx \vert \BB(n) \vert 2^{-\vert \BB(n) \vert}$ and therefore vanishes exponentially fast.
\end{corollary}
For the proof, we refer to Sec.~\ref{sec:prooftol} of the Appendix.

Note that Lemma~\ref{lem:witarbgraphstates} has been proven for the special case of linear cluster states in Ref.~\cite{ourpaper}. Moreover, entanglement criteria for Dicke states that also exhibit a white noise tolerance which converges to one have been found recently in Ref.~\cite{huberdicke}.

\subsection{Examples}
\label{sec:decex1}
{\it 2D cluster state ---} Let us consider a 2D cluster state of 16 qubits as given in Fig.~\ref{fig:clnxn} a). To construct a witness according to Lemma~\ref{lem:witarbgraphstates}, one could choose $\BB = \lbrace 1, 4, 10, 16 \rbrace$ as indicated by red circles. However, it would also be possible to choose qubit $13$ instead of qubit $10$. In both cases, the white noise tolerance is $p_{\rm tol} = \left(1-2^{-15}+5 \cdot 2^{-4} \right)^{-1} \approx 0.762$.

\begin{figure}
\includegraphics[width=0.7\columnwidth]{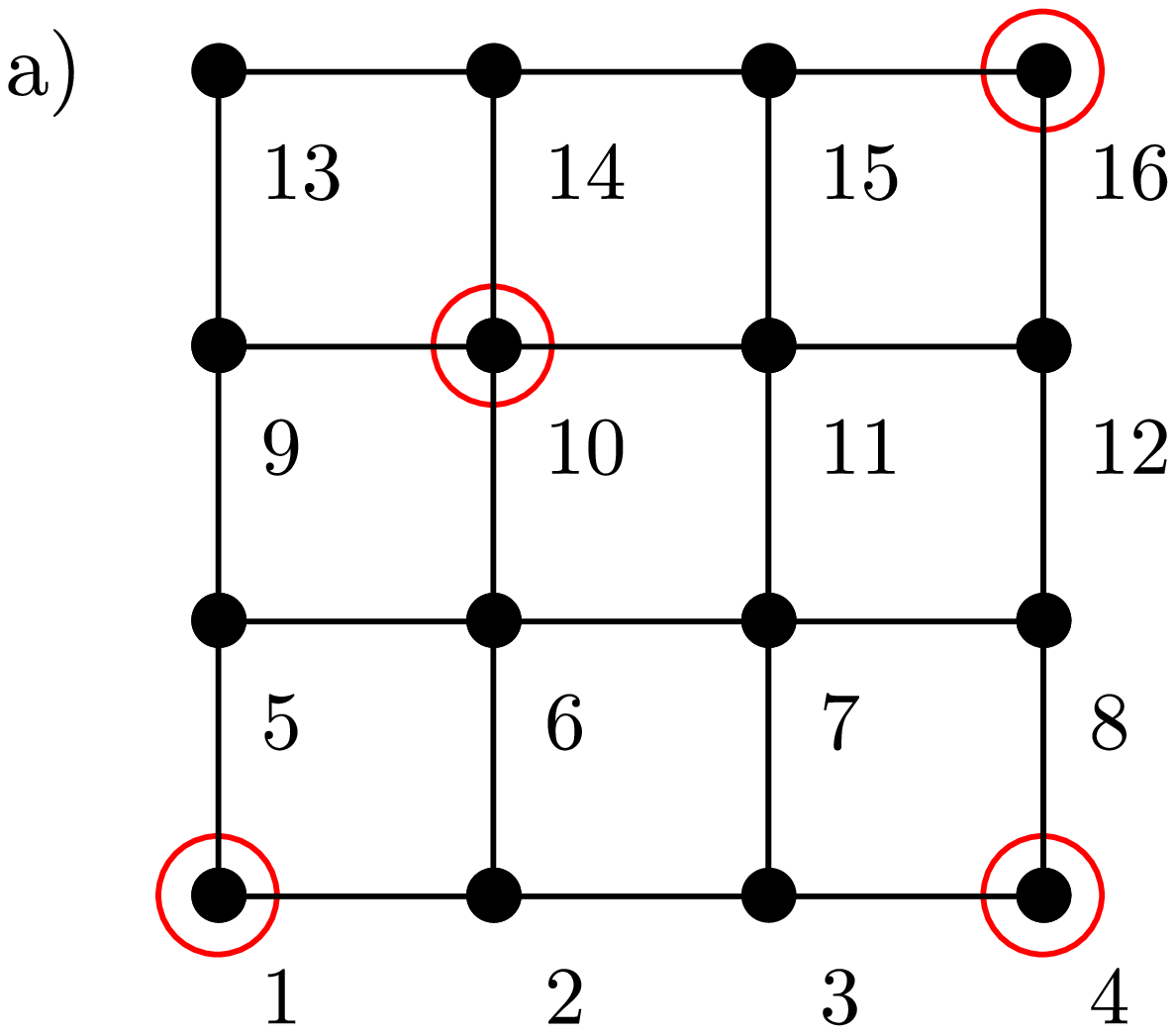}
\includegraphics[width=0.7\columnwidth]{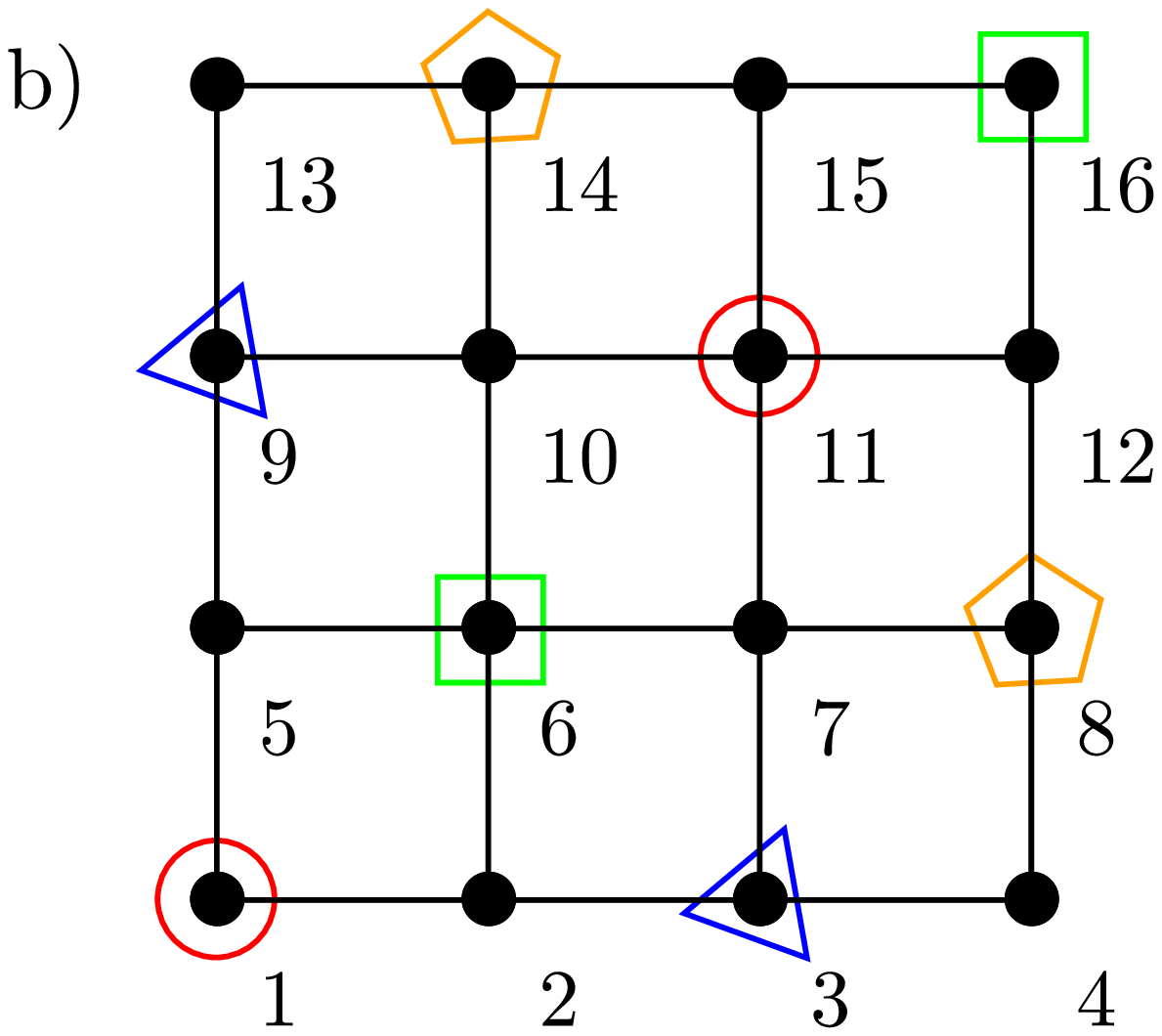}
\caption{\label{fig:clnxn} We illustrate two different ways to construct witnesses for the 2D cluster state. In a), the red circles mark qubits that belong to $\BB$, which can be used to construct a fully decomposable witness according to Lemma~\ref{lem:witarbgraphstates} (or a fully PPT witness using Lemma~\ref{lem:PPTwitarbgraphstates}). In b), we illustrate the method of Lemma~\ref{lem:PPTwitcomb} which results in a fully PPT witness. For this, one needs to define the sets $\BB_1$ (red circles), $\BB_2$ (blue triangles), $\BB_3$ (green squares) and $\BB_4$ (orange pentagons).}
\end{figure}

{\it Other graph states ---} Consider state No.~13, the $E_6$ state, of Table~\ref{tab:graphstates}. Here, $\BB = \lbrace 1, 5, 6 \rbrace$ would be a valid choice.

For state No.~11, the $H_6$ state,  $\BB = \lbrace 1, 4 \rbrace$ is a possible choice. However, one could have also selected $\BB = \lbrace 1, 3 \rbrace$, $\BB = \lbrace 2, 3 \rbrace$ or $\BB = \lbrace 2, 4 \rbrace$. Indeed, in the next section, we will see that all these choices can be combined to construct an even better witness, namely the witness of Appendix~\ref{sec:graphstatewit} which is obtained by our SDP. As mentioned before, the corresponding white noise tolerances are given in Table~\ref{tab:tolerances}.

\subsection{Extended construction method}
\label{sec:extconstr}

Although Lemma~\ref{lem:witarbgraphstates} can be applied to many graph states, for most graph states there exist witnesses with a higher white noise tolerance (cf. Appendix~\ref{sec:graphstatewit}). In this section, we provide an extended construction method that, for some states, allows one to subtract additional terms from the witnesses constructed by Lemma~\ref{lem:witarbgraphstates}. This extended method can be applied to, e.g., the states No.~6 ($Y_5$) and No.~11 ($H_6$) of Table~\ref{tab:graphstates} to obtain the witnesses of Eqs. (\ref{eq:WG6}) and (\ref{eq:WG11}).

\begin{lemma}
\label{lem:witcomb}
Given a connected graph state $\ket{G}$ and $m$ subsets $\BB_i$ of its qubits that fulfill the following two conditions:
\begin{enumerate}[(i)]
\item No two qubits in a set $\BB_i$ have a neighbor in common or are neighbors of each other.
\item Any two qubits $\beta_j^{(i)} \in \BB_i$ and $\beta_l^{(k)} \in \BB_k$ from two different subsets either have the same neighborhood or no common neighbor at all.
\end{enumerate}
Moreover, let $W_i$ be the fully decomposable witnesses that one can construct from the subsets $\BB_i$ according to Lemma~\ref{lem:witarbgraphstates}. Then,
\be
\label{eq:combwit}
W = \sum_{\vec{k} \in \lbrace 0,1 \rbrace^n} \ketbra{\vec{k}} \min_{i=1,\dots, m} \bra{\vec{k}} W_i \ket{\vec{k}}
\ee
is a fully decomposable witness. Note that $W$ is clearly better than any of the witnesses $W_i$ alone.
\end{lemma}
Note that it is possible, according to conditions (i) and (ii), that a qubit is in more than one subset $\BB_i$. The proof of Lemma~\ref{lem:witcomb} is given in Sec.~\ref{sec:proofextconstr} of the Appendix. Let us present some examples of the witnesses constructed in this lemma.

\subsection{Examples}
\label{sec:decex2}
{\it State No.~6 ($Y_5$) ---} Consider state No.~6 of Table~\ref{tab:graphstates}. Here, the subsets $\BB_i$ are given by $\BB_1 = \lbrace 1, 4 \rbrace$ and $\BB_2 = \lbrace 5,4 \rbrace$, which fulfill the conditions of Lemma~\ref{lem:witcomb}, since $\NN(1) = \NN(5)$. Lemma~\ref{lem:witarbgraphstates} then yields the two witnesses
\begin{align}
W_1 = &\frac{\eins}{2} - \ketbra{G} - \frac{1}{2} \g{1}{-} \g{4}{-} \:, \\
W_2 = &\frac{\eins}{2} - \ketbra{G} - \frac{1}{2} \g{5}{-} \g{4}{-} \:.
\end{align}
Thus, performing the minimization of Eq.~(\ref{eq:combwit}) is tantamount to subtracting the terms $\g{1}{-} \g{4}{-}/2 = \sum_{i,j,k \in \lbrace 0,1\rbrace} \ketbra{1 i j 1 k} /2 $ and $\g{5}{-} \g{4}{-}/2 = \sum_{i,j,k \in \lbrace 0,1\rbrace} \ketbra{i j k 1 1} /2$ from the projector witness and then adding the terms which have been subtracted twice in this way, namely $\g{1}{-} \g{4}{-} \g{5}{-}/2$. This results in the witness given in Eq.~(\ref{eq:WG6}).

{\it State No.~11 ($H_6$) ---} Similarly, State No.~11 allows to define four subsets, namely $\BB_1 = \lbrace 1, 4\rbrace$, $\BB_2 = \lbrace 2, 4\rbrace$, $\BB_3 = \lbrace 1, 3\rbrace$ and $\BB_4 = \lbrace 2, 3\rbrace$. Applying Lemma \ref{lem:witcomb} leads to the witness of Eq.~(\ref{eq:WG11}).

\section{Fully PPT witnesses}
\label{sec:fpptwit}
In this section, we provide analytical construction methods for fully PPT witnesses of graph states. In Sec.~\ref{sec:anamethfpptwit}, Lemma~\ref{lem:PPTwitarbgraphstates} gives a method analogous to the fully decomposable witnesses in Lemma~\ref{lem:witarbgraphstates}. An example will be given in Sec.~\ref{sec:pptex1}. 

As in the last section, we then provide an extended method to construct even better witnesses using the witnesses of Lemma~\ref{lem:PPTwitarbgraphstates}. This is done in Sec.~\ref{sec:exconstrfpptwit}, with examples in Sec.~\ref{sec:pptex2}. This time, however, the extension is more general and can be applied to a larger family of states. Thus, our main results of this section are Lemmata~\ref{lem:PPTwitarbgraphstates} and \ref{lem:PPTwitcomb}. Finally, we provide a witness for the 2D cluster state in Sec.~\ref{sec:2dclfpptwit} which does not fit into the construction methods presented so far.

As mentioned before, fully PPT witnesses are easier to characterize, since they are fully decomposable witnesses with $P_M = 0$ for all $M$. This allows for a further generalization of the construction methods presented above --- however, only resulting in fully PPT witnesses --- and for the construction of a new witness for the 2D cluster state. 

\subsection{Arbitrary graph states}
\label{sec:anamethfpptwit}
Let us first give the analogon to Lemma~\ref{lem:witarbgraphstates} for fully PPT witnesses.

\begin{lemma}
\label{lem:PPTwitarbgraphstates}
Given a connected graph state $\ket{G}$. Let $\BB = \lbrace \beta_i \rbrace$ be a subset of the set of all qubits such that any two qubits in $\BB$ are neither neighbors of each other nor have a neighbor in common. We define $b = \vert \BB\vert$. Let $\sum_{\vec{s}}$ be the sum over all vectors $\vec{s}$ of length $b$ with elements $s_i = \pm 1$ that contain at least two elements which equal $-1$, \ie, $\sum_{i=1}^{b} s_i \leq b -4$. In this case,
\be
\label{eq:PPTwitarbgraphstates}
W_{\rm G} = \frac{1}{2} \eins - \ketbra{G} - \sum_{\vec{s}}\: \left( \frac{1}{2} - \frac{1}{2^{m(\vec{s})}} \right) \: \prod_{i \in \BB} \g{i}{s_i}
\ee
is a fully PPT witness for $\ket{G}$. Here, $m(\vec{s})$ is the number of elements $s_i = -1$ in $\vec{s}$, i.e., $m(\vec{s})=\left(b-\sum_{i=1}^{b} s_i \right)/2$.
\end{lemma}

The proof is similar to the proof of Lemma~\ref{lem:witarbgraphstates}, but some parts are easier. We present it in Sec.~\ref{sec:PPTproofarbgraph} in the Appendix. 

\subsection{Examples}
\label{sec:pptex1}
{\it 2D cluster state ---} When applying the presented construction to the 2D cluster state of Fig.~\ref{fig:clnxn}, one obtains the witness
\begin{align}
\label{eq:82}
W = \:& \frac{1}{2} \eins - \ketbra{Cl_{4 \times 4}} \nonumber \\
&- \frac{1}{4} \left( \g{1}{-} \g{4}{-}\g{10}{+} \g{16}{+}  + \g{1}{-} \g{4}{+}\g{10}{+} \g{16}{-} + \g{1}{+} \g{4}{+}\g{10}{-} \g{16}{-} \right. \nonumber \\
&  \qquad \left. +\:\g{1}{+} \g{4}{-}\g{10}{-} \g{16}{+} +\g{1}{+} \g{4}{-}\g{10}{+} \g{16}{-}+\g{1}{-} \g{4}{+}\g{10}{-} \g{16}{+} \right) \nonumber \\
&- \frac{3}{8} \left( \g{1}{-} \g{4}{-}\g{10}{-} \g{16}{+} +\g{1}{-} \g{4}{-}\g{10}{+} \g{16}{-} \right. \nonumber \\
&  \qquad \left. +\: \g{1}{-} \g{4}{+}\g{10}{-} \g{16}{-}+\g{1}{+} \g{4}{-}\g{10}{-} \g{16}{-} \right) \nonumber \\
&\: -\frac{7}{16} \g{1}{-} \g{4}{-}\g{10}{-} \g{16}{-} \: .
\end{align}
This witness has a white noise tolerance of $p_{\rm tol} = \frac{32768}{51455} \approx 0.637$.

\subsection{Extended construction method}
\label{sec:exconstrfpptwit}
We can now rewrite Lemma~\ref{lem:witcomb} for fully PPT witnesses. Although these witnesses have a smaller white noise tolerance, they can be handled easier analytically, which enabled us to relax the premises of Lemma~\ref{lem:witcomb}. Therefore, one can apply the new lemma to a larger class of states.

\begin{lemma}
\label{lem:PPTwitcomb}
Given a connected graph state $\ket{G}$ and $m$ subsets $\BB_i$ of its qubits that fulfill the following two conditions:
\begin{enumerate}[(i)]
\item No set $\BB_i$ contains two qubits that have a neighbor in common.
\item No two qubits in $\operatorname*{\cup}_{i=1}^{m} \BB_i$ are neighbors of each other.
\end{enumerate}
Moreover, let $W_i$ be the fully PPT witnesses that one can construct from the subsets $\BB_i$ according to Lemma~\ref{lem:PPTwitarbgraphstates}. Then,
\be
\label{eq:PPTcombwit}
W = \sum_{\vec{k} \in \lbrace 0,1 \rbrace^n} \ketbra{\vec{k}} \min_{i=1,\dots, m} \bra{\vec{k}} W_i \ket{\vec{k}}
\ee
is a fully PPT witness.
\end{lemma}

The proof of Lemma~\ref{lem:PPTwitcomb} can be found in Sec.~\ref{sec:PPTproofextconstr} in the Appendix.

\subsection{Examples}
\label{sec:pptex2}

{\it Linear cluster state ---} Consider an $n$-qubit linear cluster state as shown in Fig.~\ref{fig:cln} b). We define a subset $\BB_1$ for the construction of a witness $W_1$ according to Lemma~\ref{lem:PPTwitarbgraphstates}) by picking the qubits $\BB_1 = \lbrace 1, 5, 9, ... \rbrace$. These are marked by red circles in Fig.~\ref{fig:cln} b). Analogously, the qubits marked by a green square belong to a second subset $\BB_2$ which is used to construct a witness $W_2$. Then, Lemma~\ref{lem:PPTwitcomb} implies that there is a witness $W$ as given in Eq.~(\ref{eq:PPTcombwit}).

Let us present this witness for a seven-qubit cluster state. Then, $\BB_1 = \lbrace 1, 5 \rbrace$ and $\BB_2 = \lbrace 3, 7 \rbrace$. Consequently,
\begin{align}
W_1 = &\frac{1}{2} \eins - \ketbra{Cl_7} - \frac{1}{4} \g{1}{-} \g{5}{-} \:,\\
W_2 = &\frac{1}{2} \eins - \ketbra{Cl_7} - \frac{1}{4} \g{3}{-} \g{7}{-} \:.
\end{align}
Since the only terms that $\g{1}{-} \g{5}{-}$ and $ \g{3}{-} \g{7}{-}$ have in common are given by $\g{1}{-} \g{3}{-}  \g{5}{-} \g{7}{-}$, Eq.~(\ref{eq:PPTcombwit}) can be expressed as

\begin{align}
\label{eq:cl7combwit}
W_{\rm Cl7,2} = \:& \frac{1}{2} \eins - \ketbra{Cl_7} - \frac{1}{4}\g{1}{-} \g{5}{-} - \frac{1}{4} \g{3}{-} \g{7}{-} \nonumber \\
&+ \frac{1}{4} \g{1}{-} \g{3}{-}  \g{5}{-} \g{7}{-} \: .
\end{align}

A fully PPT witness for the seven-qubit linear cluster state constructed according to Lemma~\ref{lem:PPTwitarbgraphstates} with $\BB = \lbrace 1, 4, 7\rbrace$ has a white noise tolerance of $p_{\rm tol} = 64/109 \approx 0.588$. The witness of Eq.~(\ref{eq:cl7combwit}), however, only has a tolerance of $p_{\rm tol} = 64/113 \approx 0.566$. While Lemma~\ref{lem:PPTwitcomb} does not allow to construct more robust witnesses for linear cluster states compared to simply using Lemma~\ref{lem:PPTwitarbgraphstates}, it still has some advantages.

First, for many graph states, e.g. the state No.~6 ($Y_5$) and the state No.~11 ($H_6$) of Table~\ref{tab:graphstates}, Lemma~\ref{lem:PPTwitcomb} {\it does} provide a method to construct witnesses that are more robust than witnesses constructed via Lemma~\ref{lem:PPTwitarbgraphstates} alone. We note that the fully decomposable witnesses of Lemma~\ref{lem:witcomb} are even more robust. However, as mentioned before, the prerequisites for Lemma~\ref{lem:witcomb} are more strict than those for Lemma~\ref{lem:PPTwitcomb} and therefore, there are graph states for which the former cannot be used, but the latter applies. For example, this is the case for the 2D cluster state of 16 qubits, to which Lemma~\ref{lem:PPTwitcomb} can be applied, as we will see at the end of this section, but Lemma~\ref{lem:witcomb} can not be used as there are no two qubits with the same neighborhood.

Second, witnesses constructed according to Lemma~\ref{lem:PPTwitcomb} using two sets $\BB_1$ and $\BB_2$ as shown in Fig.~\ref{fig:cln} b) can be used to improve the linear cluster state witnesses $\mathcal{W}^{(C_N)}$ of Ref.~\cite{entdecstab}, which results in a witness that only needs two experimental settings to be measured.

To illustrate this, we consider the seven-qubit linear cluster state and its witness $W_{\rm Cl7,2}$ of Eq.~(\ref{eq:cl7combwit}) again. The linear cluster state witness of Eq.~(9) in Ref.~\cite{entdecstab} is given by
\be
\mathcal{W}^{(C_N)} = \frac{3}{2} \eins - \left( \prod_{i = 1,3,5,7} \g{i}{+} +  \prod_{i = 2,4,6} \g{i}{+}\right) \:.
\ee

Due to the form of the generators, it can be measured locally using only two settings, namely the eigenbases of $X_1 Z_2 X_3 Z_4 X_5 Z_6 X_7$ and $Z_1 X_2 Z_3 X_4 Z_5 X_6 Z_7$. Since \mbox{$\mathcal{W}^{(C_N)} \geq \frac{1}{2} \eins - \ketbra{G}$}, one has

\begin{align}
W_{\rm Cl7,2} = \:& \frac{1}{2} \eins - \ketbra{Cl_7} - \frac{1}{4}\g{1}{-} \g{5}{-} - \frac{1}{4} \g{3}{-} \g{7}{-} \nonumber \\
&+ \frac{1}{4} \g{1}{-} \g{3}{-}  \g{5}{-} \g{7}{-} \nonumber \\
\leq \: &\mathcal{W}^{(C_N)} - \frac{1}{4}\g{1}{-} \g{5}{-} - \frac{1}{4} \g{3}{-} \g{7}{-} \nonumber \\
&+ \frac{1}{4} \g{1}{-} \g{3}{-}  \g{5}{-} \g{7}{-} \nonumber \\
= \: & \mathcal{W}^{(C_N)}_{\rm imp} \: ,
\end{align}
where the last equality sign defines the improved witness $\mathcal{W}^{(C_N)}_{\rm imp}$. This witness detects more states than $\mathcal{W}^{(C_N)}$ and also requires only two settings, since the additional terms can be determined through the measurement of $X_1 Z_2 X_3 Z_4 X_5 Z_6 X_7$. Note that this is not in contradiction with the result of Ref.~\cite{entdecstabpra} stating that $\mathcal{W}^{(C_N)}$ has the highest possible white noise tolerance amongst all stabilizer witnesses that can be measured using two settings, as only witnesses obeying \mbox{$\mathcal{W}^{(C_N)} \geq \alpha (\frac{1}{2} \eins - \ketbra{G})$} for some $\alpha > 0$ where considered in Ref.~\cite{entdecstabpra}.

Note that it is possible to construct a better witness for linear cluster state of seven qubits by adding a third witness $W_3$ constructed for the subset $\BB_3 =\lbrace 1,7\rbrace$. Then, the white noise tolerance increases to $p_{\rm tol} = \frac{64}{111}\approx 0.577$.

Finally, we apply the construction of Lemma~\ref{lem:PPTwitcomb} to the 2D cluster state of 16 qubits.

{\it 2D cluster state ---} Fig.~\ref{fig:clnxn} b) shows how to choose four subsets $\BB_i$ of qubits from a 2D cluster state $\ket{Cl_{4 \times 4}}$ made up of 16 qubits. $\BB_1$ is shown by red circles, $\BB_2$ by blue triangles, $\BB_3$ by green squares and $\BB_4$ by orange pentagons. The resulting witnesses $W_i$ can be combined as in Eq.~(\ref{eq:PPTcombwit}) to yield a witness that can be rewritten as
\begin{align}
\label{eq:95}
W =  \frac{1}{2}\eins &- \ketbra{Cl_{4 \times 4}} \nonumber \\
 - \frac{1}{4}  & \left( \g{1}{-}\g{11}{-} + \g{6}{-}\g{16}{-}+ \g{3}{-}\g{9}{-}+ \g{8}{-}\g{14}{-} \right) \nonumber \\
 + \frac{1}{4} & \left( \g{1}{-}\g{11}{-} \g{6}{-}\g{16}{-} +\g{6}{-}\g{16}{-}\g{3}{-}\g{9}{-}+\g{3}{-}\g{9}{-}\g{8}{-}\g{14}{-} \right. \nonumber \\
& \left.+\:\g{1}{-}\g{11}{-} \g{8}{-}\g{14}{-} +\g{6}{-}\g{16}{-}\g{8}{-}\g{14}{-}+\g{1}{-}\g{11}{-}\g{3}{-}\g{9}{-} \right)\nonumber\\
- \frac{1}{4} &\left(\g{1}{-}\g{11}{-} \g{6}{-}\g{16}{-} \g{3}{-}\g{9}{-}+\g{1}{-}\g{11}{-} \g{6}{-}\g{16}{-}  \g{8}{-}\g{14}{-} \right.\nonumber \\
& \left.+\: \g{1}{-}\g{11}{-}  \g{3}{-}\g{9}{-} \g{8}{-}\g{14}{-}+\g{6}{-}\g{16}{-} \g{3}{-}\g{9}{-} \g{8}{-}\g{14}{-}\right) \nonumber \\
+ \frac{1}{4} & \g{1}{-}\g{11}{-} \g{6}{-}\g{16}{-} \g{3}{-}\g{9}{-} \g{8}{-}\g{14}{-} \: .
\end{align}

This witness has a white noise tolerance of \mbox{$p_{\rm tol} = \frac{32768}{54335} \approx 0.603$}. As we noted for linear cluster state, there are even more subsets $\BB_i$ that one can use, such as $\BB_5 = \lbrace 1, 8 \rbrace$ and $\BB_6 = \lbrace 3,9,16 \rbrace$. In fact, there are 13 subsets of $\lbrace 1,3,6,8,9,11,14,16\rbrace$ that obey condition (ii) of Lemma~\ref{lem:PPTwitcomb}. Taking all of them into account, one obtains a witness with white noise tolerance $p_{\rm tol} = \frac{32768}{49791} \approx 0.658$ which is even better than the witness of Eq.~(\ref{eq:82}).

\subsection{2D cluster state}
\label{sec:2dclfpptwit}
Finally, we present a fully PPT witness for the 2D cluster state $\ket{\rm Cl_{4 \times 4}}$ of 16 qubits which does not fit into the framework of Lemma~\ref{lem:PPTwitarbgraphstates}. Although the construction can easily be generalized to $n \times n$ qubits, we present the witness for the $4 \times 4$ case here. To circumvent any problems that might occur due to the border, we consider this state on a torus, \ie, with periodic boundary conditions as shown in Fig.~\ref{fig:clnxntorus}.

\begin{figure}[h]
\includegraphics[width=0.7\columnwidth]{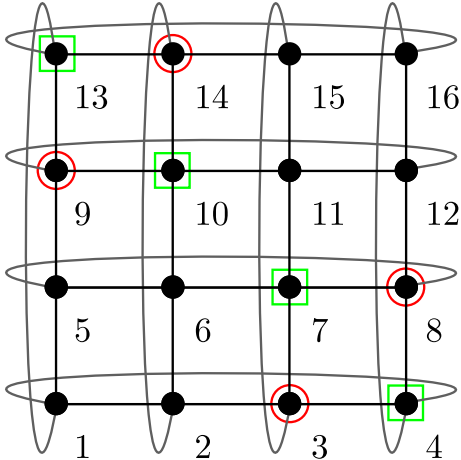}
\caption{\label{fig:clnxntorus}For the 2D cluster state on a torus, it is possible to define a fully PPT witness using the diagonals (cf. Lemma~\ref{lem:PPTwitCl4x4}).}
\end{figure}

The 2D cluster state has four parallel diagonals in one direction and, orthogonal to these, another set of four diagonals. All of these diagonals contain four qubits. The first set is made up of diagonals parallel to the diagonal $\lbrace 3,8,9,14 \rbrace$ which is indicated by red circles. We denote this set by 

\begin{align}
\label{eq:91}
\DD_{/} = \:\lbrace &\DD_{/}^{(j)} \rbrace \nonumber \\
= \: \lbrace &\lbrace 1, 6, 11, 16\rbrace, \lbrace 2, 7,12,13\rbrace, \nonumber \\
& \lbrace 3,8,9,14\rbrace, \lbrace 4,5,10,15\rbrace  \rbrace \:.
\end{align}
The second set contains diagonals parallel to the one marked by green squares, $\lbrace 4, 7, 10, 13\rbrace$. We define it as
\begin{align}
\label{eq:92}
\DD_{\backslash} = \: \lbrace & \DD_{\backslash}^{(j)} \rbrace \nonumber \\
= \: \lbrace &\lbrace 1,8,11,14\rbrace,\lbrace 2,5,12,15\rbrace ,  \nonumber \\
& \lbrace 3,6,9,16\rbrace , \lbrace 4,7,10,13\rbrace \rbrace \:.
\end{align}
We can now introduce the following witness.

\begin{lemma}
\label{lem:PPTwitCl4x4}
Given the 2D cluster state of 16 qubits with periodic boundary conditions $\ket{Cl_{4 \times 4}}$. By $\DD_{/}$ and $\DD_{\backslash}$, we denote the two sets of diagonals as defined above. For each pair of orthogonal diagonals that have no qubit in common, i.e. for each $(i,j)$ such that $\DD_{/}^{(i)} \cap\DD_{\backslash}^{(j)}= \lbrace \rbrace$, we define a projector 

\be
\label{eq:ddef}
D_{(i,j)} =   \frac{1}{2}(\eins - \prod_{k \in \DD_{/}^{(i)}} g_k) \frac{1}{2}(\eins - \prod_{l \in \DD_{\backslash}^{(j)}}g_l) \:.
\ee
Then,

\begin{align}
\label{eq:PPTwitCl4x4}
W_{4 \times 4} = & \: \frac{1}{2} \eins - \ketbra{Cl_{4 \times 4}} \nonumber \\
&- \frac{1}{4} \sum_{\vec{k}} \ketbra{\vec{k}} \max_{(i,j)}\: \bra{\vec{k}} D_{(i,j)} \ket{\vec{k}}
\end{align}
is a fully PPT witness for $\ket{Cl_{4 \times 4}}$. 

\end{lemma}

The proof can be found in Sec.~\ref{sec:PPTwitCl4x4} of the Appendix. Note that Eq.~(\ref{eq:PPTwitCl4x4}) is easy to generalize to $n\times n$ qubits, as for a larger number of qubits only the definitions of Eqs. (\ref{eq:91}) and (\ref{eq:92}) would have to be changed. The proof provided in the Appendix works for $n\times n$ qubits with $n \geq 3$.

Specifically, the maximization in Eq.~(\ref{eq:PPTwitCl4x4}) is carried out over the operators $D_{(1,2)}$, $D_{(1,4)}$, $D_{(2,3)}$, $D_{(2,1)}$, $D_{(3,2)}$, $D_{(3,4)}$, $D_{(4,3)}$ and $D_{(4,1)}$. Similarly to Eq.~(\ref{eq:95}), this maximum can also be written as a polynomial in the operators $D_{(i,j)}$. Moreover, the expectation values of these operators can be determined by measuring one experimental setting, namely $X$-measurements on all qubits. Thus, the sum in Eq.~(\ref{eq:PPTwitCl4x4}) can be obtained by implementing one experimental setting.

In order to determine the terms that the witness $W_{4 \times 4}$ contains in addition to the projector witness, i.e. the sum in Eq.~(\ref{eq:PPTwitCl4x4}), one has to measure the operators $D_{(i,j)}$ that obey $\DD_{/}^{(i)} \cap\DD_{\backslash}^{(j)}= \lbrace \rbrace$. These are the operators $D_{(1,2)}$, $D_{(1,4)}$, $D_{(2,3)}$, $D_{(2,1)}$, $D_{(3,2)}$, $D_{(3,4)}$, $D_{(4,3)}$ and $D_{(4,1)}$. From these, one can determine the elementwise maximum in Eq.~(\ref{eq:PPTwitCl4x4}), as one can show that it can be written as a polynomial in the operators $D_{(i,j)}$. Moreover, the expectation values of all of these operators can be determined by measuring one experimental setting, namely $X$-measurements on all qubits. Thus, the additional term in Eq.~(\ref{eq:PPTwitCl4x4}) can be obtained by implementing one experimental setting. 

The white noise tolerance of $W_{4 \times 4}$ is given by $p_{\rm tol} = \frac{32768}{53503} \approx 0.612$.

\section{Entanglement monotone}
\label{sec:entmeas}
Finally, we consider another variation of the program in Eq.~(\ref{eq:sdp1}) which results in an entanglement monotone for genuine multipartite entanglement. In this section, we will present one lemma which traces this monotone back to the negativity in the two-particle case and another lemma that specifies the values that the monotone can take.

For a generic multipartite state $\vr$, consider the quantity
\begin{align}
\label{eq:monotone}
&N(\vr)= -\min_{W \in \mathcal{W}} \trace(\vr W),\\
\label{eq:setfullydec}
&\mathcal{W}= \left\{ W \big| \; \mbox{for all} \: M: \exists\; P_M, Q_M \: \mbox{such that} \right.  \nonumber \\
& \qquad \qquad \left. 0 \leq P_M, Q_M \leq \mathbbm{1}  \: {\rm and} \: W=P_M + Q_M^{T_M} \right\}\!, 
\end{align}
where $M$ is a strict subset of the set of all qubits. Note that the class $\mathcal{W}$ consists of fully decomposable witnesses which are only normalized in a different way than before. Then, the following lemma holds.

\begin{lemma}
\label{lem:equalsnegativity}
$N(\varrho)$ fulfills the following properties:
\begin{itemize}
\item $N(\varrho^{\rm bs})=0$ for all biseparable states $\varrho^{\rm bs}$.
\item $N[\Lambda_{\rm LOCC}(\varrho)]\leq N(\varrho)$ for all full LOCC operations.
\item $N(U_{\rm loc} \varrho U^\dag_{\rm loc})=N(\varrho)$ for local basis changes $U_{\rm loc}$.
\item $N(\sum_i p_i \varrho_i) \leq \sum_i p_i N(\varrho_i)$ holds for all convex combinations $\sum_i p_i \varrho_i$.
\end{itemize}
Thus, $N(\varrho)$ is a monotone for genuine multipartite entanglement. In the bipartite case, the monotone $N(\varrho)$ of Eq.~(\ref{eq:monotone}) equals the negativity.
\end{lemma}
For the proof that $N(\varrho)$ is a monotone, we refer to Ref.~\cite{ourpaper}. Finally, it is interesting to know which values $N(\varrho)$ can take and what the maximally entangled states are.
\begin{lemma}
\label{lem:onehalf}
For any state $\vr$ of $n$ qubits,
\be
\label{eq:84}
N(\vr) \leq \frac{1}{2} \:.
\ee
For any connected graph state $\ket{G}$,
\be
N(\ketbra{G}) = \frac{1}{2} \: .
\ee
\end{lemma}

Therefore, connected graph states are maximally entangled states for this monotone. We note that, if the system does not only consist of qubits, but also of higher-dimensional particles, Eq.~(\ref{eq:84}) must be replaced by

\be
N(\vr) \leq \frac{1}{2} (d_{\rm min} -1)\:,
\ee
where $d_{\rm min} $ is the lowest dimension of any particle in the system.

The proofs of this section are given in Secs.~\ref{sec:proofentmeas} and \ref{sec:proofentmeas2} of the Appendix.

\section{Conclusion}

In this paper we presented general construction methods for graph state witnesses in the 
framework of PPT mixtures \cite{ourpaper}. These methods can be applied to a 
large class of graph states, resulting in witnesses that are significantly better 
than previously known witnesses. In many 
cases, the white noise tolerances approach one for an increasing particle number. This means 
that for many qubits, the state fidelity can decrease exponentially, but still entanglement is
present and can be detected. Moreover, the improvement of the witnesses comes with very low 
experimental costs, as the additional terms which are not part of the standard projector witness 
can be measured with one local setting.

For these reasons, we believe that the presented entanglement witnesses will prove to be 
useful in experiments, also for future experiments involving larger qubit numbers. 
Furthermore, the applied methods can serve as starting points for the construction 
of even better entanglement criteria.

For future work, there are several open questions. First, as we have seen, the approach
of Ref.~\cite{ourpaper} results in strong separability conditions for noisy graph states.
It would be interesting to find out whether these conditions are already necessary sufficient 
for entanglement, or whether they can still be improved.

Second, there are many other interesting families of multi-qubit states besides graph states, 
e.g., Dicke states or singlet states. It would be desirable to similarly develop witnesses 
for these families of states using the framework developed here.

We thank M.~Kleinmann, S.~Niekamp, M.~Hofmann and G.~T\'oth for discussions and acknowledge support 
by the FWF (START Prize and SFB FOQUS).

\appendix

\section{Proofs}
\subsection{Linear program for graph-diagonal states (Lemma~\ref{lem:lp})}
\label{sec:prooflp}
\begin{proof}
Let us define a simplifying notation for this proof: For any operator $O$ we define its graph-diagonal form as $\overline{O}= \sum_{\vec{k}} \ketbra{\vec{k}} O \ketbra{\vec{k}} $. Note that any state $\varrho$ can be transformed into its graph-diagonal form $\overline{\varrho}$ by local operations. Now suppose that the operator $W$ is the fully decomposable entanglement witnesses that minimizes the expectation value for the graph diagonal state $\varrho_G$ according to the original problem of Eq.~(\ref{eq:sdp1}). Then its graph-diagonal operator $\overline{W}$ has the same expectation value $\trace(W\rho_G)=\trace(\overline{W} \rho_G)$ as the original witness. Given any valid decomposition $W=P+Q^{T_M}$ for a particular chosen bipartition $M$, the operator $\overline{W} = \overline{P} + \overline{Q^{T_M}}$ can be expressed in its corresponding graph-diagonal operators $\overline{P}$ and $\overline{Q^{T_M}}$ due to linearity, but note that $\overline{Q^{T_M}}$ stands for the graph-diagonal form of the partially transposed operator. 

However, it is straightforward to see that this operator is actually identical to the partial transpose of the graph-diagonal operator, \ie, $\overline{Q^{T_M}} = \overline{Q}^{T_M}$, as follows: The mapping of $Q \mapsto \overline{Q}$ is achieved by expanding $O$ in the Pauli basis, $Q = \sum_{\vec{x} \in \lbrace 0, 1,2,3 \rbrace^n }\alpha_{\vec{x}} \operatorname*{\otimes}_{i=1}^{n} \sigma_{x_i}$, and then setting to zero all coefficients $\alpha_{\vec{x}}$ of Pauli matrix products which are no stabilizers of the given graph state. Note that $\sigma_{1}$, $\sigma_{2}$, $\sigma_{3}$ denote the Pauli matrices and $\sigma_{0}$ is the identity. In this picture, the partial transposition only corresponds to flipping the sign of coefficients $\alpha_{\vec{x}}$ of Pauli matrix products which change under partial transposition. These are the Pauli matrix products in which there is an odd number of $\sigma_{2}$s, i.e. of $Y$s, in the set $M$, since $Y^T = -Y$ and all other Pauli matrices are invariant under transposition.

Then, it is clear that the partial transposition and the mapping $Q \mapsto \overline{Q}$ commute. Thus the witness decomposition simplifies to $\overline{W} = \overline{P} + \overline{Q}^{T_M}$. Since the operator $P\geq0$ is positive semidefinite, the overlap with any basis element $\bra{\vec{k}} P \ket{\vec{k}}\geq0$ is non-negative. However this is equivalent to $\overline{P} \geq 0$ because $\overline{P}$ is diagonal in exactly this basis. The same argument applies to the operator $Q$, which concludes the proof.   
\end{proof}

\subsection{Fully decomposable witnesses for arbitrary graph states (Lemma~\ref{lem:witarbgraphstates})}
\label{sec:proofarbgraph}

\begin{proof}
Consider an arbitrary, connected graph $G = (V,E)$ consisting of a set $V$ of vertices/qubits and a set $E$ of edges that connect some of these vertices. 

In the following, $\NNN(i) = \NN(i) \cup \lbrace i \rbrace$ denotes the union of qubit $i$ and its neighborhood. Moreover, all states in the following are given in the graph state basis of the corresponding graph.

Let us first cite four lemmata to prepare the main proof. For the proofs of the first three of these lemmata, we refer to Ref.~\cite{ourpaper}. The proof of the fourth one will be given here.

The first of these lemmata shows which kind of partial transposition one can apply to one of two orthogonal vectors without affecting their orthogonality. The second one can be used to estimate the eigenvalues of a partially transposed state. More precisely, it provides an upper bound on these eigenvalues in terms of the state's Schmidt coefficients. The third lemma demonstrates that certain expressions are invariant under partial transpositions on a single qubit. Finally, the fourth lemma helps to estimate the largest Schmidt coefficient of a graph state. In order to prove it, we will count the Bell pairs that can be distilled from it using local operations and classical communication (LOCC).

We will then apply these lemmata to prove that the operator $W_{\rm G}$ of Eq.~(\ref{eq:decwitarbgraphstates}) is a fully decomposable witness.

\begin{lemma} \textnormal{\cite[Appendix E, Lemma 1]{ourpaper}}
Given a graph $G = (V,E)$ of $n$ qubits and an arbitrary bipartition $M\vert \comp{M}$ of these qubits. Let $\ket{\vec{a}}$ and $\ket{\vec{c}}$ be two arbitrary states in the associated graph state basis. If there is a qubit $i$ with $\NNN(i) \subseteq M$ or $\NNN(i) \subseteq \comp{M}$, such that $c_i \neq a_i$, then
\label{lem:rule1}
\be
\bra{\vec{c}} \left( \ketbra{\vec{a}}\right)^{T_M} \ket{\vec{c}} = 0 \: .
\ee
\end{lemma}

The states in the following lemma are generic states and no graph state basis vectors.
\begin{lemma} \textnormal{\cite[Appendix E, Lemma 2]{ourpaper}}
\label{lem:rule2}
Given a state $\ket{\psi}$ and its Schmidt decomposition $\ket{\psi} = \sum_{i=1}^{d_1} \lambda_i \ket{\mu_i} \otimes \ket{\nu_i}$ with respect to some bipartition $M|\comp{M}$, where $\lambda_i \geq 0$, $d_1 =  {\rm dim}(M)$, $d_2 = {\rm dim}(\comp{M})$ and w.l.o.g. $d_1 \leq d_2$. Then, for any state $\ket{\phi}$,
\be
\bra{\phi} \left( \ketbra{\psi} \right)^{T_M} \ket{\phi} \leq \max_{i} \lambda_i^2 \: .
\ee
\end{lemma}

Let us now return to the graph state basis and recall that the application of the Pauli operator $Z_k$ to a graph state basis vector results in a bit flip on bit $k$, \ie,
\be
\label{eq:flipprop}
Z_k \ket{\vec{a}} = \ket{a_1 \dots a_{k-1} \: a_k \oplus 1 \: a_{k+1} \dots a_n} \: .
\ee

\begin{lemma} \textnormal{\cite[Appendix E, Lemma 3]{ourpaper}}
\label{lem:rule3}
Given a graph $G$. Then, in the associated graph state basis,
\begin{align}
\label{eq:rule3}
\left( \ketbra{\vec{a}} + \ketbra{\vec{c}}\right)^{T_k} =  \ketbra{\vec{a}} + \ketbra{\vec{c}}\:,
\end{align}
i.e. $\ketbra{\vec{a}} + \ketbra{\vec{c}}$ is invariant under partial transposition on qubit $k$, if
\be
\ket{\vec{c}} = \prod_{i \in \NN(k)}Z_i \ket{\vec{a}} \: .
\ee
\end{lemma}

\begin{lemma}
\label{lem:rule4}
Let $\ket{G}$ be a graph state that is defined by a bipartite graph $G = (V,E)$, i.e. the qubits can be grouped into two partitions $M$ and $\comp{M}$, such that no two qubits in the same partition are connected with each other. Let $\lambda_{i}$ be the Schmidt coefficients of $\ket{G}$ with respect to the bipartition $M | \comp{M}$. If there exists a subset $\BB = \lbrace \beta_i \rbrace$ of $m$ qubits which have at least one neighbor and are chosen in such a way that no two qubits in $\BB$ have a neighbor in common or are neighbors of each other, then 
\be
\label{eq:rule4}
\max_{i} \lambda_{i}^2 \leq 2^{-m}\: .
\ee
\end{lemma}
\begin{proof}
Note that any graph can be made bipartite with respect to a fixed bipartition $M\vert \comp{M}$ using operations which are local with respect to $M\vert \comp{M}$. These operations are controlled-Z between two qubits $i,j$ of the same partition and they correspond to a deletion of the edge between qubits $i$ and $j$ \cite{hein}.

In order to prove that the square of the largest Schmidt coefficient of $\ket{G}$ is smaller than (or equal to) $2^{-m}$, it is sufficient to show that $\ket{G}$ can be converted into at least $m$ Bell pairs via local operations and classical communication. Since the largest Schmidt coefficient does not decrease under LOCC \cite{nielsen} and a Bell pair has Schmidt coefficients $\lbrace 1/\sqrt{2}, 1/\sqrt{2} \rbrace$, this implies the given bound.

In the first step, we choose a set of edges $F = \lbrace (\beta_i, w_i) \rbrace \subseteq E$ by selecting, for every qubit $\beta_i$ in $\BB$, a neighboring $w_i$. The edge $(\beta_i, w_i)$ between them then belongs to $F$. Since no qubit $w_i$ can be a neighbor of two different qubits in $\BB$ according to the assumptions, every qubit in the graph is endpoint of at most one of the edges in $F$. A set with this property is also called a {\it matching}. For our proof, each edge in the matching $F$ marks two qubits between which we will create a Bell pair which is disconnected from the rest of the graph.

As a second step, we measure every qubit, which is not an end point of an edge in $F$, in the $Z$-basis. In terms of the graph, this deletes all edges that are incident on a measured qubits. Fig.~\ref{fig:belltrafo1} shows an example of a graph that emerges from these measurements. There are two kinds of edges left: edges that are contained in the matching (shown as thick, red lines in Fig.~\ref{fig:belltrafo1}) and edges that connect a qubit $w_i$ to a qubit $w_j$ in the opposite partition, which are not in the matching (drawn thinner and in black). Note that, after the measurements, the qubits $\beta_i$ are only connected through edges of the matching. Any other edge would either contradict the fact that the graph is bipartite with respect to $M \vert \comp{M}$ or the condition that qubits in $\BB$ have no neighbor in common. As seen in Fig.~\ref{fig:belltrafo1}, some qubits $\beta_i$ are in $M$, some are in $\comp{M}$. This distinction, however, is of no importance in this proof. Also, there might be other, isolated qubits. These are not shown in Fig.~\ref{fig:belltrafo1}, since they do not play any role in the proof.

\begin{figure}
\includegraphics[width=0.8\columnwidth]{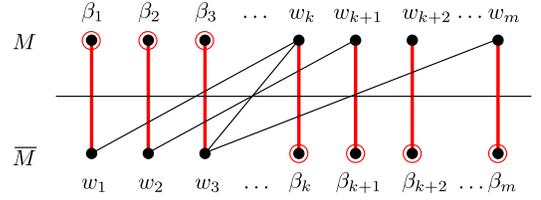}
\caption{\label{fig:belltrafo1} After measuring out all qubits that are not needed for the creation of Bell pairs, one obtains a graph as the one shown. Edges of the matching are indicated by red, thick lines, while other edges are shown in black and with thin lines.}
\end{figure}

Finally, we need to delete all edges that are not in the matching, i.e. the edges $(w_i,w_j)$. Consider an edge, say $(w_1, w_k)$ (cf. Fig.~\ref{fig:belltrafo1}). It can be deleted using the following steps: 

First, connect $\beta_1$ and $w_k$. Such a creation of an edge corresponds to an application of a local unitary to the graph state, namely a controlled-Z gate acting on the two qubits to be connected.

Second, apply a local complementation operation on qubit $\beta_1$. This operation corresponds to a local unitary and inverts the neighborhood graph of $\beta_1$. More precisely, all edges between neighbors of $\beta_1$ are deleted and all neighbors of $\beta_1$ which are not connected become connected \cite{hein}. Since $w_1$ and $w_k$ are the only neighbors of $\beta_1$, this means that the edge $(w_1, w_k)$ is deleted.

Finally, delete the edge $(\beta_1,w_k)$ again. The described steps now have to be repeated for all other edges that do not belong to the matching. After that, one ends up with $m$ pairs of connected qubits which are disconnected from the rest of the graph. These $m$ Bell pairs have a largest Schmidt coefficient of $\sqrt{2}^{-m}$ and the performed LU and LOCC operations cannot have decreased it \cite{nielsen}. Thus, the square of the largest Schmidt coefficient of $\ket{G}$ must be smaller than $2^{-m}$.
\end{proof}

Let us now start with the main part of the proof in which Lemmata \ref{lem:rule1} - \ref{lem:rule4} will be used.

For the sake of brevity, we define $P_{+} =  \sum_{\vec{s}}\: \prod_{i \in \BB} \g{i}{s_i}$ [cf. Eq.~(\ref{eq:decwitarbgraphstates})]. Note that $P_{+}$ is a sum of all projectors onto graph state basis vectors that contain at least two excitations in $\BB$, \ie, two bits $\beta_i$ that equal one. For example, in the case of a linear cluster state (cf. Fig.~\ref{fig:cln}), we can choose $\BB = \lbrace 1, 4, 7, \dots \rbrace$. Then,
\begin{align}
\label{eq:6}
P_{+} = \sum_{\vec{x} \in \lbrace 0,1 \rbrace^{n-b}} & \left( \ketbra{0 x_1 x_2 1 x_3 x_4 1 \dots} \right.\nonumber \\
+ \: &\ketbra{1 x_1 x_2 0 x_3 x_4 1 \dots} \nonumber \\
+ \: &\ketbra{1 x_1 x_2 1 x_3 x_4 0 \dots} \nonumber \\
+ \: &\ketbra{1 x_1 x_2 1 x_3 x_4 1 \dots} \nonumber \\
+ \: & \left. \dots \right).
\end{align}
Note that the following proof is an extension of the proof for linear cluster states in Ref.~\cite{ourpaper}.

{\it Main part of the proof of Lemma~\ref{lem:witarbgraphstates} --- } In order to prove that $W_G$ is a fully decomposable witness, we have to show that, for every strict subset $M$, there exists a positive operator $P_M$ such that 
\be
\label{eq:10}
Q_M = \left(W_G-P_M\right)^{T_M} \geq 0 \:.
\ee
We proceed in two steps. First, for a given $M$, we transform our problem for the graph state $\ket{G}$ into a problem for another graph state $\ket{G'}$ in which some edges have been deleted by local operations. Second, in the main part of the proof, we provide an algorithm for a given $M$ to construct a positive operator $P_M$ that obeys Eq.~(\ref{eq:10}).

{\it First step: Transformation of the graph state ---} The goal of the first step is to transform the graph $G$ to a graph $G'$ by deleting all edges that connect qubits in the same partition. A graph, in which the vertices can be divided into two subsets $M$ and $\comp{M}$ such that only vertices of different subsets are connected with each other, is called {\it bipartite}. As we will see later, this property will be useful, since it allows us to make use of Lemma \ref{lem:rule3}.

We start by noting that any operator $O$ that is diagonal in a graph basis can be written in the form 
\be
O= \sum_{\vec{x}} c_{\vec{x}} \prod_{i=1}^{n} g_i^{x_i} \:,
\ee
where the sum runs over the set of binary vectors $\vec{x} \in \lbrace 0,1 \rbrace^n$. Moreover, $c_{\vec{x}}$ are coefficients that depend on the operator $O$. Since any partial transposition can at most introduce minus signs in some terms of this sum, such operators remain diagonal under any partial transposition.

As both $W_G$ of Eq.~(\ref{eq:decwitarbgraphstates}) is graph-diagonal and we restrict ourselves to operators $P_M$ which are also graph-diagonal, it is enough to prove that 
\be
\label{eq:diagonal_positivity}
\bra{\vec{k}}\left(W_{G} - P_M\right)^{T_M}\ket{\vec{k}} \geq 0
\ee
holds for all $M$ and all graph state basis vectors $\ket{\vec{k}}$.

Now, we perform the graph transformation $G \mapsto G'$ by deleting all edges that connect qubits in the same partition. This corresponds to applying a controlled-Z operation $C_{j,l}$ to all such pairs of qubits $j,l$. Altogether, this results in a unitary $A = \prod_{(j,l)} C_{j,l}$ that acts on $M$, where the product runs over all edges $(j,l)$ that connect qubits in $M$, and an analogous unitary $B = \prod_{(j,l)} C_{j,l}$, where the product includes edges in $\comp{M}$ and which acts on $\comp{M}$.

Since the controlled-Z operation is real and diagonal, we have

\be
\label{eq:89}
A = A^{\ast} = A^{\dagger} = A^{T}
\ee
and analogously for $B$.

Together with the unitarity of $A$ and $B$, these equalities imply the equivalence 

\begin{align}
\label{eq:12}
& \bra{\vec{k}} \left(W_G-P_M\right)^{T_M} \ket{\vec{k}} \geq 0\\
\Leftrightarrow \: &_{G'}\bra{\vec{k}}  A \otimes B \left(W_G-P_M\right)^{T_M} A^{\dagger} \otimes B^{\dagger}\ket{\vec{k}}_{G'} \geq 0 \\
\Leftrightarrow \: &_{G'}\bra{\vec{k}}  A \otimes B^{\ast} \left(W_G-P_M\right)^{T_M} A^{\dagger} \otimes B^{T}\ket{\vec{k}}_{G'} \geq 0 \\
\Leftrightarrow \: &_{G'}\bra{\vec{k}} \left[ A \otimes B \left(W_G-P_M\right) A^{\dagger} \otimes B^{\dagger} \right]^{T_M} \ket{\vec{k}}_{G'} \geq 0\\
\label{eq:11}
\Leftrightarrow \: &_{G'}\bra{\vec{k}} \left(W_{G'}-P_M'\right)^{T_M} \ket{\vec{k}}_{G'} \geq 0
\end{align}
where $\ket{\vec{k}}_{G'} = A \otimes B \: \ket{\vec{k}}$ are the basis vectors that are associated to the transformed generators ${g_i}' = (A \otimes B) g_i (A^{\dagger} \otimes B^{\dagger})$. Also, the transformed witness is given by $W_{G'} = (A \otimes B) W_G (A^{\dagger} \otimes B^{\dagger}) =\frac{1}{2}\eins - \ketbra{G'} - \frac{1}{2}\sum_{\vec{k}^2 > 1}\: \prod_{i=1}^{\vert \BB \vert} \frac{\eins + (-1)^{k_{i}} g_{\beta_i}'}{2}$.

Thus, the transformed Eq.~(\ref{eq:11}) has the same form as Eq.~(\ref{eq:diagonal_positivity}). Keep in mind that one needs to prove Eq.~(\ref{eq:11}) for all subsets $M$ and all basis vectors $\ket{\vec{k}}_{G'}$. 

For better readability, we drop the subscript $G'$ of the graph basis vectors: $\ket{\vec{k}}_{G'} \mapsto \ket{\vec{k}}$. Every state in the remainder of this proof is to be understood in the graph basis of graph $G'$. 

Finally, we note that the most important thing to keep in mind from this step is that the graph $G'$ is bipartite with respect to the two sets $M$ and $\comp{M}$.

{\it Second step: Algorithm to construct $P_M'$ ---} Let us now provide an algorithm to construct $P_M'$ for any given $M$. Note that we order the qubits $\beta_i$ in a canonical way such that $\beta_i < \beta_{i+1}$. 

\begin{enumerate}
\item Start with $P^{(0)}_M = \ketbra{G'} = \ketbra{0 \dots 0}$.
\item Set $i=1$.
\item If $\beta_i$ has no neighbors (in graph $G'$), set $P^{(i)}_M = P^{(i-1)}_M$. If $\beta_i$ has neighbors, define $P^{(i)}_M$ as $P^{(i)}_M = P^{(i-1)}_M+ \left(\prod_{j \in \NN(\beta_i)}Z_j\right) P^{(i-1)}_M \left(\prod_{j \in \NN(\beta_i)}Z_j\right)$.
\item If $i \leq b$, increase $i$ by one and repeat step 3. Otherwise, proceed with step 5.
\item Let $r$ be the number of qubits in $\BB$ that have neighbors (in graph $G'$), \ie, the number of steps in which $P^{(i)}_M$ changed.\\
If $r \leq 1$, define
\be
\label{eq:P_def1}
P_M' = 0 \:.
\ee
Let $t$ be the value of $i$ for which $P^{(i)}_M$ was changed the last time, \ie, $P^{(i)}_M = P^{(t)}_M \: \forall \: i> t$. If $r > 1$, define 
\be
\label{eq:P_def}
P_M' = P^{(t-1)}_M -\ketbra{G'} \:.
\ee
\end{enumerate}

Note that the operator $P_M'$ constructed via the given algorithm is either zero or a sum of one-dimensional projectors onto basis states, \ie, 

\be
\label{eq:formofP}
P_M' = \sum_{\vec{a}} \ketbra{\vec{a}} \:.
\ee

This can be seen by the fact that $P_M^{(0)} = \ketbra{G'} = \ketbra{0 \dots 0}$, the application of $Z$ only flips a bit and finally $\ketbra{G'}$ is subtracted again. Let us illustrate the algorithm by a concrete example.

\textbf{Example of the algorithm:} Consider state No.~16 of Table \ref{tab:graphstates} and the bipartition given by $M = \lbrace 1,2,5,6\rbrace$. Then, the transformation in the first step of the proof deletes the edges $(1,2)$ and $(3,4)$, since $1,2 \in M$ and $3,4 \in \comp{M}$.

Let us choose set $\BB = \lbrace 1,5,6\rbrace$. Thus, the algorithm produces the following operators. From step 1, we have

\be
P_M^{(0)} = \ketbra{000000} \: .
\ee
As qubit 1 does not have any neighbors, since edge $(1,2)$ has been deleted, step 2 does not change the operator $P_M^{(0)}$ and therefore results in

\be
P_M^{(1)} = \ketbra{000000} \: .
\ee
Then, the loop in step 3 produces

\begin{align}
P_M^{(2)} =\: & \ketbra{000000} + \ketbra{000100}\:, \\
P_M^{(3)} =\: & \ketbra{000000} + \ketbra{000100} \nonumber \\
& +\ketbra{001000} + \ketbra{001100} \:.
\end{align}

$P_M^{(i)}$ was changed in two steps or, in other words, two qubits in graph $G'$ which are also in $\BB$, namely qubits 5 and 6, have a neighbor. Thus, $r=2$. Moreover, as $P_M^{(i)}$ was changed in the third step, we have $t=3$ and therefore

\begin{align}
P_M' = \: & P_M^{(2)} - \ketbra{000000} \nonumber \\
= \: & \ketbra{000100} \: .
\end{align}
Therefore, in this example, the sum in Eq.~(\ref{eq:formofP}) has only one term.

Let us now return to the general case and understand the properties of the operator $P_M'$ for an arbitrary $M$. The construction uses Lemma \ref{lem:rule3} to ensure that, in every step, either

\begin{subequations}
\label{eq:P_invariance1}
\be
\label{eq:P_invariance1a}
\left(P^{(i)}_M\right)^{T_{M_{i}}} = \left(P^{(i)}_M\right)^{T_{\NNN(\beta_i)}}
\ee
or
\be
\label{eq:P_invariance1b}
\left(P^{(i)}_M\right)^{T_{M_{i}}} = P^{(i)}_M
\ee
\end{subequations}
hold, where we defined $M_k = M \cap \NNN(\beta_k)$. Therefore, as we will se later, the qubits $\beta_i$ can be treated as if they had no neighbor in the opposite partition.

To see that Eqs.~(\ref{eq:P_invariance1}) hold, assume that $\beta_i \in M$. Since qubits that were neighbors of $\beta_i$ in graph $G$ and were also in $M$ are not connected to $\beta_i$ in graph $G'$ anymore, we know that $\NN(\beta_i) \subseteq \comp{M}$. Then, the given algorithm sets

\be
P_M^{(i)} = \sum_{\vec{c}} \left[ \ketbra{\vec{c}} +  \left(\prod_{j \in \NN(\beta_i)}Z_j\right) \ketbra{\vec{c}} \left(\prod_{j \in \NN(\beta_i)}Z_j\right)\right] \: .
\ee

This expression is invariant under the partial transposition $T_{\beta_i}$ due to Lemma~\ref{lem:rule3}. Therefore, Eq.~(\ref{eq:P_invariance1b}) holds. Similarly, in the case $\beta_i \in \comp{M}$, Eq.~(\ref{eq:P_invariance1a}) holds.

Eqs.~(\ref{eq:P_invariance1}) hold in every step, \ie, for $i = j$ and for $i=k$, where $j \neq k$. According to the premise of non-overlapping neighborhoods of the qubits in $\BB$, we have $\NNN(\beta_j) \cap \NNN(\beta_k) = \lbrace \rbrace$. Therefore, the partial transpositions in Eqs.~(\ref{eq:P_invariance1}) for $i=j$ always affect qubits different from the ones that are affected by the partial transpositions for $i=k$. For this reason, Eqs.~(\ref{eq:P_invariance1}) for $P_M^{(t-1)}$ hold with respect to every value of $k$, except for $k=t$. More precisely,

\begin{subequations}
\label{eq:P_invariance}
\be
\label{eq:P_invariance2a}
\left(P_M^{(t-1)}\right)^{T_{M_{k}}} = \left(P_M^{(t-1)}\right)^{T_{\NNN(\beta_k)}} 
\ee
or
\be
\label{eq:P_invariance2b}
\left(P_M^{(t-1)}\right)^{T_{M_{k}}} =P_M^{(t-1)} 
\ee
\end{subequations}
is true for every $k \neq t$.
We will use this important property later.

Let us proceed with the proof. Since $P_M'$ is zero or has the form of Eq.~(\ref{eq:formofP}), we know that $P_M' \geq 0$. Thus, it remains to show that Eq.~(\ref{eq:11}) holds.

Note that the transformed operator $P_{+}' =  \left(A \otimes B \right) P_{+} \left(A \otimes B\right)$ is invariant under any partial transposition. This can be seen using Eq.~(\ref{eq:89}) and the fact that $P_{+}$ is invariant under any partial transposition. $P_{+} = \sum_{\vec{s}}\: \prod_{i \in \BB} \g{i}{s_i}$ is invariant, since it only contains generators of qubits that have no neighbor in common and are no neighbors of each other. Thus, the form of the generators as given in Eq.~(\ref{eq:generators}) implies that $P_{+}$ does not contain any $Y$ operators which are the only Pauli matrices that change under transposition.

Together with the explicit form of the witness given in Eq.~(\ref{eq:decwitarbgraphstates}), we can therefore rewrite Eq.~(\ref{eq:11}) as

\be
\label{eq:diagonal_positivity2}
\frac{1}{2} - \frac{1}{2}  \bra{\vec{k}} P_{+}' \ket{\vec{k}}- \bra{\vec{k}}\left(\ketbra{G'} + P_M'\right)^{T_M}\ket{\vec{k}} \geq 0\:.
\ee
In order to prove this, we distinguish two different cases:
\begin{enumerate}
\item $\boxed{\bra{\vec{k}} P_{+}' \ket{\vec{k}} \neq 0 \Leftrightarrow \bra{\vec{k}} P_{+}' \ket{\vec{k}} = 1}$\\
\vskip 0.1cm
Note that this equivalence is due to the form of $P_{+}'$ as shown in Eq.~(\ref{eq:6}). Also, this form implies that, in the vectors $\ket{\vec{k}}$ with non-zero overlap, there must be at least two qubits $i_0,j_0 \in \BB$, with $i_0 \neq j_0,$ such that $k_{i_0} = k_{j_0} = 1$.

In the case $P_M' = 0$, Eq.~(\ref{eq:diagonal_positivity2}) and $\bra{\vec{k}} P_{+}' \ket{\vec{k}} = 1$ are equivalent to

\be
\label{eq:0}
- \bra{\vec{k}}\left(\ketbra{G'}\right)^{T_M}\ket{\vec{k}}  \geq 0 \:.
\ee

To see that the left-hand side always vanishes for all $M$ and all $\ket{\vec{k}}$, one uses that $P_M' = 0$ is equivalent to $r \leq 1$, \ie, $\NNN(\beta_i) \subseteq M$ or $\NNN(\beta_i) \subseteq \comp{M}$ holds for all qubits $\beta_i \in \BB$ with at most one exception, namely $\beta_t$. With $k_{i_0} = k_{j_0} = 1$, Lemma 1 can be applied to see that the left-hand side of Eq.~(\ref{eq:0}) vanishes.

In the case $P_M' \neq 0$, Eq.~(\ref{eq:diagonal_positivity2}) can be simplified using $\bra{\vec{k}} P_{+}' \ket{\vec{k}} = 1$ to

\begin{align}
\label{eq:1}
& - \bra{\vec{k}}\left(\ketbra{G'} + P_M'\right)^{T_M}\ket{\vec{k}}  \geq 0\nonumber \\
\Leftrightarrow & \: - \bra{\vec{k}}\left(P^{(t-1)}_M\right)^{T_M}\ket{\vec{k}}\geq 0\:.
\end{align}

Here, the definition of $P_M'$, Eq.~(\ref{eq:P_def}), has been used.

Now, $P_M'$ and therefore $P_M^{(t-1)}$ consists of a sum of projectors onto graph basis states $\ket{\vec{a}}$ [see Eq.~(\ref{eq:formofP})]. Since the algorithm starts with $P_M^{(0)} = \ketbra{0 \dots 0}$ and never flips any bits on the qubits $\beta_i \in \BB$, these states $\ket{\vec{a}}$ obey $a_{\beta_i} = 0, \; \forall \: i = 1, \dots, b$. Also, depending on whether $i_0 = t$ or $j_0 = t$, Eqs.~(\ref{eq:P_invariance}) can be applied to whichever of these two qubits is different from $t$. Let us assume that $i_0 \neq t$. Then, one can use Eq.~(\ref{eq:P_invariance2a}) or (\ref{eq:P_invariance2b}) to replace $M$ by a slightly modified subset $M'$ with $\NNN(\beta_{i_0}) \subseteq M'$ or $\NNN(\beta_{i_0}) \subseteq \comp{M'}$, respectively. Thus, Lemma \ref{lem:rule1} applied to $i_0$ yields

\begin{align}
\label{eq:49}
\bra{\vec{k}}\left(P^{(t-1)}_M\right)^{T_M}\ket{\vec{k}} &= \bra{\vec{k}}\left(P^{(t-1)}_M\right)^{T_{M'}}\ket{\vec{k}} \nonumber \\
 &= 0
\end{align}

and therefore Eq.~(\ref{eq:1}) holds.

\item $\boxed{\bra{\vec{k}} P_{+}' \ket{\vec{k}} = 0}$\\
\vskip 0.1cm
To show that Eq.~(\ref{eq:diagonal_positivity2}) holds, we need to prove that 

\be
\label{eq:7}
\bra{\vec{k}}\left(\ketbra{G'} + P_M' \right)^{T_M}\ket{\vec{k}} \leq  \frac{1}{2} \: .
\ee

In the case $P_M' \neq 0$, $P_M'$ is given by Eq.~(\ref{eq:formofP}) and Eq.~(\ref{eq:7}) is equivalent to 

\be
\label{eq:8}
\bra{\vec{k}}\left(\ketbra{G'} + \sum_{\vec{a}} \ketbra{\vec{a}} \right)^{T_M}\ket{\vec{k}} \leq  \frac{1}{2} \: .
\ee

Note that $\ketbra{G'} + \sum_{\vec{a}} \ketbra{\vec{a}} = P_M^{(t-1)}$ consists of $2^{r-1}$ terms, as one starts with one term and doubles this number $(r-1)$ times to obtain $P_M^{(t-1)}$. Therefore, it is enough to prove the upper bounds

\be
\bra{\vec{k}}\left(\ketbra{G'}\right)^{T_M}\ket{\vec{k}} \leq  2^{-r}
\ee

and

\be
\bra{\vec{k}}\left(\ketbra{\vec{a}}\right)^{T_M}\ket{\vec{k}} \leq  2^{-r} \: \forall \: \ket{\vec{a}} \: .
\ee

We will show these bounds using Lemma \ref{lem:rule2}. However, since the vectors $\ket{\vec{a}}$ are basis vectors of the graph state basis of $\ket{G'}$, $\ket{\vec{a}}$ and $\ket{G'}$ are LU-equivalent. Therefore, they have the same Schmidt coefficients and Lemma \ref{lem:rule2} results in the same upper bounds. For this reason, it suffices to show only one of these upper bounds, namely

\be
\label{eq:9}
\bra{\vec{k}}\left(\ketbra{G'}\right)^{T_M}\ket{\vec{k}} \leq  2^{-r} \:.
\ee

In order to apply Lemma~\ref{lem:rule2}, we need the largest Schmidt coefficient of $\ket{G'}$. According to Lemma~\ref{lem:rule4}, the largest Schmidt coefficient is smaller than (or equal to) $2^{-r}$, since $r$ is the number of qubits in $\BB$ that have at least one neighbor. Note that the conditions of Lemma~\ref{lem:rule4} are met, since $G'$ is a bipartite graph. Thus, Eq.~(\ref{eq:9}) holds.

In the case $P_M' = 0$, we need to show that 

\be
\label{eq:3}
\bra{\vec{k}}\left(\ketbra{G'}\right)^{T_M}\ket{\vec{k}} \leq \frac{1}{2} \:.
\ee

Since $M \neq \lbrace 1, \dots, n\rbrace$, there is at least one Bell pair in $G$ that connects a qubit in $M$ and a qubit in $\comp{M}$. Since the transformation $G \rightarrow G'$ only deletes connections between qubits in the same partition, this pair is also connected in graph $G'$. Then, however, deleting all edges besides the one of this pair by measuring all other qubits leads to one Bell pair. One Bell pair with Schmidt coefficients $\lbrace \frac{1}{\sqrt{2}}, \frac{1}{\sqrt{2}} \rbrace$ is enough to show that Eq.~(\ref{eq:3}) holds (using Lemma~\ref{lem:rule2}).

This finishes the proof of Lemma~\ref{lem:witarbgraphstates}.
\end{enumerate}
\end{proof}

\subsection{White noise tolerance of fully decomposable witnesses (Corollary~\ref{cor:whitenoise})}
\label{sec:prooftol}

\begin{proof} The definition of the white noise tolerance $p_{\rm tol}$ for state $\ket{G}$ and witness $W$ implies that
\be
\label{eq:tolform}
p_{\rm tol} = \left[ 1- \frac{\trace(W)}{2^n \bra{G} W \ket{G}} \right]^{-1} \:.
\ee
Since $\bra{G} W \ket{G} = - 1/2$, it remains to calculate

\begin{align}
\trace(W) = & \:2^{n-1} - 1 - \frac{1}{2} 2^{n-\vert \BB \vert}\sum_{j=2} \binom{\vert \BB \vert}{j} \nonumber \\
=&\: 2^{n-1} - 1 - 2^{n-\vert \BB \vert -1} \left( 2^{\vert \BB \vert} - \vert \BB \vert - 1\right)
\end{align}
Together with Eq.~(\ref{eq:tolform}), this results in Eq.~(\ref{eq:generalnoise}).
\end{proof}

\subsection{Extended construction of fully decomposable witnesses (Lemma~\ref{lem:witcomb})}
\label{sec:proofextconstr}

Before we begin with the proof of Lemma~\ref{lem:witcomb}, let us state the following lemma, which we will need later in this proof and also in Sec.~\ref{sec:PPTproofextconstr}.

\begin{lemma}
\label{lem:ptinvariance}
Given a graph state $\ket{G}$ of $n$ qubits, the associated generators $g_i$ and the projectors $\g{i}{\pm} = \left(\eins \pm g_i\right)/2$. Let $\BB$ be any subset of all qubits in which no two qubits are neighbors of each other. Let $\BB_i$ for $ i = 1, \dots, m$, be some arbitrary subsets of $\BB$ and $P_i$ for $ i = 1, \dots, m$, some operators that can be written as

\be
\label{eq:69}
P_i = \sum_{\vec{s}} \alpha_{\vec{s}} \prod_{j \in \BB_i} \g{j}{s_j} \:,
\ee
where $\sum_{\vec{s}}$ sums over some subset of $\lbrace -1, +1 \rbrace ^{\vert \BB_i \vert }$, i.e. over vectors of length $\vert \BB_i \vert$ with elements $\pm 1$, and $\alpha_{\vec{s}}$ are some coefficients. Then, the operator

\be
\label{eq:70}
\max_{i=1,\dots, m} P_i = \sum_{\vec{k}} \ketbra{\vec{k}} \max_{i=1,\dots, m} \bra{\vec{k}} P_i \ket{\vec{k}}
\ee
is invariant under any partial transposition.
\end{lemma}

\begin{proof}
We prove the invariance by showing that $\max_{i=1,\dots, m} P_i$ can be written as a linear combination of operators

\be
\label{eq:71}
T_{\vec{s}} = \prod_{j \in \BB} \g{j}{s_j} \:,
\ee
where $\vec{s} \in \lbrace -1,+1\rbrace^{\vert \BB \vert}$. $T_{\vec{s}}$ is graph-diagonal and

\be
\label{eq:72}
\bra{\vec{k}} T_{\vec{s}} \ket{\vec{k}} =
\begin{cases}
& 1, \: {\rm if} \: (-1)^{k_j} = s_j \:\mbox{for all} \: j \in \BB\\
& 0, \: {\rm otherwise}
\end{cases}
\: .
\ee
Now, we note that 

\be
\label{eq:73}
\bra{\vec{k}} P_i \ket{\vec{k}}  = \bra{\vec{l}} P_i \ket{\vec{l}}, \: {\rm if} \: k_j = l_j \: \mbox{for all} \: j \in \BB \: .
\ee
This follows from the fact that, if $\ket{\vec{k}}$ and $\ket{\vec{l}}$ have the same bit values on all qubits in $\BB$, it is possible to obtain $\ket{\vec{k}}$ from $\ket{\vec{l}}$ by applying operators $Z_j$ on qubits $j \notin \BB$. Since $P_i$ only has $Z$-operators (or $\eins$) on these qubits, it commutes with $Z_j$, $j \notin \BB$, and one has

\be
\bra{\vec{k}} P_i \ket{\vec{k}} = \bra{\vec{l}} ( \prod_{j} Z_j ) P_i ( \prod_{j} Z_j ) \ket{\vec{l}}  = \bra{\vec{l}} P_i \ket{\vec{l}} \: .
\ee

Equation~(\ref{eq:73}) implies that $\bra{\vec{k}} \max P_i \ket{\vec{k}}$ only depends on the bit values $k_j$ with $j \in \BB$. We can therefore set $\alpha_{\vec{s}} = \bra{\vec{k}} \max P_i \ket{\vec{k}}$, where $\vec{s} \in \lbrace -1,+1\rbrace^{\vert \BB \vert}$ and $s_j = (-1)^{k_j}$ for all $j \in \BB$. Then, we have

\be
\label{eq:75}
\max_{i=1,\dots, m} P_i = \sum_{\vec{s}} \alpha_{\vec{s}} T_{\vec{s}} \:.
\ee

The operators $T_{\vec{s}}$ are invariant under partial transposition, since $\BB$ only consists of qubits that are not neighbors of each other [cf. Eq.~(\ref{eq:71})]. Thus, $\max_{i=1,\dots, m} P_i$ is invariant under any partial transposition.
\end{proof}

Let us now come to the main part of the proof of Lemma~\ref{lem:witcomb}.

\begin{proof} 
We write the given fully decomposable witnesses $W_i$ in the form

\be
\label{eq:49b}
W_i = \frac{1}{2} \eins - \ketbra{G} - \frac{1}{2} P_{+}^{(i)} \:,
\ee
where

\be
\label{eq:50}
P_{+}^{(i)} = \sum_{\vec{s}}\: \prod_{j \in \BB_i} \g{j}{s_j}\:.
\ee
If we introduce the shorthand notation

\be
\label{eq:21}
\max_{i} P_{+}^{(i)} = \sum_{\vec{k} \in \lbrace 0,1 \rbrace^n} \ketbra{\vec{k}} \max_{i=1,\dots, m} \bra{\vec{k}} P_{+}^{(i)} \ket{\vec{k}}\:,
\ee
we can write the operator of Eq.~(\ref{eq:combwit}) as

\be
\label{eq:22}
W = \frac{1}{2} \eins - \ketbra{G} - \frac{1}{2} \max_{i} P_{+}^{(i)} \:.
\ee

We need to prove that this is indeed a fully decomposable witness. In order to do so, we proceed in two steps. First, we prove that there is a positive operator $P_M$ for every $M$ which is independent from $i$ such that $\left( W_i - P_M \right)^{T_M} \geq 0$ holds for all $i$. Second, we use the positive operators $P_M$ of the first step to prove that the operator of Eq.~(\ref{eq:22}) is a fully decomposable witness.

{\it First step ---} Let us show that, for a given $M$, there exists a positive operator $P_M$ independent from $i$ that obeys

\be
\label{eq:86}
\left( W_i - P_M\right)^{T_M} \geq 0 
\ee
for all $i$. In order to prove this, we apply the algorithm for the construction of such operators given in the proof of Lemma~\ref{lem:witarbgraphstates}. However, instead of applying it to any of the sets $\BB_i$ directly, we construct a new set $\AAA$ out of these sets $\BB_i$. Although the set $\AAA$ will contain at least as many qubits as the largest one of the sets $\BB_i$, in most cases even more qubits, it still obeys the condition that no two qubits in $\AAA$ have a neighbor in common. Therefore, we can then apply the algorithm to it.

First, we assume that for any qubit $\beta_j^{(i)}$ from any subset $\BB_i$, there is, in every other set $\BB_k$ a qubit $\beta_l^{(k)}$ that has the same neighborhood as $\beta_j^{(i)}$. In principle, according to condition (ii) of Lemma~\ref{lem:witcomb}, there can also be subsets $\BB_k$ in which no qubit has the same neighborhood as $\beta_j^{(i)}$. However, adding the qubit $\beta_j^{(i)}$ itself to such a subset $\BB_k$ causes the mentioned assumption to hold. After this addition, $\BB_k$ still fulfills condition (i), since there was no qubit in $\BB_k$ that had a neighbor in common with $\beta_j^{(i)}$ before the addition according to condition (ii).

Furthermore, for a more convenient notation, we relabel the qubits $\beta_j^{(i)}$ in such a way that the two qubits $\beta_j^{(i)} \in \BB_i$ and $\beta_j^{(k)} \in \BB_k$ with the same subscript $j$ also have the same neighborhood. According to the assumption in the last paragraph, there is exactly one qubit in $\BB_k$ that has the same neighborhood as $\beta_j^{(i)}$.

Before constructing $\AAA$ and applying the algorithm, we perform the aforementioned transformation $G \rightarrow G'$ (cf. Sec.~\ref{sec:proofarbgraph}) for the given partition $M$ by deleting all edges that connect qubits in the same partition. Note that two qubits that had the same neighborhood in $G$ do not need to have the same neighborhood in $G'$ anymore.

As argued after Eq.~(\ref{eq:diagonal_positivity}), this transformation changes Eq.~(\ref{eq:86}) into

\be
\label{eq:90}
\bra{\vec{k}}\left(W_{i}' - P_M'\right)^{T_M}\ket{\vec{k}} \geq 0 \:,
\ee
which has to be shown for all vectors $\ket{\vec{k}}$ of the basis given by the transformed generators $g_i' = (A \otimes B) g_i (A^{\dagger} \otimes B^{\dagger})$, for all $i$ and for all $M$. Here, $W_i' = (A \otimes B) W_i (A^{\dagger} \otimes B^{\dagger})$ and $P_M' = (A \otimes B) P_M (A^{\dagger} \otimes B^{\dagger})$, where $A = \prod_{(j,l)} C_{j,l}$ and $B = \prod_{(j,l)} C_{j,l}$ are the unitary operators that correspond to the deletion of the edges in $M$ and in $\comp{M}$, respectively. Therefore, $A$ acts on qubits in $M$ and $B$ on qubits in $\comp{M}$ (which is not obvious from our above notation). Also, we have used that $W_i'$ and $P_M'$ are diagonal in the basis given by the vectors $\ket{\vec{k}}$.

Then, we construct a set $\AAA = \lbrace \alpha_i \rbrace$ in the following way.

\begin{enumerate}
\item Start with the empty set $\AAA = \lbrace \rbrace$.
\item Let $j=1$.
\item If all qubits $\beta_j^{(i)}$, $i = 1, \dots, m$, from the $M$ subsets $\BB_i$ are in the same partition, then add $\beta_j^{(1)}$ to the set $\AAA$. Otherwise, there exists a qubit $\beta_j^{(x)}$ that is in the opposite partition as $\beta_j^{(1)}$. Then, add both $\beta_j^{(x)}$ and $\beta_j^{(1)}$ to $\AAA$.
\item Increase $j$ by one. If $j\leq \vert \BB_1 \vert$, repeat the last step. Otherwise, the construction is finished. Note that any other set $\BB_i$ contains the same number of qubits as $\BB_1$.
\end{enumerate}

Step 3 is the crucial one and we note the following points: If all qubits $\beta_j^{(i)}$, $i = 1, \dots, m$, are in the same partition, we add $\beta_{j}^{(1)}$ from $\BB_1$ to $\AAA$. In principle, in this case one can instead add the $j^{\rm th}$ qubit $\beta_j^{(i)}$ from any other set $\BB_i$ to $\AAA$, since all of them have the same neighborhood even after the transformation $G \rightarrow G'$, as they are all in the same partition.

In the other case, there are two qubits $\beta_j^{(1)}$ and $\beta_j^{(x)}$ in opposite partitions. Then, both of them are added to $\AAA$. However, since they are in opposite partitions and had the same neighborhood in graph $G$, they cannot have a neighbor in common after the transformation $G \rightarrow G'$. Such a neighbor in common would have to be in the opposite partition as $\beta_j^{(1)}$, in order to be its neighbor in graph $G'$, but at the same time in the opposite partition as $\beta_j^{(x)}$. This is impossible since $\beta_j^{(1)}$ and $\beta_j^{(x)}$ are in opposite partitions. 

Together with the fact that any two qubits of the set $\BB_1$ do not have a neighbor in common according to the conditions of Lemma~\ref{lem:witcomb}, this shows that no two qubits in $\AAA$ have a neighbor in common (in graph $G'$). Also, there cannot be two qubits in $\AAA$ which are neighbors of each other, since these must have also been neighbors in $G$, which contradicts the conditions of Lemma~\ref{lem:witcomb}.

Now, we use the algorithm presented after Eq.~(\ref{eq:11}) to construct $P_M'$, but we apply it the qubits $\alpha_i$ in the set $\AAA$ instead of the qubits in set $\BB$ as in the original algorithm. Again, $P_M'$ is a sum over projectors onto graph basis states [cf. Eq.~(\ref{eq:formofP})].

Since $P_{+}^{(i)}$ is invariant under partial transposition, Eq.~(\ref{eq:89}) implies that $P_{+}'^{(i)} = (A \otimes B) P_{+}^{(i)}(A^{\dagger} \otimes B^{\dagger})$ is also invariant. Using the explicit form of $W_i'$, we can thus rewrite Eq.~(\ref{eq:90}) as

\be
\label{eq:44}
\frac{1}{2} - \frac{1}{2} \bra{\vec{k}}P_{+}'^{(i)}\ket{\vec{k}}  - \bra{\vec{k}} (\ketbra{G'} + P_M')^{T_M} \ket{\vec{k}} \geq 0 \: .
\ee

For a given $i$ and $M$, we consider two cases for the vectors $\ket{\vec{k}}$. Then, the reasoning is analogous to the two cases in Sec.~\ref{sec:proofarbgraph}.

\begin{enumerate}
\item $\boxed{\bra{\vec{k}} P_{+}'^{(i)} \ket{\vec{k}} = 1}$\\
\vskip 0.1cm

Since $P_{+}'^{(i)}$ has the form of Eq.~(\ref{eq:50}), but with the transformed generators $g_i'$, and is therefore a sum of projectors as in Eq.~(\ref{eq:formofP}), in this case there are two qubits $j$, $l$ $\in \BB_i$ with $k_j = k_l = 1$. Moreover, Eq.~(\ref{eq:44}) reduces to

\be
\label{eq:45}
-\bra{\vec{k}} ( \ketbra{G'} + P_M')^{T_M} \ket{\vec{k}} \geq 0 \:.
\ee

Per construction, $\AAA$ contains qubits $j$ and $l$ or qubits that have the same neighborhood as qubits $j$ and $l$. Therefore, the algorithm constructs an operator $P_M'$ that contains only projectors $\ketbra{\vec{a}}$ that obey $a_j = a_l = 0$. This can be seen in step 3 of the algorithm for the construction of $P_M'$, in which only qubits in the neighborhood of qubits in $\AAA$ are flipped. Note that $\ket{G'} = \ket{0 \dots 0}$ and therefore also here, the $j^{\rm th}$ and the $l^{\rm th}$ bit equal zero. For this reason, Lemma~\ref{lem:rule1} implies that 

\be
\label{eq:46}
-\bra{\vec{k}} ( \ketbra{G'} + P_M')^{T_M} \ket{\vec{k}} = 0
\ee

and that Eqs.~(\ref{eq:45}) and (\ref{eq:44}) hold.

\item $\boxed{\bra{\vec{k}} P_{+}'^{(i)} \ket{\vec{k}} = 0}$\\
\vskip 0.1cm

If $P_M' \neq 0$, it has $2^{r-1}$ terms, where $r$ is the number of qubits $\alpha_i$ that have a neighbor in graph $G'$. Since no two qubits  in $\AAA$ have a neighbor in common, one can invoke Lemmata~\ref{lem:rule2} and \ref{lem:rule4} to show that Eq.~(\ref{eq:44}) holds for $\ket{\vec{k}}$ with $\bra{\vec{k}} P_{+}'^{(i)}\ket{\vec{k}} = 0$. 

In the case $P_M = 0$, there must be at least one pair of qubits in $G'$ which are connected with each other and in opposite partitions. These can be transformed into a Bell pair via LOCC, such that Eq.~(\ref{eq:44}) holds.

\end{enumerate}

Thus, Eq.~(\ref{eq:86}) holds for all $i$ and the constructed operators $P_M = \left( A^{\dagger} \otimes B^{\dagger} \right) P_M' \left( A \otimes B \right)$.

{\it Second step ---} In the second step, we can now use the positive operators $P_M$ constructed in the last step to show that the operator $W$ of Eq.~(\ref{eq:22}) is a fully decomposable witness. In order to do so, we show that, for every $M$, the positive semidefinite operator $P_M$ of the last step fulfills $\left(W - P_M\right)^{T_M} \geq 0$. Since $W$ and the operators $P_M$ are graph-diagonal, it is enough to show the positivity of $\bra{\vec{k}} \left(W - P_M\right)^{T_M} \ket{\vec{k}}$ for all $\ket{\vec{k}}$.

We define
\be
\label{eq:85}
R_i = \max_{j} P_{+}^{(j)} - P_{+}^{(i)} \:.
\ee
Note that $R_i$ is invariant under partial transposition, as $P_{+}^{(i)}$ does not contain generators of neighboring qubits, and is therefore invariant, and $\max_{j} P_{+}^{(j)}$ is invariant according to Lemma ~\ref{lem:ptinvariance}. Moreover, for a given $\ket{\vec{k}}$, let $i_0$ be the value of $i$ that maximizes $\bra{\vec{k}} P_{+}^{(i)} \ket{\vec{k}}$. Then, we have

\begin{align}
&\: \bra{\vec{k}} \left(W - P_M\right)^{T_M} \ket{\vec{k}} \nonumber \\
= & \: \bra{\vec{k}} \left(W_{i_0} - \frac{1}{2} R_{i_0} - P_M\right)^{T_M} \ket{\vec{k}} \nonumber \\
= & \: \bra{\vec{k}} \left(W_{i_0}- P_M\right)^{T_M} \ket{\vec{k}} - \frac{1}{2} \bra{\vec{k}} R_{i_0} \ket{\vec{k}} \nonumber \\
\geq & \: 0 \: .
\end{align}

In the first line, we have employed the definitions in Eqs.~(\ref{eq:49b}), (\ref{eq:22}) and (\ref{eq:85}). In the second line, we have used the invariance of $R_{i}$ under partial transposition. Finally, for the positivity, we used Eq.~(\ref{eq:86}) and

\begin{align}
\bra{\vec{k}} R_{i_0} \ket{\vec{k}} = & \: \bra{\vec{k}} \max_{j} P_{+}^{(j)} \ket{\vec{k}} - \bra{\vec{k}}P_{+}^{(i_0)} \ket{\vec{k}} \nonumber \\
= & \: \bra{\vec{k}} P_{+}^{(i_0)} \ket{\vec{k}} - \bra{\vec{k}}P_{+}^{(i_0)} \ket{\vec{k}} \nonumber \\
= & \: 0 \: .
\end{align}

Thus, the operator $W$ of Eq.~(\ref{eq:combwit}) is a fully decomposable witness.
\end{proof}

\subsection{Fully PPT witnesses for arbitrary graph states (Lemma~\ref{lem:PPTwitarbgraphstates})}
\label{sec:PPTproofarbgraph}

\begin{proof}
The proof that we present here is similar to the proof of Lemma~\ref{lem:witarbgraphstates} (cf. Sec.~\ref{sec:proofarbgraph}). In fact, it is much shorter, since we do not have to provide a construction algorithm for the positive operators $P_M$, as these equal zero for fully PPT witnesses.

Here, we have to prove that
\be
\label{eq:13}
W^{T_M}_G \geq 0
\ee
holds for every strict subset $M$ of the set of all qubits [cf. Eq.~(\ref{eq:10}) in the proof of Lemma~\ref{lem:witarbgraphstates}]. Since $P_{+} =  \sum_{\vec{s}}\: \prod_{i \in \BB} \g{i}{s_i}$ is invariant under partial transposition and $W^{T_M}_G$ is graph-diagonal, one can plug Eq.~(\ref{eq:PPTwitarbgraphstates}) into Eq.~(\ref{eq:13}) to obtain

\be
\label{eq:14}
\frac{1}{2} - \bra{\vec{k}}\left(\ketbra{G}\right)^{T_M}\ket{\vec{k}} - \left(\frac{1}{2} - \frac{1}{2^{m(\vec{k})}}\right)  \bra{\vec{k}} P_{+} \ket{\vec{k}} \geq 0
\ee
which has to hold for all $M$ and all graph state basis vectors $\ket{\vec{k}}$ [cf. Eq.~(\ref{eq:diagonal_positivity2})]. Here, $m(\vec{k})$ denotes the number of ones in the binary vector $\vec{k}$ that are on qubits contained in $\BB$. These correspond to $-1$s in a sign vector $\vec{s}$ (and zeros in $\vec{k}$ correspond to $+1$s in $\vec{s}$). In formulas, $m(\vec{k})=\left(b-\sum_{i \in \BB} (-1)^{k_i} \right)/2$.

As before [cf. Eq.~(\ref{eq:12})], we transform graph $G$ into the graph $G'$ by deleting all edges that connect qubits in the same partition.
As in the proof of Lemma~\ref{lem:witarbgraphstates}, we now distinguish two cases.

\begin{enumerate}
\item $\boxed{\bra{\vec{k}} P_{+}' \ket{\vec{k}} = 0}$\\
\vskip 0.1cm

In this case, Eq.~(\ref{eq:14}) can be rewritten as

\be
\label{eq:15}
\bra{\vec{k}}\left(\ketbra{G'}\right)^{T_M}\ket{\vec{k}} \leq \frac{1}{2} \:.
\ee

This equation holds, as we have already argued after Eq.~(\ref{eq:3}), since there must at least be two neighboring qubits in opposite partitions.

\item $\boxed{\bra{\vec{k}} P_{+}' \ket{\vec{k}} \neq 0 \Leftrightarrow \bra{\vec{k}} P_{+}' \ket{\vec{k}} = 1}$\\
\vskip 0.1cm

Here, we need to prove that

\be
\label{eq:16}
\bra{\vec{k}}\left(\ketbra{G'}\right)^{T_M}\ket{\vec{k}}\leq \frac{1}{2^{m(\vec{k})}} \:,
\ee

where $m(\vec{k})$ is the number of ones in $\vec{k}$ on qubits in $\BB$. If there is a qubit $i \in \BB$ with $k_i = 1$ with only neighbors that are in the same partition as qubit $i$, then Lemma~\ref{lem:rule1} applies and the left-hand side of Eq.~(\ref{eq:16}) vanishes.

In the case in which no qubit $i \in \BB$ with $k_i = 1$ has only neighbors in the same partition, Lemmata~\ref{lem:rule2} and \ref{lem:rule4} imply that

\be
\label{eq:61}
\bra{\vec{k}}\left(\ketbra{G'}\right)^{T_M}\ket{\vec{k}}\leq \frac{1}{2^{b}} \:,
\ee

where $b = \vert \BB \vert$. Since $m(\vec{k}) \leq b$, Eq.~(\ref{eq:16}) holds.
\end{enumerate}
\end{proof}

\subsection{Extended construction of fully PPT witnesses (Lemma~\ref{lem:PPTwitcomb})}
\label{sec:PPTproofextconstr}

\begin{proof} 
Here, we prove that the operator $W$ of Eq.~(\ref{eq:PPTcombwit}) is a fully PPT witness. To this end, we write the given fully PPT witnesses $W_i$ as

\be
\label{eq:64}
W_i = \frac{1}{2} \eins - \ketbra{G} - P_{+}^{(i)} \:, 
\ee
with the definition

\be
\label{eq:65}
P_{+}^{(i)} = \sum_{\vec{s}}\: \left( \frac{1}{2} - \frac{1}{2^{m(\vec{s})}} \right) \: \prod_{j \in \BB_i} \g{j}{s_j} \: .
\ee

Here, $m(\vec{s})$ is the number of elements $s_j = -1$ in $\vec{s}$, i.e., $m(\vec{s})=\left(\vert \BB_i \vert-\sum_{j=1}^{\vert \BB_i \vert} s_j \right)/2$. 

As we did before, we now introduce the shorthand notation

\be
\label{eq:66}
\max_{i} P_{+}^{(i)} = \sum_{\vec{k} } \ketbra{\vec{k}} \max_{i=1,\dots, m} \bra{\vec{k}} P_{+}^{(i)} \ket{\vec{k}}\:,
\ee
and can thus write the operator of Eq.~(\ref{eq:PPTcombwit}) as

\be
\label{eq:67}
W = \frac{1}{2} \eins - \ketbra{G} - \max_{i} P_{+}^{(i)} \:.
\ee

Now, we proceed similarly to the proof of Lemma~\ref{lem:witcomb} (Sec.~\ref{sec:proofextconstr}), but, since fully PPT witnesses have $P_M = 0$, we do not need to construct such operators here. Therefore, the proof in this section is much shorter.

Again, we define
\be
\label{eq:68}
R_i = \max_{j} P_{+}^{(j)} - P_{+}^{(i)} \:,
\ee
which is invariant under any partial transposition due to Lemma~\ref{lem:ptinvariance}. For a given $\ket{\vec{k}}$, let $i_0$ be the value of $i$ that maximizes $\bra{\vec{k}} P_{+}^{(i)} \ket{\vec{k}}$. Then, we have

\begin{align}
\bra{\vec{k}} W ^{T_M} \ket{\vec{k}} = & \: \bra{\vec{k}} \left(W_{i_0} -  R_{i_0}\right)^{T_M} \ket{\vec{k}} \nonumber \\
= & \: \bra{\vec{k}} W_{i_0}^{T_M} \ket{\vec{k}} - \bra{\vec{k}} R_{i_0} \ket{\vec{k}} \nonumber \\
\geq & \: 0 \: .
\end{align}

In the first line, we plugged in the definitions in Eqs.~(\ref{eq:64}), (\ref{eq:67}) and (\ref{eq:68}). In the next step, we used that $R_{i}$ is invariant under partial transposition. Finally, for the positivity, we used that the operators $W_i$ are fully PPT witnesses and that

\begin{align}
\bra{\vec{k}} R_{i_0} \ket{\vec{k}} = & \: \bra{\vec{k}} \max_{j} P_{+}^{(j)} \ket{\vec{k}} - \bra{\vec{k}}P_{+}^{(i_0)} \ket{\vec{k}} \nonumber \\
= & \: \bra{\vec{k}} P_{+}^{(i_0)} \ket{\vec{k}} - \bra{\vec{k}}P_{+}^{(i_0)} \ket{\vec{k}} \nonumber \\
= & \: 0 \: .
\end{align}

Thus, the operator $W$ of Eq.~(\ref{eq:PPTcombwit}) is a fully PPT witness.

\end{proof}

\subsection{Fully PPT witness for the 2D cluster state (Lemma~\ref{lem:PPTwitCl4x4})}
\label{sec:PPTwitCl4x4}

\begin{proof} In order to show that the operator of Eq.~(\ref{eq:PPTwitCl4x4}) is a fully PPT witness, we provide two lemmata first. The first one specifies some conditions, under which the overlap of the partially transposed 2D cluster state with another basis vector vanishes for certain bipartitions. The second lemma provides an upper bound for the largest Schmidt coefficient of the 2D cluster state for bipartitions in which no partition contains less than two qubits. Note that both lemmata hold for 2D cluster states of $n \times n$ qubits with $n > 2$ and it is therefore straightforward to see that the proof presented here also holds for more than 16 qubits.

\begin{lemma}
\label{lem:rule1b}
Given a 2D cluster $\ket{\rm Cl_{n \times n}}$ of $n^2$ qubits. Consider an arbitrary qubit $q$ of these. Let $\ket{\vec{a}}$ be a state of the corresponding graph state basis. If there is a qubit $i \neq q$ with $a_i = 1$ and there is a qubit $j \in \NN(q)$ with $a_j = 0$, then
\be
\label{eq:rule1b}
\bra{\vec{a}} \left( \ketbra{\rm Cl_{n \times n}}\right)^{T_q} \ket{\vec{a}} = 0 \: .
\ee
\end{lemma}
\begin{proof}
First, note that $\ket{\rm Cl_{n \times n}}$ can be written as $\ket{0 \dots 0}$ in its graph basis. Thus, according to Lemma~\ref{lem:rule1}, Eq.~(\ref{eq:rule1b}) holds if $i \notin \NN(q)$, independent of the condition on qubit $j$.\\
Thus, it remains to show Eq.~(\ref{eq:rule1b}) for the case $i \in \NN(q)$. Due to Eq.~(\ref{eq:projectorform}), we can write Eq.~(\ref{eq:rule1b}) as

\begin{align}
\label{eq:51}
&\bra{\vec{a}} \left( \ketbra{0 \dots 0}\right)^{T_q} \ket{\vec{a}} \nonumber \\
= &\: \trace \left\lbrace  \left[\prod_{l}\frac{1}{2} \left( \eins + g_l\right) \right]^{T_q} \prod_{k} \frac{1}{2} \left[ \eins + (-1)^{a_k} g_k\right] \right\rbrace \: .
\end{align}

To simplify this expression, we note that 

\be
\label{eq:52}
\prod_{k} \frac{1}{2} \left[ \eins + (-1)^{a_k} g_k\right] = \sum_{\vec{x}} (-1)^{\vec{a} \vec{x}} \prod_{k} g_k^{x_k} \: ,
\ee
where the sum runs over all binary vectors $\vec{x}$ of length $n^2$.

Moreover, we define a boolean function $f$ that characterizes the action of the partial transposition on products of generators in the following way:

\begin{align}
f: \lbrace 0,1 \rbrace^{n^2} &\rightarrow \lbrace 0,1 \rbrace\\
\vec{x} &\mapsto f(\vec{x}) =
\begin{cases}
0, & \mbox{if} \: \left(\prod \limits_{i=1}^{n} g_i^{x_i}\right)^{T_q} = \prod \limits_{i=1}^{n} g_i^{x_i}\\
1, & \mbox{if} \: \left(\prod \limits_{i=1}^{n} g_i^{x_i}\right)^{T_q} = - \prod \limits_{i=1}^{n} g_i^{x_i}
\end{cases}
\nonumber
\:.
\end{align}
Note that $f$ depends on $q$. With these definitions, we can write

\be
\label{eq:53}
\left(\prod_{l} g_l^{y_l}\right)^{T_q} = (-1)^{f(\vec{y})} \prod_{l} g_l^{y_l} \: .
\ee

Applying Eqs.~(\ref{eq:52}) and (\ref{eq:53}) to simplify Eq.~(\ref{eq:51}) results in

\begin{align}
\label{eq:54}
&\bra{\vec{a}} \left( \ketbra{0 \dots 0}\right)^{T_q} \ket{\vec{a}} \nonumber \\
= &\: 4^{-n^2} \trace \left\lbrace \left[ \sum_{\vec{y}} (-1)^{f(\vec{y})} \left(\prod_{k} g_k^{y_k}\right) \right] \left[\sum_{\vec{x}} (-1)^{\vec{a} \vec{x}} \prod_{l} g_l^{x_l}\right] \right\rbrace \nonumber \\
= & \: 2^{-n^2} \sum_{\vec{x} \in \lbrace 0,1 \rbrace^{n^2}} (-1)^{\vec{a} \vec{x} + f(\vec{x})} \:.
\end{align}

In the last step, we have used $\trace ( \prod_{\vec{k}} g^{y_k} \prod_{\vec{l}} g^{x_l}) = 2^{n} \delta_{\vec{x}, \vec{y}}$ vanishes if $\vec{x} \neq \vec{y}$. Since both generators $g_i$ and $g_j$ have $Z$ operator on qubit $q$, their product $g_i g_j$ acts trivially on qubit $q$. Therefore,

\begin{align}
&f(x_1, \dots ,x_i, \dots, x_j ,\dots, x_{n^2}) \nonumber \\
=\:& f(x_1 ,\dots ,x_i \oplus 1 ,\dots ,x_j \oplus 1, \dots ,x_{n^2}) \: .
\end{align}
Furthermore, since $a_i = 1$, $a_j = 0$, a term in the sum of Eq.~(\ref{eq:54}) with $x_i = 1$, $x_j = 1$ will have the opposite as the same term with $x_i$ and $x_j$ flipped. Also, flipping $x_i = 1$, $x_j = 0$ to $x_i = 0$, $x_j = 1$ changes the sign of the corresponding term. Thus, the sum in Eq.~(\ref{eq:54}) vanishes.

\end{proof}

\begin{lemma}
\label{lem:2dbellpairs} Given a 2D cluster state $\ket{Cl_{n \times n}}$ with periodic boundary conditions, $n > 2$ and a bipartition $M \vert \comp{M}$. Let $\lambda_{i}$ be the Schmidt coefficients of $\ket{Cl_{n \times n}}$ with respect to this bipartition. If $\vert M \vert \geq 2$ and $\vert \comp{M} \vert \geq 2$, then

\be
\label{eq:55}
\max_{i} \lambda_{i}^2 \leq \frac{1}{4} \: .
\ee

\end{lemma}

\begin{proof} To prove this claim, we provide an LOCC protocol for every possible case which results in two disconnected Bell pairs. Since LOCC does not decrease the largest Schmidt coefficient and a single Bell pair has a Schmidt coefficient of $1/\sqrt{2}$, the upper bound of Eq.~(\ref{eq:55}) follows from it.

Due to the assumptions, there are at least two qubits $i,k \in M$ which each must have a neighbor in $\comp{M}$, say $j \in \NN(i)$, $l \in \NN(k)$ and $j \neq l$. Let us now describe how one can create a Bell pair between $i$ and $j$ and one between $k$ and $l$, both of which are disconnected from the rest of the graph.

If $i$ and $k$ can be chosen in such a way that the qubits $i,j,k,l$ are not connected with each other except for the two edges between $i$ and $j$, $k$ and $l$, measuring out all qubits besides $i, j, k$ and $l$ results in the desired two Bell pairs.

Now, consider the case that the qubits $i,j,k,l$ have more connections amongst each other than the two connections that will be used for the Bell pairs. Then, we first delete all edges that connect qubits of the same partition, which is an LU operation. Then, there are four possible situations as shown in Fig.~\ref{fig:2dcltwobell}. Note that edges that connect the four qubits with other qubits are drawn dashed and in gray, since they might have been deleted by the last operation (and are not needed for the protocol anyway). Morever, situations a) and b) are equivalent to a number of other ones that we did not explicitly draw, in which all the qubits $i,j,k,l$ form a one-dimensional chain [and not a square as in c) and d)]. Note that, if $i$ and $k$ are disconnected in d), $j$ and $l$ must also be disconnected, as $i$ and $k$ being in the same partition implies that also $j$ and $l$ are in the same one.

In cases a) and d), simply measuring out all qubits besides $i,j,k,l$ results in the two Bell pairs.

For cases b) and c), the procedure is slightly more complicated. In case b), we create an edge between the qubits $j$ and $l$, which are in the same partition, via an LU operation. Then, local complementation on $l$ deletes the unwanted edge between $j$ and $k$. Afterwards, we can delete the edge between $j$ and $l$ again. Note that the local complementation possibly also creates (or deletes) other edges between neighbors of $l$. Note that, however, it does not delete the important edges between $i$ and $j$, $k$ and $l$. Moreover, our last step is to measure all qubits besides $i,j,k,l$ which also deletes any such edges that might have been created.

Finally, consider case c). Here, it is not enough to simply consider the qubits $i,j,k,l$, since the four-qubit ring cluster that they build (disregarding any connections to other qubits) is LU-equivalent with a single Bell pair. However, a closer look at situation c) shows that it actually implies that there are four qubits as in b) or in d).

Assume that there is any qubit that neighbors any of the qubits $i,j,k$ or $l$ --- say $k$ --- and is in the opposite partition as $k$. Note that this qubit could not be, in the case of a $3 \times 3$-cluster, a neighbor of $j$ or $l$, since these qubits also lie in the opposite partition as $k$ (and edges between qubits that are in the same partition have been deleted). Thus, we have a situation as in b).

Assume now that there is no neighbor of any of the qubits $i,j,k,l$ that is in the opposite partition as the qubit it neighbors. In other words, each of the four qubits has only neighbors in the same partition. Then, there is another pair of qubits which is in opposite partitions, namely a neighbor of $k$ and one of $l$. This means that there we have a situation as in d).

Thus, we always obtain two Bell pairs and the proof is finished. 

\begin{figure}
\includegraphics[width=0.6\columnwidth]{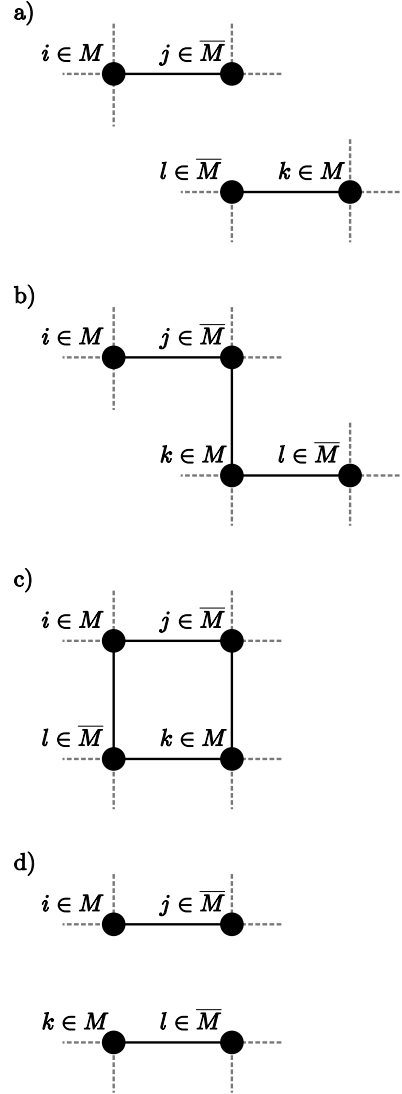}
\caption{\label{fig:2dcltwobell}In a 2D cluster state and for any bipartition $M \vert \comp{M}$, where there are at least two qubits in each partition, one can always obtain two Bell pairs via LOCC operations. Here, we illustrate all possible cases, in which the two Bell pairs are connected to each other in the original 2D cluster state. Note that edges between qubits in the same partition have been deleted. Also, edges that lead to qubits which are not part of the Bell pairs are shown as dashed lines.}
\end{figure}

\end{proof}

Now, we return to the proof of Lemma~\ref{lem:PPTwitCl4x4} which is easy to prove having the last two lemmata in mind. We consider a 2D cluster state of $n \times n$ qubits with $n \geq 3$. As in the proofs before, we write the witness in the form

\be
\label{eq:56}
W_{n \times n} = \frac{1}{2} \eins - \ketbra{Cl_{n \times n}} - \frac{1}{4} P_{+} \:,
\ee
where we defined $P_{+} = \sum_{\vec{k}} \ketbra{\vec{k}} \max_{(i,j)}\: \bra{\vec{k}} D_{(i,j)} \ket{\vec{k}}$ with the operators $D_{(i,j)}$ of Eq.~(\ref{eq:ddef}).

Every generator in $P_{+}$ is neighbored by either two or no other generator. Therefore, $P_{+}$ does not have a $Y$ on any qubit in $M$. Thus, $P_{+}$ is invariant under any partial transposition. Since the witness of Eq.~(\ref{eq:56}) is diagonal in the graph state basis, we need to prove that

\be
\label{eq:57}
\frac{1}{2} - \bra{\vec{k}}\left(\ketbra{Cl_{n \times n}}\right)^{T_M}\ket{\vec{k}} - \frac{1}{4}\bra{\vec{k}}P_{+} \ket{\vec{k}} \geq 0
\ee
holds for all partitions $M$ and all graph basis vectors $\ket{\vec{k}}$.

In the case of a vector $\ket{\vec{k}}$ with $\bra{\vec{k}} P_{+} \ket{\vec{k}} = 0$, we need to show, according to Lemma~\ref{lem:rule2}, that the largest Schmidt coefficient of $\ket{\rm Cl_{n \times n}}$ with respect to $M\vert \comp{M}$ is smaller than (or equal to) $1/\sqrt{2}$. This is trivial, since every connected graph state can be distilled to at least one Bell pair via LOCC operations.

In the case of a vector $\ket{\vec{k}}$ with $\bra{\vec{k}} P_{+} \ket{\vec{k}} = 1$, we have to prove that

\be
\label{eq:58}
\bra{\vec{k}} \left( \ketbra{\rm Cl_{n \times n}} \right)^{T_M} \ket{\vec{k}} \leq \frac{1}{4} \: .
\ee

If $\vert M \vert \geq 2$ and $\vert \comp{M} \vert \geq 2$, we can apply Lemma~\ref{lem:2dbellpairs} (and Lemma~\ref{lem:rule2}) to show this.

Thus, it remains to show Eq.~(\ref{eq:58}) for the case that $\vert M \vert = 1$ or $\vert \comp{M} \vert = 1$. Since $W^{T_M} \geq 0 \Leftrightarrow W^{T_{\comp{M}}} \geq 0$, we can assume w.l.o.g. that $\vert M \vert =1$ and write $M = \lbrace q \rbrace$. Moreover, $\bra{\vec{k}} P_{+} \ket{\vec{k}} = 1$ together with the form of $P_{+}$ [cf. Eq. (\ref{eq:PPTwitCl4x4})] implies that there are two diagonals which we denote by $\DD_{/}^{(x)}$ and $\DD_{\backslash}^{(y)}$ here (cf. Fig.~\ref{fig:2dcases}), on which $\ket{\vec{k}}$ has an odd number of ones. In formulas, 

\begin{figure}
\includegraphics[width=0.9\columnwidth]{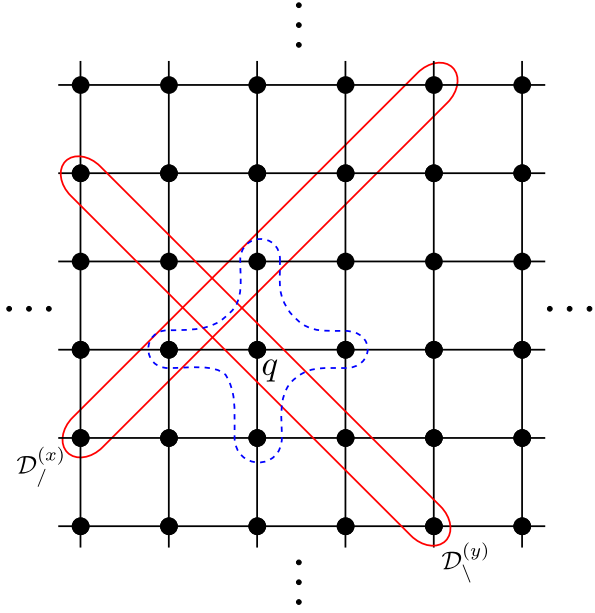}
\caption{\label{fig:2dcases} The proof of the fully PPT witness for an $n \times n$ 2D cluster state considers a one-particle partition $M = \lbrace q \rbrace$ and distinguishes different cases as depicted here. In red, we mark the diagonals $\DD_{/}$ and $\DD_{\backslash}$ mentioned in the text. For more details, see text.}
\end{figure}

\begin{align}
\label{eq:parityprop1}
\prod_{l \in \DD_{/}^{(x)}} g_l \ket{\vec{k}} = \: &- \ket{\vec{k}}, \\
\label{eq:parityprop2}
\prod_{l \in \DD_{\backslash}^{(y)}} g_l \ket{\vec{k}} =  \: & - \ket{\vec{k}} \: .
\end{align}

Note that Lemma~\ref{lem:rule1b} can be used to show that Eq.~(\ref{eq:58}) holds if there exists another qubit $i \neq q$ with $k_i = 1$ and a qubit $j$ in the neighborhood of $q$ with $k_j = 0$. 

Since it is impossible that all qubits $i \neq q$ are zero as this would contradict Eqs.~(\ref{eq:parityprop1}), (\ref{eq:parityprop2}) and the fact that $\DD_{/}^{(x)}$ and $ \DD_{\backslash}^{(y)}$ have non qubit in common, the only other case is that all qubits in the neighborhood of $q$ equal one.

In this case, we only need to consider the five qubits in $\NNN(q) = \NN(q) \cup q$ (marked by a blue, dashed line in Fig.~\ref{fig:2dcases}). Note that the following argumentation is independent from the value of $k_q$ itself. Since $\DD_{/}^{(x)}$ and $ \DD_{\backslash}^{(y)}$ have no qubit in common, it is impossible to choose $q$ in such a way that both the intersection of $\NNN(q)$ with $\DD_{/}^{(x)}$ and the intersection of $\NNN(q)$ with $\DD_{\backslash}^{(y)}$ consist of an odd number of qubits. One of the two intersections always has two or zero qubits. Without loss of generality, we assume that the intersection of $\NNN(q)$ with $\DD_{/}^{(x)}$ has an even number of qubits. An example for this situation is given in Fig.~\ref{fig:2dcases}. Then, due to Eq.~(\ref{eq:parityprop1}), there is a qubit in $\DD_{/}^{(x)}$ which equals one and to which therefore Lemma~\ref{lem:rule1b} can be applied.

This shows that Eq.~(\ref{eq:58}) holds in all cases and finishes the proof.
\end{proof}

\subsection{Entanglement monotone (Lemma~\ref{lem:equalsnegativity})}
\label{sec:proofentmeas}

\begin{proof} 
For the proof that $N(\varrho)$ is indeed an entanglement monotone, we refer to Ref.~\cite{ourpaper}. Moreover, we can rewrite

\begin{align}
\trace (W \vr) = & \: \trace(P_A \vr) + \trace(Q_A^{T_A} \vr ) \nonumber \\
=& \: \trace(P_A \vr) + \trace(Q_A \vr^{T_A} ) \:,
\end{align}
where we used $\trace(C^{T_A} D) = \trace(C D^{T_A})$. This expression is minimized under the constraints $ \eins \geq P_A, Q_A \geq 0$ by letting $P_A = 0$ and $Q_A = \sum_{i} \ketbra{\phi_i}$, where $\ket{\phi_i}$ are the eigenvectors of $\vr^{T_A}$ that correspond to negative eigenvalues. The trace then sums over all negative eigenvalues of $\vr^{T_A}$, which equals the definition of the negativity \cite{negativity}.
\end{proof}

\subsection{Values of the entanglement monotone for graph states (Lemma~\ref{lem:onehalf})}

\label{sec:proofentmeas2}

\begin{proof} {\it (Lemma~\ref{lem:onehalf}) ---}
We define the set of all appropriately normalized witnesses that are decomposable with respect to bipartition $M \vert \comp{M}$ as

\begin{align}
&\mathcal{W}_{M}= \left\{ W \big| \exists\; P, Q \: \mbox{such that} \right.  \nonumber \\
& \qquad \qquad \left. 0 \leq P, Q \leq \mathbbm{1}  \: {\rm and} \: W=P + Q^{T_M} \right\}\!, 
\end{align}
such that the set of all similarly normalized, fully decomposable witnesses $\mathcal{W}$ of Eq.~(\ref{eq:setfullydec}) obeys $\W = \operatorname*{\cap}_{M} \W_M$. Since $\W \subseteq \W_{M_0}$ for any fixed bipartition $M_0 \vert \comp{M_0}$, we have

\be
N(\vr) \leq -\min_{W \in \W_{M_0}} \trace(W \vr) \:.
\ee

According to Lemma~\ref{lem:equalsnegativity}, $\min_{W \in \W_{M_0}} \trace(W \vr)$ equals the negativity with respect to the bipartition $M_0
 \vert \comp{M_0}$. If all particles are qubits, we now choose any bipartition $M_0 \vert \comp{M_0}$ in which $M_0$ only contains one particle, e.g. the bipartition $A \vert B C D \dots$. Then,

\begin{align}
\label{eq:83}
N(\vr) \leq & \: -\min_{W \in \W_{M_0}} \trace(W \vr) \nonumber \\
\leq & \: \max_{\ket{\varphi}} \left(-\min_{W \in \W_{M_0}} \trace(W \ketbra{\varphi})\right) \nonumber \\
= & \: \frac{1}{2} \: .
\end{align}

Here, we use that the expectation value is a linear function and must therefore attain its maximum on a pure state $\ket{\varphi}$. Moreover, the last equality stems from the fact that the negativity with respect to a bipartition $A \vert B C D \dots$, where $A$ is a single qubit, can maximally take on the value one half. This maximum is obtained for the Bell state $\ket{\psi^{+}} =  \left( \ket{00} + \ket{11}\right)/\sqrt{2}$.

If not all particles are qubits, one chooses $M_0$ to consist of a particle that has the smallest dimension of all occuring particles. For example, if $A$ and $B$ are four-level particles and $C$ and $D$ are qutrits, then $M_0 = C$ is a valid choice. In this case, the value of one half in the last line of Eq.~(\ref{eq:83}) must be replaced by $(d_{\rm min}-1)/2$, where $d_{\rm min}$ is the dimension of the particle with lowest dimension. This maximum is obtained for the state $\ket{\psi} = \sum_{i=0}^{d-1} \ket{ii}^{\otimes n}/\sqrt{d_{\rm min}}$. It is the maximal value for the negativity with respect to the given bipartition as can be easily seen using the Schmidt decomposition and the fact that $\ketbra{\psi}^{T_{M_0}}$ has $d_{\rm min} (d_{\rm min}-1) /2$ negative eigenvalues that all equal $-1/d_{\rm min}$.

We now know that the entanglement measure is upper-bounded by one half for states that consist only of qubits. Let us now show that, for graph states, the lower bound is also one half. This is easy to see, since we only have to pick one witness $W_G \in \W$ for the given graph state $\ketbra{G}$. Such a witness is the projector witness

\be
W_G = \frac{1}{2} \eins - \ketbra{G}
\ee
which is even a fully PPT witness. It remains to show that this witness also obeys

\be
\label{eq:93}
\eins \geq W_G^{T_M} = \frac{1}{2} \eins - \left(\ketbra{G}\right)^{T_M} \geq 0 \:.
\ee

Here, positivity follows from the projector witness being a fully PPT witness. For the inequality on the left, we need to show that

\be
\label{eq:94}
\bra{\vec{k}} \left( \ketbra{G}\right)^{T_M} \ket{\vec{k}} \geq - \frac{1}{2}
\ee
holds for all $\ket{\vec{k}}$, since the partial transpose of $\ketbra{G}$ is again graph-diagonal.

In order to prove Eq.~(\ref{eq:94}), we use the Schmidt decomposition $\ket{G} = \sum_{i=1} \lambda_i \ket{\mu_i} \otimes \ket{\nu_i}$ with respect to bipartition $M|\comp{M}$ with positive and real Schmidt coefficients $\lambda_i$. Performing the partial transpose in the basis $\ket{\mu_i} \otimes \ket{\nu_j}$ allows to derive a lower bound on $\bra{\vec{k}} \left( \ketbra{G}\right)^{T_M} \ket{\vec{k}}$ in terms of the Schmidt coefficients in the following way:

\begin{align}
\bra{\vec{k}} \left( \ketbra{G}\right)^{T_M} \ket{\vec{k}} \geq \: & \min_{i \neq j} (- \lambda_i \lambda_j)\nonumber \\
\geq \: & \min_{i} (- \lambda_i \sqrt{1-\lambda_i^2}) \nonumber \\
\geq \: & - \frac{1}{2} \: .
\end{align}
In the second line, we used that, as an entangled state, $\ket{G}$ has at least two non-zero Schmidt coefficients and that the squares of all coefficients must sum up to one. The last line follows from the fact that $0 < \lambda_i < 1$.

Consequently, Eq.~(\ref{eq:93}) holds and $W_G$ lies in $\W$. Therefore,

\be
N(\ketbra{G}) \geq - \trace(W_G \ketbra{G}) = \frac{1}{2} \: .
\ee

Therefore, when considering the entanglement measure of Eq.~(\ref{eq:monotone}), the connected graph states are the maximally entangled states. For them, the measure equals one half.
\end{proof}

\section{Witnesses}
\label{sec:graphstatewit}
Note that all witnesses are presented in their graph state basis. As before, we defined $\g{i}{\pm} = \frac{\eins \pm g_i}{2}$. Since all witnesses are diagonal in the graph basis, we use the shorter notation $\ketbradot{\vec{k}} = \ketbra{\vec{k}}$. Moreover, for some states one can make use of their translational symmetry. In these cases, $\TT(\vec{k})$ denotes all translations of the bit string $\vec{k} = k_1 \dots k_n$. For example, 
\begin{align}
\ketbradot{k_1 k_2 \TT(k_3 k_4 k_5 k_6)} = \: & \ketbradot{k_1 k_2 k_3 k_4 k_5 k_6} \nonumber \\
&+\ketbradot{k_1 k_2 k_6 k_3 k_4 k_5} \nonumber \\
&+ \ketbradot{k_1 k_2 k_5 k_6 k_3 k_4} \nonumber \\
&+\ketbradot{k_1 k_2 k_4 k_5 k_6 k_3}
\end{align}

\noindent \textbf{No. 1, Bell state}
\be
W = \frac{\eins}{2} - \ketbra{G} \nonumber
\ee
\hrule
\vskip 10pt
\noindent \textbf{No. 2, ${\rm GHZ_3}$}
\be
W = \frac{\eins}{2} - \ketbra{G} \nonumber
\ee
\hrule
\vskip 10pt
\noindent \textbf{No. 3, ${\rm GHZ_4}$}
\be
W = \frac{\eins}{2} - \ketbra{G} \nonumber
\ee
\hrule
\vskip 10pt
\noindent \textbf{No. 4, ${\rm Cl_4}$}
\be
W = \frac{\eins}{2} - \ketbra{G} - \frac{1}{2} \g{1}{-} \g{4}{-} \nonumber
\ee
\hrule
\vskip 10pt
\noindent \textbf{No. 5, ${\rm GHZ_5}$}
\be
W = \frac{\eins}{2} - \ketbra{G} \nonumber
\ee
\hrule
\vskip 10pt
\noindent \textbf{No. 6, ${\rm Y_5}$} 
\be
W = \frac{\eins}{2} - \ketbra{G} - \frac{1}{2} \g{1}{-} \g{4}{-} - \frac{1}{2}  \g{1}{+} \g{4}{-} \g{5}{-}
\ee
\hrule
\vskip 10pt
\noindent \textbf{No. 7, ${\rm Cl_5}$}
\begin{align*} 
W = &\frac{\eins}{2} - \ketbra{G} - \frac{1}{2} \g{1}{-} \g{5}{-} \\ 
&- \frac{1}{4} \g{1}{+} \g{2}{-} \g{5}{-} - \frac{1}{4} \g{1}{-} \g{4}{-} \g{5}{+} 
\end{align*}
\hrule
\vskip 10pt
\noindent \textbf{No. 8, ${\rm R_5}$}
\begin{align*}
W = & \: 3 \Big[-  \ketbra{G} +  \ketbradot{\TT(00001)}\Big.\\
&+  \Big. \ketbradot{\TT(00101)}+\ketbradot{\TT(00111)}  \Big]\\
&- \ketbradot{11111} + \ketbradot{\TT(11110)}  \\
&+  \ketbradot{\TT(11010)}+\ketbradot{\TT(11000)} 
\end{align*}
\hrule
\vskip 10pt
\noindent \textbf{No. 9, ${\rm GHZ_6}$}
\be
W = \frac{\eins}{2} - \ketbra{G} \nonumber
\ee
\hrule
\vskip 10pt
\noindent \textbf{No. 10}
\begin{align*} 
W = &\frac{\eins}{2} - \ketbra{G} - \frac{1}{2}\g{1}{-} \g{4}{-}   \\ 
&- \frac{1}{2} \g{1}{+} \g{2}{-} \g{4}{-} - \frac{1}{2} \g{1}{+} \g{2}{+} \g{3}{-} \g{4}{-}
\end{align*}
\hrule
\vskip 10pt
\noindent \textbf{No. 11, ${\rm H_6}$}
\begin{align*} 
W = &\frac{\eins}{2} - \ketbra{G}- \frac{1}{2}\g{1}{-} \g{4}{-}  - \frac{1}{2} \g{1}{+} \g{2}{-} \g{4}{-} \\ 
&- \frac{1}{2}\g{2}{-} \g{3}{-} \g{4}{+} - \frac{1}{2} \g{1}{-} \g{2}{+} \g{3}{-} \g{4}{+} 
\end{align*}
\hrule
\vskip 10pt
\noindent \textbf{No. 12, ${\rm Y_6}$}
\begin{align*}
W = &\frac{\eins}{2} - \ketbra{G}- \frac{1}{2}\g{1}{-} \g{5}{-} - \frac{1}{2} \g{1}{-} \g{4}{-} \g{5}{+}\\
&- \frac{1}{2} \g{1}{+} \g{4}{-} \g{6}{-} - \frac{1}{2}\g{1}{+} \g{4}{+}\g{5}{-} \g{6}{-}
\end{align*}
\hrule
\vskip 10pt
\noindent \textbf{No. 13, ${\rm E_6}$}
\begin{align*} 
W = &\frac{\eins}{2} - \ketbra{G}- \frac{1}{2} \g{1}{-} \g{5}{-} \\ 
&- \frac{1}{2}  \g{1}{-} \g{5}{+} \g{6}{-} - \frac{1}{2}  \g{1}{+} \g{5}{-} \g{6}{-}  \\
&- \frac{1}{4} \g{1}{+} \g{2}{-} \g{5}{-} \g{6}{+}- \frac{1}{4} \g{1}{-} \g{4}{-} \g{5}{+} \g{6}{+}
\end{align*}
\hrule
\vskip 10pt
\noindent \textbf{No. 14, ${\rm Cl_6}$}
\begin{align*}
W = &\frac{\eins}{2} - \ketbra{G}- \frac{1}{2} \g{1}{-} \g{4}{-} \\ 
&- \frac{1}{2} \g{1}{+} \g{3}{-} \g{6}{-} - \frac{1}{2} \g{1}{-} \g{4}{+} \g{6}{-}\\
& - \frac{1}{4} \g{1}{+} \g{2}{-} \g{3}{+} \g{6}{-} - \frac{1}{4} \g{1}{-} \g{4}{+} \g{5}{-} \g{6}{+}\\
& - \frac{1}{4} \ketbradot{011110}
\end{align*}
\hrule
\vskip 10pt
\noindent \textbf{No. 15}
\begin{align*} 
W = &\frac{\eins}{2} - \ketbra{G} - \frac{1}{2}\g{1}{-} \g{2}{-} \g{3}{+} \g{5}{-} - \frac{1}{2} \g{1}{-} \g{2}{-} \g{3}{-} \g{5}{+} \\ 
&- \frac{1}{3} \Big[ \ketbradot{00 \: \TT(0011)} + \ketbradot{01 \: \TT(0011)} \Big. \\
& + \ketbradot{10 \: \TT(0011)} + \ketbradot{010001} + \ketbradot{010010} \\
& + \ketbradot{010101} + \ketbradot{010111} + \ketbradot{011000} \\
& + \ketbradot{011011} + \ketbradot{011101} + \ketbradot{011111} \\
& + \ketbradot{100010} + \ketbradot{100100} + \ketbradot{100101} \\
& + \ketbradot{100111} + \ketbradot{101000} + \ketbradot{101101} \\
& + \ketbradot{101110} + \ketbradot{101111} + \ketbradot{110000} \\
& + \ketbradot{110001} + \ketbradot{110100} + \ketbradot{110101} \\
& + \ketbradot{111010} + \ketbradot{111011} + \ketbradot{111110} \\
& \Big.+ \ketbradot{111111} \Big]
\end{align*}
\hrule
\vskip 10pt
\noindent \textbf{No. 16}
\begin{align*} 
W = \: &\frac{\eins}{2} - \ketbra{G}- \frac{1}{2}\g{1}{-} \g{5}{-} \\ 
&- \frac{1}{2} \g{1}{-} \g{5}{+} \g{6}{-}  - \frac{1}{2} \g{1}{+} \g{5}{-}  \g{6}{-} \\ 
&-\frac{1}{4}  \g{1}{+} \g{2}{+} \g{4}{-} \g{5}{+} \g{6}{-}  -\frac{1}{4}  \g{1}{+} \g{2}{-} \g{4}{+} \g{5}{+} \g{6}{-}  \\
&-\frac{1}{4}  \g{1}{+} \g{2}{+} \g{3}{-} \g{5}{-} \g{6}{+} -\frac{1}{4}  \g{1}{+} \g{2}{-} \g{3}{+} \g{5}{-} \g{6}{+}  \\
&-\frac{1}{4}  \g{1}{-} \g{3}{+} \g{4}{-} \g{5}{+} \g{6}{+}  -\frac{1}{4}  \g{1}{-} \g{3}{-} \g{4}{+} \g{5}{+} \g{6}{+}
\end{align*}
\hrule
\vskip 10pt
\noindent \textbf{No. 17}
\begin{align*}
W = \: &\frac{1}{2} \eins - \ketbra{G} - \frac{1}{2} \left(\g{2}{+} \g{5}{-} +\g{2}{-} \g{5}{+}\right) \g{6}{-}\\
& - \frac{1}{2} \Big[ \ketbradot{001101} +\ketbradot{010110}+\ketbradot{011010} \Big. \\
& \qquad+\ketbradot{011110} +\ketbradot{011111}+\ketbradot{101101}\\
& \qquad+\ketbradot{110110} +\ketbradot{111010}+\ketbradot{111110}\\
& \qquad\Big.+\ketbradot{111111} \Big]\\
& - a  \left(\g{2}{+} \g{5}{+} +\g{2}{-} \g{5}{-}\right)\left(\g{3}{+} \g{4}{-} +\g{3}{-} \g{4}{+}\right) \g{6}{-}\\
& - b \Big[ \ketbradot{000110} + \ketbradot{011000} + \ketbradot{100010} \Big. \\
& \qquad+ \ketbradot{101010} + \ketbradot{101110} + \ketbradot{110000} \\
& \qquad\Big. + \ketbradot{110100} + \ketbradot{111100} \Big] \:, \\
a \approx \: & 0.336, \: b \approx \: 0.163
\end{align*}
\hrule
\vskip 10pt
\noindent \textbf{No. 18, ${\rm R_6}$} 
\begin{align*}
W = \: & \frac{1}{2} \eins - \ketbra{G} \nonumber \\
& - \frac{1}{3} \Big[ \ketbradot{\TT(000011)} + \ketbradot{\TT(001011)} \Big. \nonumber \\
& \qquad \Big. + \ketbradot{\TT(001101)} + \ketbradot{\TT(001111)} \Big] \nonumber \\
& - a \ketbradot{111111} \nonumber \\
& - b \Big[ \ketbradot{\TT(011111)}  + \ketbradot{\TT(010111)} \Big]\\
& - c \Big[ \ketbradot{\TT(001001)}  + \ketbradot{\TT(011011)} \Big. \nonumber \\
& \qquad + \Big. \ketbradot{\TT(010101)}\Big] \:, \\
a \approx \: & 0.455,\: b \approx \: 0.363, \:c \approx \: 0.272
\end{align*}
Note that expressions like $\TT(001001)$ only sum over distinct translations, i.e. $\ketbradot{\TT(011011)} = \ketbradot{011011} +\ketbradot{101101} + \ketbradot{110110}$.
\vskip 10pt
\hrule
\vskip 10pt
\noindent \textbf{No. 19}
\begin{align*}
W =\:&  \frac{1}{2} \eins - \ketbra{G} \nonumber \\
& - \frac{1}{3} \Big[ + \ketbradot{\TT(111110)} + \ketbradot{\TT(000011)} \Big. \nonumber \\
& \qquad + \ketbradot{\TT(000101)} + \ketbradot{\TT(000111)}  \\
& \qquad + \ketbradot{\TT(001001)} + \ketbradot{\TT(011011)} \nonumber \\
& \qquad + \ketbradot{001101} + \ketbradot{010011} + \ketbradot{010110} \nonumber \\
& \qquad + \ketbradot{011010} + \ketbradot{011101} + \ketbradot{011110} \nonumber \\
& \qquad + \ketbradot{100101} + \ketbradot{101001} + \ketbradot{101011} \nonumber \\
& \qquad + \ketbradot{101100} + \ketbradot{101110} + \ketbradot{110010} \nonumber \\
& \qquad \Big. + \ketbradot{110011} + \ketbradot{110101} + \ketbradot{111111} \Big] \\
\end{align*} 
Again, expressions like $\TT(001001)$ only sum over distinct translations.

\bibliographystyle{apsrev}

\end{document}